\documentclass[a4paper,10pt,reqno]{amsart}
\usepackage[latin1]{inputenc}
\usepackage[T1]{fontenc}
\usepackage[english]{babel}
\usepackage{amsmath,amsthm}
\usepackage{amsfonts}
\usepackage{subfigure}
\usepackage{amssymb}
\usepackage{graphicx}
\usepackage{bbm}
\usepackage{color}
\usepackage{wasysym}
\usepackage[colorlinks=true, citecolor=red, linkcolor=blue, urlcolor=black]{hyperref}

\usepackage[foot]{amsaddr}
\usepackage[multiple]{footmisc}
\usepackage{bigfoot}

\newtheorem{theorem}{Theorem}[section]
\newtheorem{proposition}{Proposition}[section]
\newtheorem{lemma}{Lemma}[section]

\newtheorem{definition}{Definition}[section]

\numberwithin{equation}{subsection}
\numberwithin{figure}{subsection}

\makeatletter
\@namedef{subjclassname@2020}{\textup{2020} Mathematics Subject Classification}
\makeatother

\begin{document}

\title[]
{
Convergence to the equilibrium for the kinetic transport equation in the two-dimensional periodic Lorentz Gas
}

\author[F. Pieroni]{Francesca Pieroni$^1$}
\address{$^1$ Dipartimento di Matematica, Sapienza Universit\`a di Roma, 00185 Roma, Italia.}
\email{francesca.pieroni@uniroma1.it}
\thanks{This work has been partially supported by the grant "Progetti di ricerca d'Ateneo 2024" by Sapienza University, Rome and by GNFM - INdAM. Co-funded by the European Union (ERC CoG KiLiM, project number 101125162).}


\begin{abstract} 
We consider the kinetic transport equation that arise in the Boltzmann-Grad limit of the two-dimensional periodic Lorentz Gas. This equation has been obtained by extending the phase space of positions and velocities through the introduction of two new variables, representing respectively the time to the next collision and the corresponding impact parameter. Here we mostly focus on the case of periodic boundary conditions on the positions space: we prove that, under suitable hypothesis, the time evolution of a probability density on the extended phase space converges to the equilibrium state with respect to the $L^p$ norm ($^*$-weakly if $p=\infty$), if such initial density is $L^p$. If $p=2$, or if the initial datum does not depend on the position, we also get more precise estimates about the rate of the approach to the equilibrium. Our proof is based on the analysis of the long time behavior of the Fourier coefficients of the solution. 
\end{abstract}

\maketitle

\tableofcontents

\section{Introduction}
The Lorentz Gas is the dynamical system describing the motion of a classical particle interacting through elastic collisions with a system of obstacles. We will focus here on hard-core scatterers with spherical symmetry, and since the obstacles are infinitely heavy, they do not move. This model was first proposed by Lorentz \cite{Lorentz1905} to describe the motion of the electrons in a metal.

Let the spherical obstacles with centers $c:=\{c_j\}_{j\in\mathbb{N}}$ in $\mathbb{R}^d$ and radius $\varepsilon$. Then, the motion of the particle is described as follows. If no collisions occur, that is, as long as the distance between the particle and any obstacle is larger than $\varepsilon$, the particle's motion is free: given $(x,v)\in\mathbb{R}^d\times\mathbb{R}^d$ such that $|x-c_j|>\varepsilon\quad\forall j\in\mathbb{N}$, and denoting by $(X_t,V_t)$ the position and velocity of the particle at time $t$, the dynamics is described by

{\small\begin{align}\label{motolibero}
\left\{\begin{array}{lcr}X_0(x,v;c)=x, \dot X_t(x,v;c)=V_t(x,v;c),\\V_0(x,v;c)=v,\dot V_t(x,v;c)=0,\end{array}\right.\quad\textrm{ if }|X_t(x,v;c)-c_j|>\varepsilon\quad\forall j\in\mathbb{N},
\end{align}}

but when the particle hits an obstacle, it gets specularly reflected. That is, if the particle collides at time $t$ with the obstacle $\bar j$, then $|X_t(x,v;c)-c_{\bar j}|=\varepsilon$ and 

\begin{align}\label{motourto}
\left\{\begin{array}{lcr}X_{t^+}(x,v;\{c_j\}_j)=X_{t^-}(x,v;c),\\V_{t^+}(x,v;c)={\mathcal R}\left[\frac{X_t(x,v;c)-c_{\bar j}}{\varepsilon}\right]V_{t^-}(x,v;c),\end{array}\right.
\end{align}

where $f_{t^{\pm}}$ denote respectively $\lim_{s\to t^{\pm}}f(s)$ and ${\mathcal R}[x]v$ is the orthogonal reflection of $v$ with respect to the real line of direction $x$, i.e. ${\mathcal R}[x]v=-v_x+v_x^{\perp}$ when $v=v_x+v_x^{\perp},v_x\in\mathbb{R}x,v_x^{\perp}\in(\mathbb{R}x)^{\perp}$.

The obstacles may overlap, that is, it could happen that, for $j\neq k\in\mathbb{N}$, $0<|c_j-c_k|\leq2\varepsilon$. In this case, the dynamics is not well defined only if the particles hits both the obstacles at the same time, i.e, if $|X_t(x,v;c)-c_j|=|X_t(x,v;c)-c_k|=\varepsilon$, since the obstacles should reflect the velocity vector in two different directions. The points that belong to the boundary of two different obstacles are usually referred as "angular points". This problem is overcome by noticing that the Lebesgue measure of the set of initial data $(x,v)$ such that the particle hits an angular point is zero (see for example \cite{BBS1983}).

Since $|v|$ is preserved within the motion, it is assumed to be 1 with no loss of generality.

Therefore if one considers a particle with randomly distributed initial data $(x,v)\in\mathbb{R}^d\times\mathbb{S}^{d-1}$, for example through a density function $\mu_{in}(x,v)$, for finite $\varepsilon>0$ the time evolution of the density is
\[
\mu_t(x,v;c)=\mu_{in}(X_{-t}(x,v;c),V_{-t}(x,v;c)),
\]
with $X_t$ and $V_t$ described in \eqref{motolibero} and \eqref{motourto}, and one may ask under what assumptions on $c$ the quantity $\lim_{\varepsilon\to0}\mu_t(x,v;c)$ exists and is non trivial. Typically, one should scale also the time $t$ (i.e., dividing it by $\varepsilon$).

Here we are focusing on the low density case.

The case of low density and randomly Poisson distributed obstacles was first studied by Gallavotti \cite{Gallavotti1972}. The author proved that if a point particle moves in $\mathbb{R}^2$ and the centers of the obstacles are a Poisson point process of intensity $n\sim\frac{1}{2\lambda\varepsilon}$, then for any continuous and bounded probability density $\mu_{in}:\mathbb{R}^2\times\mathbb{S}^1\to[0,+\infty)$ there exists the limit of the averaged particle density

\begin{align}\label{thm:G}
\mathbb{E}[\mu_t(\cdot,\cdot;c)]\xrightarrow[\varepsilon\to0]{L^1(\mathbb{R}^2\times\mathbb{S}^1)}\mu_t,\textrm{ uniformly on compact }t\textrm{-sets},
\end{align}

where the expected value is taken with respect to the Poisson distribution of the obstacles' centers $c$. 

The previous limit is the low-density limit, that is, the Boltzmann-Grad limit, that is, one lets the size of the scatterer radius $\varepsilon\to0$ but keeping constant the mean free path length. The mean free path length is $\lambda\sim(2n\varepsilon)^{-1}$.

Moreover the limiting $\mu_t$ satisfies the Boltzmann equation

\begin{align}\label{thm:Gpt2}
\left\{\begin{array}{lcr}\partial_t\mu_t(x,v)+v\cdot\nabla_x\mu_t(x,v)+\lambda\mu_t(x,v)=\frac{\lambda}{4}\int_0^{2\pi}d\theta\mu_t(x,{\mathcal R}[\theta]v)\sin(\frac{\theta}{2}),\\\mu_0(x,v)=\mu_{in}(x,v),\end{array}\right.
\end{align}

where ${\mathcal R}[\theta]v$ is the rotation of $v$ by angle $\pi-\theta$, i.e., the same rotation ${\mathcal R}$ as before expressed as a function of the angle of the impact.

Spohn \cite{Spohn1978} strengthened the previous result by proving that the convergence of the Lorentz process
\[
(X_t(x,v;c),V_t(x,v;c))
\]
to a stochastic process holds with respect to the weak$^*$ topology of regular Borel measures on the paths space, even if more general random distributions of the obstacles are taken into account.

A further related result was obtained by Boldrighini-Bunimovich-Sinai \cite{BBS1983}: the authors proved that the scaled Lorentz process converges for almost all $c$ configurations of the obstacles' centers.

Then a natural question is what changes in \eqref{thm:G} and in \eqref{thm:Gpt2} if the obstacles' centers have periodic configuration, instead of being randomly distributed, and the expectation sign in \eqref{thm:G} is transferred to the initial data. This model is better known as periodic Lorentz Gas and it has been studied by several authors.

\subsection{The periodic Lorentz Gas.}
From now on, in this introduction we will assume $0<\varepsilon<\frac{1}{2}$, the obstacles' centers $c$ to be located in $\varepsilon\mathbb{Z}^d$ and to have radius $r_{\varepsilon}:=\varepsilon^{\frac{d}{d-1}}$. Therefore the available region where a particle can move is

\begin{align}\label{scalinggiusto}
Z_{\varepsilon}:=\{x\in\mathbb{R}^d:\textrm{dist(}x,\varepsilon\mathbb{Z}^d\textrm{)}\geq\varepsilon^{\frac{d}{d-1}}\}.
\end{align}

This way, the free path length of a particle moving in $Z_{\varepsilon}$ according to \eqref{motolibero} and \eqref{motourto} can be defined as
\[
\tau_{\varepsilon}(x,v):=\inf\{t>0:x+tv\in\partial Z_{\varepsilon}\},\qquad(x,v)\in Z_{\varepsilon}\times\mathbb{S}^1.
\]
The setting where the gas is to be studied for fixed time $t>0$ is exactly \eqref{scalinggiusto}, indeed this condition ensures that the obstacles' density is $\simeq1/\varepsilon^d$. Therefore, since a circular tube of radius $\varepsilon^{\frac{d}{d-1}}$, that is the obstacles' radius, around the line drawn by trajectory in the time interval $[0,t]$ has volume $\simeq(\varepsilon^{\frac{d}{d-1}})^{d-1}t=t\varepsilon^d$, the mean obstacles number hit before time $t$ is $\simeq t\varepsilon^d\cdot1/(\varepsilon^d)=t$. Thus, in this setting the mean free path length should have order 1 as $\varepsilon\to0$ (see for example Dumas-Dumas-Golse \cite{DDG1996} and Golse \cite{Golse2009}). Nevertheless, an equivalent and perhaps more common setting considered in literature involves placing the obstacles in $\mathbb{Z}^d$, instead of $\varepsilon\mathbb{Z}^d$, making them have radius $\varepsilon^{\frac{1}{d-1}}$, instead of $\varepsilon^{\frac{d}{d-1}}$, and studying the dynamics at times $t/\varepsilon$, instead of $t$.
\newline
\subsubsection{The distribution of the free path length in the Boltzmann-Grad limit.} In the Boltzmann-Grad limit of the periodic Lorentz gas, Bourgain-Golse-Wennberg \cite{BGW1998} and Golse-Wennberg \cite{GW2000} proved that, if one denotes by $A/B$ the quotient of $A$ with respect to the equivalence relation $x\sim y\Leftrightarrow x-y\in B$, being $\nu_{\varepsilon}$ is the uniform probability measure on $Z_{\varepsilon}/(\varepsilon\mathbb{Z}^d)\times\mathbb{S}^{d-1}$, then

\begin{align}\label{mfpl}
\Phi_{\varepsilon}(t):=\nu_{\varepsilon}\{(x,v)\in Z_{\varepsilon}/(\varepsilon\mathbb{Z}^d)\times\mathbb{S}^{d-1}:\tau_{\varepsilon}(x,v)>t\}\simeq\frac{1}{t^{d-1}},\quad\forall t\geq1,
\end{align}

where $f\simeq g$ means that there exists a constant $C>0$ such that $C^{-1}g\leq  f\leq Cg$, and such a constant in the above inequality does not depend on $\varepsilon$. In other words, $\Phi_{\varepsilon}$ is the probability that the free path length is larger than $t$: in dimension $d=2$, \eqref{mfpl} makes clear that the mean free path length is not finite if the average is computed with respect to $\nu_{\varepsilon}$. Instead, if one substitutes $\nu_{\varepsilon}$ with another probability measure concentrated on the boundary of the obstacles (see \cite{DDG1996}), the mean free path length is finite.

In dimension $d=2$, a result related to \eqref{mfpl} was obtained by Caglioti-Golse \cite{CG2003}, indeed the authors provided the exact behavior of the limit for large $t$, as $\varepsilon\to0$, of $\Phi_{\varepsilon}$, i.e.
\[
\lim_{t\to+\infty}t\lim_{\varepsilon\to0}\frac{1}{|\log\varepsilon|}\int_{\varepsilon}^{\frac{1}{4}}\frac{dr}{r}\Phi_r(t)=\frac{2}{\pi^2}.
\]
Boca-Zaharescu \cite{BZ2007} strenghtened their method obtaining exact estimates for
\[
\lim_{\varepsilon\to0}\Phi_{\varepsilon}(t),
\]
that is, the limiting distribution of the free path length, at any fixed time $t\geq0$.

Then, Caglioti-Golse \cite{CG2008,CG2010}, Marklof-Str\"ombergsson \cite{MS2008d2} and Bykovskii-Ustinovin \cite{BU2009} found an explicit expression for the transition kernel $Q$ in dimension $d=2$, that is, the Boltzmann-Grad limit of the infinitesimal probability that the next obstacle will be hit in time $s$ and with impact parameter $h$, conditioning the impact parameter of the previous collision to be $h'$. The impact parameter $h$ is defined for finite $\varepsilon>0$ as $h:=\sin(\widehat{v' n_x})$, with $v'$ the velocity after the next collision and $n_x$ the exterior normal to the obstacle in the impact site $x$. This way, the limiting probability that the next obstacle will be hit in time $t$ writes as a function of the transition kernel $Q$. Such transition kernel $Q$ is
\begin{definition}\label{defQ}

\begin{align*}
 Q(s,h|h'):=\frac{6}{\pi^2}\left\{\begin{array}{lcr}1&&0\leq s\leq \frac{1}{1+h},\\\frac{\frac{1}{s}-(1+h')}{h-h'}&&\frac{1}{1+h}<s\leq\frac{1}{1+h'},\\0&&s<0\textrm{ or }s>\frac{1}{1+h'},\end{array}\right.\quad\textrm{ if }|h'|\leq h,
\end{align*}

and the definition extends to all $h,h'\in[-1,1]$ by using the symmetries

 \begin{align*}
 Q(s,h|h')= Q(s,h'|h)= Q(s,-h|-h').
\end{align*}

\end{definition}
The stated results about the structure of the transition probability $Q$ hold for $(x,v)$ randomly distributed on $Z_{\varepsilon}\times\mathbb{S}^1$ with a probability density absolutely continuous with respect to Lebesgue measure instead of being uniformly distributed on $Z_{\varepsilon}/(\varepsilon\mathbb{Z}^2)\times\mathbb{S}^1$. Moreover, in \cite{MS2010fpl} the authors also provided a definition for the impact parameter $h$ and for the kernel $Q$ for dimensions $d\geq3$: their results involve all the dimensions $d\geq2$ but $Q$ has not an explicit formulation for all the times $s$ for dimensions $d\neq2$.
\newline
\subsubsection{The kinetic theory for the periodic Lorentz gas in the Boltzmann-Grad limit.} Going back to \eqref{thm:Gpt2}, one may also ask whether it is possible to obtain such a kinetic equation in the periodic case. The answer is no: using the heavy tail of the distribution of the free path length, Golse \cite{Golse2008} proved that in any dimension $d\geq2$ there exists an initial datum $\mu_{in}:\mathbb{T}^d\times\mathbb{S}^{d-1}\to[0,+\infty)$ such that if the obstacles are balls of radius $\varepsilon^{\frac{1}{d-1}}$ located in $\mathbb{Z}^d$, no subsequence of the density $\{\mu_{\frac{t}{\varepsilon}}\}_{\varepsilon}$ with initial datum $\mu_{in}$ can converge to the solution of the linear Boltzmann equation.

In \cite{CG2008,CG2010} the authors obtained an equation, rigorously derived in \cite{MS2008bgl}, for the time evolution of the limiting density. We restate it in dimension $d=2$, even if in \cite{MS2008bgl} any dimension $d\geq2$ has been considered. Denoting by $\mu_{in}$ the limiting probability density of the initial data, the equation writes as

{\footnotesize\begin{align}\label{evoluzionetemp}
\left\{\begin{array}{lcr}\frac{\partial\mu_t}{\partial t}(x,v,s,h)+v\cdot\nabla_x\mu_t(x,v,s,h)-\frac{\partial\mu_t}{\partial s}(x,v,s,h)=\int_{-1}^1dh'Q(s,h|h')\mu_t(x,{\mathcal R}[\theta(h')]v,0,h'),\\\mu_t(x,v,s,h)|_{t=0}=\mu_{in}(x,v)E(s,h).\end{array}\right.
\end{align}}

Let us comment the above expression. The key argument used by the authors to understand the time evolution of the Boltzmann-Grad limit of a probability density on the phase space was to extend the phase space itself by adding the couple $(s,h)=$(time to the next collision, impact parameter of the next collision) defined before. Indeed the random flight $(X_t,V_t)_{t\geq0}$ in the Boltzmann-Grad limit obtained in \cite{MS2008bgl} is not a Markov process (in $(x,v)$), therefore it is not possible to find a memoryless equation for the limiting density $\mu_t(x,v)$. Thus, the probability measure on the new expanded phase space $\mathbb{R}^2\times\mathbb{S}^1\times[0,+\infty)\times[-1,1]$ is to be understood as the probability of having at time $t$ position $x$ and velocity $v$, hitting the next obstacle within time $s$, i.e., at time $t+s$, and with impact parameter $h$.

There are still some quantities to be commented in \eqref{evoluzionetemp}. One of them is $E$, that is, the invariant probability measure for the time evolution of the density in $(s,h)$, obtained as
\begin{definition}\label{defE}

\begin{align*}
E(s,h):=\int_s^{\infty}ds'\int_{-1}^1dh' Q(s',h|h'),
\end{align*}

\end{definition}
while ${\mathcal R}[\theta(h')]$ rotates a vector $v=(\cos\theta,\sin\theta)\in\mathbb{S}^1$ by angle $\theta(h')=\pi-2\arcsin(h')$, that is
\[
{\mathcal R}[\theta(h')]v=(\cos(\theta+\theta(h')),\sin(\theta+\theta(h'))).
\]
Let us finally point out that in the previous equation $x$ can be a point in $\mathbb{T}^2$ or in $\mathbb{R}^2$. This ambiguity comes from the obstacles' lattice being periodic.

The time evolution equation \eqref{evoluzionetemp} raises the problem of studying the long time behavior of the solution $\mu_t$. To begin with, we notice that one can introduce randomness also in the new parameters $(s,h)$, getting thus the following equation for the time evolution of a density

{\footnotesize\begin{align}\label{eq:ev_completa}
\left\{\begin{array}{lcr}\frac{\partial\mu_t}{\partial t}(x,v,s,h)+v\cdot\nabla_x\mu_t(x,v,s,h)-\frac{\partial\mu_t}{\partial s}(x,v,s,h)=\int_{-1}^1dh'Q(s,h|h')\mu_t(x,{\mathcal R}[\theta(h')]v,0,h'),\\\mu_t(x,v,s,h)|_{t=0}=\mu_0(x,v,s,h).\end{array}\right.
\end{align}}

For our purpose using $\theta\in\mathbb{T}^1_{2\pi}:=\mathbb{R}/(2\pi\mathbb{Z})$ as a parameter is more comfortable  rather than $v\in\mathbb{S}^1$, therefore hereafter we will denote
\[
\mu(x,\theta,s,h):=\mu(x,v(\theta),s,h)\quad\textrm{ if }\quad v(\theta)=(\cos\theta,\sin\theta),
\]
and writing $\mu_t$ as a function of $\theta$ instead of $v$, the equation \eqref{eq:ev_completa} writes as

{\tiny\begin{align}\label{eq:ev}
\left\{\begin{array}{lcr}\frac{\partial\mu_t}{\partial t}(x,\theta,s,h)+v(\theta)\cdot\nabla_x\mu_t(x,\theta,s,h)-\frac{\partial\mu_t}{\partial s}(x,\theta,s,h)=\int_{-1}^1dh' Q(s,h|h')\mu_t(x,\theta+\pi-2\arcsin(h'),0,h'),\\\mu_t(x,\theta,s,h)|_{t=0}=\mu_0(x,\theta,s,h).\end{array}\right.
\end{align}}

In \cite{MS2008bgl} existence and uniqueness for the solutions of \eqref{eq:ev} have been proved, as well as the fact that the $L^1$ distance between two solutions is non increasing in time. As pointed out by the authors, a solution of \eqref{eq:ev} can be represented by

{\footnotesize\begin{align}
\nonumber\mu_t(x,\theta,s,h)&=\mu_0(x-tv(\theta),\theta,s+t,h)
\\
\label{rapp:ev}&+\int_0^tdt'\int_{-1}^1dh'Q(s+t-t',h|h')\mu_{t'}(x-(t-t')v(\theta),\theta+\pi-2\arcsin(h'),0,h'),
\end{align}}

so that the solution turns out to be the sum of two contributions: the first one can be understood as the probability density of the particles that collide for the first time at time $t$, while the second one represents the probability density of the particles that have collided at least one time before time $t$. Moreover, if we evaluate the equation \eqref{rapp:ev} at $s=0$ we get

{\small\begin{align*}
\mu_t(x,\theta,0,h)&=\underbrace{\mu_0(x-tv(\theta),\theta,t,h)}_{=:\tilde\mu_0(x,\theta,t,h)}
\\
&+\underbrace{\int_0^tdt'\int_{-1}^1dh'Q(t-t',h|h')\mu_{t'}(x-(t-t')v(\theta),\theta+\pi-2\arcsin(h'),0,h')}_{=:{\mathcal F}(\mu)(x,\theta,t,h)},
\end{align*}}

and, except for the initial datum $\mu_0$, this equation involves only $\mu_t(x,\theta,s=0,h)$, that is, the probability density of the particles that collide at time $t$.

The first contribution $\tilde\mu_0$ can be understood again as the probability density of the particles whose first collision occurs at time $t$, while the second contribution ${\mathcal F}(\mu)(t,x,\theta,h)$ represents the probability density of the particles that have collided at least one time before time $t$.

Since ${\mathcal F}$ is a linear function of $\mu$, we can formally write $\mu_t(x,\theta,0,h)$ as
\[
\mu_t(x,\theta,0,h)=\tilde\mu_0(x,\theta,t,h)+\sum_{n=1}^{\infty}{\mathcal F}^n(\tilde\mu_0)(x,\theta,t,h),
\]
where each term ${\mathcal F}^n(\tilde\mu_0)(x,\theta,t,h)$ represents the probability density of the particles that have collided exactly $n$ times before time $t$. Going back to \eqref{rapp:ev}, the advantage of such a representation is that it makes sense not only for regular initial data but for general $L^1$ functions.

In the previous equation four variables are involved: $(x,\theta,s,h)$, but by integrating $\mu_t$ with respect to $x$ or $(x,\theta)$, one gets two new equations in the remaining variables.
\newline
\textbf{Averaging on the position $x$.} By integrating $\mu_t$ with respect to $x$ one gets

\begin{align}\label{eq:ev_x}
\left\{\begin{array}{lcr}\frac{\partial\mu_t}{\partial t}(\theta,s,h)-\frac{\partial\mu_t}{\partial s}(\theta,s,h)=\int_{-1}^1dh' Q(s,h|h')\mu_t(\theta+\pi-2\arcsin(h'),0,h'),\\\mu_t|_{t=0}(\theta,s,h)=\mu_0(\theta,s,h).\end{array}\right.
\end{align}

The same equation would appear by considering an initial datum that does not depend on $x$, but of course this makes sense only for $x\in\mathbb{T}^2$ and not for $x\in\mathbb{R}^2$. That is, if $\mu_0$ does not depend on $x$, the same holds for $\mu_t$ at any time $t$. Of course, a solution of \eqref{eq:ev_x} admits also an alternative representation which does not include derivatives, that is,

{\scriptsize\begin{align}\label{rapp:ev_x}
\mu_t(\theta,s,h)=\mu_0(\theta,s+t,h)+\int_0^tdt'\int_{-1}^1dh'Q(s+t-t',h|h')\mu_{t'}(\theta+\pi-2\arcsin(h'),0,h'),\quad\forall t\geq0,
\end{align}}

with the same interpretation about the number of collision within time $t$ as before.
\newline
\textbf{Averaging on both the position $x$ and the velocity $v(\theta)$.}  Instead if one integrates both respect to $x$ and to $v$, the final equation would be

\begin{align}\label{eq:ev_x,v}
\left\{\begin{array}{lcr}\frac{\partial\mu_t}{\partial t}(s,h)-\frac{\partial\mu_t}{\partial s}(s,h)=\int_{-1}^1dh' Q(s,h|h')\mu_t(0,h'),\\\mu_t|_{t=0}(s,h)=\mu_0(s,h),\end{array}\right.
\end{align}

also written for more general $L^1$ initial data as

\begin{align}\label{rapp:ev_x,v}
\mu_t(s,h)=\mu_0(s+t,h)+\int_0^tdt'\int_{-1}^1dh'Q(s+t-t',h|h')\mu_{t'}(0,h'),\qquad\forall t\geq0.
\end{align}

One can notice again that if the initial datum $\mu_0$ does not depend on $x$ or on $(x,\theta)$, neither does $\mu_t$ at any time $t$, and therefore the evolution of the probability density with respect to the remaining variables $(s,h)$ is described by the equation \eqref{eq:ev_x} (respectively \eqref{eq:ev_x,v}). In \cite{CG2010,MS2008bgl} it has been proven that the only equilibria states for the equation \eqref{eq:ev_x,v} in $(s,h)$ are $cE$, for a constant $c\in\mathbb{R}$ and $E\in L^1([0,+\infty)\times[-1,1])$ introduced in Definition \ref{defE}.

Moreover, since we already observed that a generalized solution of the three kinetic equations above exists also for non regular initial data, we will use the following notation.

\begin{definition}\label{def:soluzione} A mild solution of equation \eqref{eq:ev} (respectively \eqref{eq:ev_x} and \eqref{eq:ev_x,v}) with initial datum $\mu_0\in L^1(\mathbb{T}^2\times\mathbb{T}^1_{2\pi}\times[0,+\infty)\times[-1,1])$ or $\mu_0\in L^1(\mathbb{R}^2\times\mathbb{T}^1_{2\pi}\times[0,+\infty)\times[-1,1])$ (respectively $\mu_0\in L^1(\mathbb{T}^1_{2\pi}\times[0,+\infty)\times[-1,1])$ and $\mu_0\in L^1([0,+\infty)\times[-1,1])$) is a function that satisties \eqref{rapp:ev} (respectively \eqref{rapp:ev_x} and \eqref{rapp:ev_x,v}).
\end{definition}

\subsubsection{Further literature.} Marklof-Str\"ombergsson \cite{MS2011ae} also provided asymptotic estimates, that is, for small and large $s$, for the kernel $Q$ and the invariant measure $E$ in any dimension $d\geq2$. These bounds rely on the explicit formulation of the transition kernel $Q$ only in dimension $d=2$ and allow to improve the estimates about the distribution of the free path length in \eqref{mfpl}. In \cite{MS2013pld} the authors also generalized these asymptotic estimates and the time evolution law \eqref{evoluzionetemp} to the case of finite unions of lattices, while in \cite{MS2024} also the case of spherically symmetric finite-range potentials has been studied. Moreover Marklof-T\'oth \cite{MT2016} proved a superdiffusive limit with normalization factor $\sqrt{t\log t}$, instead of $\sqrt t$, for the continuous and discrete time displacement in any dimension $d\geq2$.
\subsection{Main results.} The main results we prove concern the asymptotic behavior of the mild solutions of the equations \eqref{eq:ev}, \eqref{eq:ev_x} and \eqref{eq:ev_x,v}. In \cite{CG2010}, using relative entropy estimates, it has been proved that if $\mu_{in}\in L^{\infty}(\mathbb{T}^2\times\mathbb{T}^1_{2\pi})$ is a probability density, then the time evolution of $\mu_0:=\mu_{in}E$ converges to the equilibrium state

\begin{align}\label{convergenza*debole}
\mu_t\xrightarrow[t\to+\infty]{}\frac{1}{2\pi}E,\quad\textrm{weak}-^*\textrm{ in }L^{\infty}(\mathbb{T}^2\times\mathbb{T}^1_{2\pi}\times[0,+\infty)\times[-1,1]),
\end{align}

and also that the rate of the approach to the equilibrium with respect to the $L^2$ norm is worse than $t^{-\frac{3}{2}}$.

Our purpose is to improve this result including also some estimates on the rate of convergence to the equilibrium.

Hereafter we will denote by $\langle f\rangle$ the integral of $f$.

Our first result is the following.
\begin{theorem}\label{thm:convergenza_v,s,h}
There exists a constant $C>0$ depending only on $Q$ such that for any $p\in[1,+\infty]$ and  $\mu_0\in L^1\cap L^p(\mathbb{T}^1_{2\pi}\times[0,+\infty)\times[-1,1])$, if $\mu_t$ is the mild solution of the equation \eqref{eq:ev_x} with initial datum $\mu_0$, then

{\tiny\begin{align}
\nonumber\left\|\mu_t-\frac{\langle\mu_0\rangle}{2\pi}E\right\|_{L^p(\mathbb{T}^1_{2\pi}\times[0,+\infty)\times[-1,1])}&\leq C\frac{\left\|\mu_0\right\|_{L^1(\mathbb{T}^1_{2\pi}\times[0,+\infty)\times[-1,1])}+\left\|\mu_0\right\|_{L^p(\mathbb{T}^1_{2\pi}\times[0,+\infty)\times[-1,1])}}{t+1}
\\
\label{thm1:st1}&+C\left[\left\|\mu_0\right\|_{L^1(\mathbb{T}^1_{2\pi}\times[t/4,+\infty)\times[-1,1])}+\left\|\mu_0\right\|_{L^p(\mathbb{T}^1_{2\pi}\times[t/4,+\infty)\times[-1,1])}\right].
\end{align}}

In particular, if $\mu_0(\theta,s,h)=\mu_{in}(\theta)E(s,h)$ with $\mu_{in}\in L^p(\mathbb{T}^2\times\mathbb{T}^1_{2\pi})$, then $\mu_0\in L^1\cap L^p(\mathbb{T}^1_{2\pi}\times[0,+\infty)\times[-1,1])$ and it holds 

\begin{align}\label{thm1:st2}
\left\|\mu_t-\frac{\langle\mu_{in}\rangle}{2\pi}E\right\|_{L^p(\mathbb{T}^1_{2\pi}\times[0,+\infty)\times[-1,1])}\leq\frac{C}{t+1}\left\|\mu_{in}\right\|_{L^p(\mathbb{T}^1_{2\pi})}.
\end{align}

\end{theorem}
Notice that if $\mu_0$ does not depend on $\theta$ neither $\mu_t$ does at any time $t$, thus Theorem \ref{thm:convergenza_v,s,h} includes also the mild solutions of \eqref{eq:ev_x,v}. We also point out that if $p$ is finite, \eqref{thm1:st1} of the previous Theorem \ref{thm:convergenza_v,s,h} states that the left-hand side of the inequality vanishes as $t\to+\infty$. This holds also for $p=\infty$ only with the further assumption that $\|\mu_0\|_{L^{\infty}(\mathbb{T}^1_{2\pi}\times[t,+\infty)\times[-1,1])}\xrightarrow[t\to+\infty]{}0$, such as, for example, if $\mu_0(\theta,s,h)=\mu^{in}(\theta)E(s,h)$, as in the statement \eqref{thm1:st2}.

To state the second result we have to introduce the Fourier coefficients of a mild solution of \eqref{eq:ev}, defined both for a mild solution of the equation with $x\in\mathbb{T}^2$ and with $x\in\mathbb{R}^2$ respectively as

{\scriptsize\begin{align}\label{def:mutk}
\mu_t^k(\theta,s,h):=\int_{\mathbb{T}^2}dxe^{2\pi ik\cdot x}\mu_t(x,\theta,s,h),\quad k\in\mathbb{Z}^2;\quad\mu_t^k(\theta,s,h):=\int_{\mathbb{R}^2}dxe^{2\pi ik\cdot x}\mu_t(x,\theta,s,h),\quad k\in\mathbb{R}^2.
\end{align}}

We shall prove the following result about the Fourier coefficients above.
\begin{theorem}\label{thm:mutk}
There exists a constant $C>0$ depending only on $Q$ such that for any $p\in[1,+\infty]$ and $k\in\mathbb{Z}^2,k\neq(0,0)$, if $\mu_0\in L^1\cap L^p(\mathbb{T}^2\times\mathbb{T}^1_{2\pi}\times[0,+\infty)\times[-1,1])$ and $\{\mu_t^k\}_{k\in\mathbb{Z}^2}$ are the Fourier coefficients \eqref{def:mutk} of the mild solution of equation \eqref{eq:ev} with initial datum $\mu_0$, then

{\small\begin{align}
\nonumber\left\|\mu_t^k\right\|_{L^p(\mathbb{T}^1_{2\pi}\times[0,+\infty)\times[-1,1])}&\leq C\frac{\|\mu_0^k\|_{L^1(\mathbb{T}^1_{2\pi}\times[0,+\infty)\times[-1,1])}+\|\mu_0^k\|_{L^p(\mathbb{T}^1_{2\pi}\times[0,+\infty)\times[-1,1])}}{t+1}
\\
\label{thm2:st1}&+C\left[\|\mu_0^k\|_{L^1(\mathbb{T}^1_{2\pi}\times[\frac{t}{4},+\infty)\times[-1,1])}+\|\mu_0^k\|_{L^p(\mathbb{T}^1_{2\pi}\times[\frac{t}{4},+\infty)\times[-1,1])}\right],
\end{align}}

and in particular, if $\mu_0(x,\theta,s,h)=\mu_{in}(x,\theta)E(s,h)$ with $\mu_{in}\in L^p(\mathbb{T}^2\times\mathbb{T}^1_{2\pi})$, then $\mu_0\in L^1\cap L^p(\mathbb{T}^1_{2\pi}\times[0,+\infty)\times[-1,1])$ and it holds 

\begin{align}\label{thm2:st2}
\left\|\mu_t^k\right\|_{L^p(\mathbb{T}^1_{2\pi}\times[0,+\infty)\times[-1,1])}\leq \frac{C}{t+1}\|\mu_{in}^k\|_{L^p(\mathbb{T}^1_{2\pi})}.
\end{align}

Up to substituting the constant $C$ with $\frac{C}{\min\{1,|k|^6\}}$, the same estimates hold for $p=1$, initial datum $\mu_0\in L^1(\mathbb{R}^2\times\mathbb{T}^1_{2\pi}\times[0,+\infty)\times[-1,1])$ and Fourier coefficients $\{\mu_t^k\}_{k\in\mathbb{R}^2\setminus\{(0,0)\}}$.
\end{theorem}
As we commented before about Theorem \ref{thm:convergenza_v,s,h}, also in this case we can notice that, in the first statement \eqref{thm2:st1} of Theorem \ref{thm:mutk}, the left-hand side of the inequality is vanishing for $p\in[1,+\infty)$, but it is not necessarily infinitesimal if $p=\infty$. It is under the further assumption that also $\|\mu_0^k\|_{L^{\infty}(\mathbb{T}^1_{2\pi}\times[t,+\infty)\times[-1,1])}$ vanishes as $t\to+\infty$. That is, for example, the case of initial data $\mu_0(x,\theta,s,h)=\mu_{in}(x,\theta)E(s,h)$, as stated in \eqref{thm2:st2}.

If combined with Proposition \ref{prop:es/un} in Section \ref{app:es/un}, Theorems \ref{thm:convergenza_v,s,h} and \ref{thm:mutk} imply the following result on the flat torus $\mathbb{T}^2$.
\begin{theorem}\label{thm:conv_x,v,s,h}
Fix $p\in[1,+\infty)$, let $\mu_0\in L^p(\mathbb{T}^2\times\mathbb{T}^1_{2\pi}\times[0,+\infty)\times[-1,1])$ such that
\[
\int_{\mathbb{T}^1_{2\pi}}d\theta\int_0^{\infty}ds\int_{-1}^1dh\|\mu_0(\cdot,\theta,s,h)\|_{L^p(\mathbb{T}^2)}<+\infty,
\]
and let $\mu_t$ be the mild solution of  \eqref{eq:ev} with initial datum $\mu_0$. Then $\mu_0\in L^1(\mathbb{T}^2\times\mathbb{T}^1_{2\pi}\times[0,+\infty)\times[-1,1])$ and it holds

\begin{align}\label{thm3:st1}
\left\|\mu_t-\frac{\langle\mu_0\rangle}{2\pi}E\right\|_{L^p(\mathbb{T}^2\times\mathbb{T}^1_{2\pi}\times[0,+\infty)\times[-1,1])}\xrightarrow[t\to+\infty]{}0.
\end{align}

Under the same conditions, if $p=\infty$

\begin{align}\label{thm3:st2}
\mu_t\xrightarrow[t\to+\infty]{L^{\infty}*\textrm{-weakly}}\frac{\langle\mu_0\rangle}{2\pi}E.
\end{align}

Moreover there exists a constant $C>0$ such that for any $\mu_0$ satisfying the hypothesis above for $p=2$ it holds

{\small\begin{align}
\nonumber\left\|\mu_t-\frac{\langle\mu_0\rangle}{2\pi}E\right\|_{L^2(\mathbb{T}^2\times\mathbb{T}^1_{2\pi}\times[0,+\infty)\times[-1,1])}&\leq\frac{C}{t+1}\|\mu_0\|_{L^2(\mathbb{T}^2\times\mathbb{T}^1_{2\pi}\times[0,+\infty)\times[-1,1])}
\\
\nonumber&+\frac{C}{t+1}\int_{\mathbb{T}^1_{2\pi}}d\theta\int_0^{\infty}ds\int_{-1}^1dh\|\mu_0(\cdot,\theta,s,h)\|_{L^2(\mathbb{T}^2)}
\\
\nonumber&+C\|\mu_0\|_{L^2(\mathbb{T}^2\times\mathbb{T}^1_{2\pi}\times[\frac{t}{4},+\infty)\times[-1,1])}
\\
\label{thm3:st3}&+C\int_{\mathbb{T}^1_{2\pi}}d\theta\int_{\frac{t}{4}}^{\infty}ds\int_{-1}^1dh\|\mu_0(\cdot,\theta,s,h)\|_{L^2(\mathbb{T}^2)},
\end{align}}

and, in particular, if $\mu_0(x,\theta,s,h)=\mu_{in}(x,\theta)E(s,h)$ with $\mu_{in}\in L^2(\mathbb{T}^2\times\mathbb{T}^1_{2\pi})$, it holds

\begin{align}\label{thm3:st4}
\left\|\mu_t-\frac{\langle\mu_{in}\rangle}{2\pi}E\right\|_{L^2(\mathbb{T}^2\times\mathbb{T}^1_{2\pi}\times[0,+\infty)\times[-1,1])}\leq\frac{C}{t+1}\|\mu_{in}\|_{L^2(\mathbb{T}^2\times\mathbb{T}^1_{2\pi})}.
\end{align}

\end{theorem}
The two hypothesis on $\mu_0$ in the previous Theorem, actually coincident if $p=1$, are exactly the hypothesis that ensure that $\{\mu_t\}_{t\geq0}$ is bounded in $L^p(\mathbb{T}^2\times\mathbb{T}^1_{2\pi}\times[0,+\infty)\times[-1,1])$ (see \eqref{normaLpfinita} of Proposition \ref{prop:es/un}). These conditions cover, for example, the cases $\mu_0(x,\theta,s,h)=\mu_{in}(x,\theta)E(s,h)$, with $\mu_{in}\in L^p(\mathbb{T}^2\times\mathbb{T}^1_{2\pi})$, as well as any $\mu_0(x,\theta,s,h)=\mu_{in}(x)\nu_0(\theta,s,h)$, $\mu_{in}(x,\theta)\nu_0(s,h)$, $\mu_{in}(x,\theta,h)\nu_0(s)$ with $\mu_{in}\in L^p$ and $\nu_0\in L^1\cap L^p$ on the respective spaces.

The $L^2$ norm is the only one that we can use to get quantitative estimates about the rate of the approach to the equilibrium because the results are achieved by studying the long time behavior of the Fourier coefficients, as stated in Theorems \ref{thm:convergenza_v,s,h} and \ref{thm:mutk}.

We also point out that \eqref{thm3:st2} extends the result \eqref{convergenza*debole} in \cite{CG2010} to a slightly more general class of initial data, and also that \eqref{thm3:st3} complies with the negative result in \cite{CG2010} we mentioned before, according to which the rate of the approach to the equilibrium with respect to the $L^2$ norm should be worse than $t^{-\frac{3}{2}}$, for given initial data $\mu_{in}(x,\theta)E(s,h)$, $\mu_{in}\in L^2(\mathbb{T}^2\times\mathbb{T}^1_{2\pi})$.

Lastly, we prove the following result concerning the mild solutions of equation \eqref{eq:ev} for $x\in\mathbb{R}^2$. Before that, notice that if $\mu_t$ is a mild solution defined on $\mathbb{R}^2$, the previous results on the flat torus $\mathbb{T}^2$ can also be applied to 
\[
\sum_{k\in\mathbb{Z}^2}\mu_t(\cdot+k,\cdot,\cdot,\cdot),
\]
indeed the previous one is a periodic solution of the equation.
\begin{theorem}\label{thm:conv_R2}
Let $\mu_0\in L^1(\mathbb{R}^2\times\mathbb{T}^1_{2\pi}\times[0,+\infty)\times[-1,1])$ and $\mu_t$ the mild solution of  \eqref{eq:ev} with initial datum $\mu_0$. Then for any $\eta\in{\mathcal S}(\mathbb{R}^2)$ it holds:

\begin{align*}
\left\|\int_{\mathbb{R}^2}dx\eta(x)\mu_t(x,\cdot,\cdot,\cdot)\right\|_{L^1(\mathbb{T}^1_{2\pi}\times[0,+\infty)\times[-1,1])}\xrightarrow[t\to+\infty]{}0,
\end{align*}

where ${\mathcal S}(\mathbb{R}^2)$ is the Schwartz space of $\mathbb{R}^2$.
\end{theorem}

Let us point out that the theorems above do not require any assumption about the sign or the total mass of $\mu_0$, and also that the convergence result in Theorem \ref{thm:conv_R2} of course can not be improved by a convergence with respect to $L^1$ norm because the total mass of $\mu_t$ is preserved in time. 
\newline
\subsubsection{Outline of the paper.} In Section \ref{app:funzionidiQ} we recall and prove some properties of $\Pi$, $Q^{(n)}$, $E^{(n)}$, $f$ and $g^k$ defined in Subsection \ref{sez:defnot}. Then, in Section \ref{app:es/un} we focus on the existence and the uniqueness of the mild solutions in $L^p$ of the three equations, and we spend a few lines about the stationary solutions. In Section \ref{thetash} we prove Theorem \ref{thm:convergenza_v,s,h} and all the preliminary Lemmas we need for this purpose. In Section \ref{xthetash} we first focus on Theorem \ref{thm:mutk}, whose proof is quite similar to the proof of Theorem \ref{thm:convergenza_v,s,h}. Then, we use it to prove Theorems \ref{thm:conv_x,v,s,h} and \ref{thm:conv_R2}.

\subsection{Notations and Definitions.}\label{sez:defnot}
\subsubsection{Notations.} As we said in the introduction, we will denote by $\mathbb{T}^2:=\mathbb{R}^2/\mathbb{Z}^2$ the two-dimensional flat torus, and by $\mathbb{T}^1_{2\pi}:=\mathbb{R}/(2\pi \mathbb{Z})$ the one-dimensional flat torus with period $2\pi$. This last notations may be uncommon but we decided to use it since using only $\mathbb{T}^1$ could create misunderstandings about the period.

We also denote by $\langle f\rangle$ the integral of $f$ over the space it is defined on.

Moreover,  for $\theta\in\mathbb{T}^1_{2\pi}$, we denote $v(\theta):=(\cos\theta,\sin\theta)$ and $v^{\perp}(\theta):=(-\sin\theta,\cos\theta)$.
\newline
\subsubsection{Definitions.} Here we define some quantities that we will need in the following.
First, we recall from \cite{CG2008,MS2008d2,BU2009} that for $Q$ and $E$ as in Definitions \ref{defQ} and \ref{defE} it holds:
\[
\int_0^{\infty}ds\int_{-1}^1dhQ(s,h|h')=1,\qquad\forall h'\in[-1,1],\quad\textrm{ and }\quad\int_0^{\infty}ds\int_{-1}^1dhE(s,h)=1.
\]
Other properties are stated in Section \ref{app:funzionidiQ}. Recall also the definition of the transition probability $\Pi(h|h')$ in \cite{CG2010}, that is, the probability that the impact parameter of the next collision is $h$ if in the previous one it was $h'$.
\begin{definition}\label{def:Pi} The transition probability $\Pi:[-1,1]\times[-1,1]\to(0,+\infty)$ is

\begin{align*}
\Pi(h|h'):=\int_0^{\infty}dsQ(s,h|h').
\end{align*}

\end{definition}
By Definition \ref{defQ}, we know that $\Pi$ writes as

\begin{align*}
\Pi(h|h'):=\frac{6}{\pi^2}\frac{\log(1+h)-\log(1+h')}{h-h'}\quad\forall|h'|\leq h,
\end{align*}

and that it has the symmetries

{\footnotesize\begin{align*}
\Pi(h|h')=\Pi(h'|h)=\Pi(-h|-h')>0\quad\forall (h,h')\in[-1,1]^2,\quad\int_{-1}^1dh'\Pi(h|h')=1\quad\forall h\in[-1,1].
\end{align*}}

We also need a generalization of the kernel $Q$, as follows.
\begin{definition}\label{def:Qn} The kernels $Q^{(n)}:[0,+\infty)\times[-1,1]\times[-1,1]\to[0,+\infty)$ are defined inductively

\begin{align*}
Q^{(n)}(s,h|h'):=\int_0^sds'\int_{-1}^1dh''Q(s-s',h|h'')Q^{(n-1)}(s',h''|h'),\qquad n\geq2,
\end{align*}

with $Q^{(1)}:=Q$ of Definition \ref{defQ}.
\end{definition}
For fixed $n$, $Q^{(n)}(s,h|h')$ is to be understood as the probability (density) of having impact parameter $h$ exactly $n$ collisions and time $s$ after a collision with impact parameter $h'$.

We also define the functions $E^{(n)}$, extending the Definition \ref{defE} of $E$, but replacing $Q$ by $Q^{(n)}$.
\begin{definition}\label{def:En} The functions $E^{(n)}:[0,+\infty)\times[-1,1]\to[0,+\infty)$ are

\begin{align*}
E^{(n)}(s,h):=\int_s^{\infty}ds'\int_{-1}^1dh'Q^{(n)}(s',h|h'),\qquad n\geq1,
\end{align*}

with kernels $Q^{(n)}$ by Definition \ref{def:Qn}.
\end{definition}
We will also use the following functions $f$, mostly in Section \ref{thetash}, and $\{g^k\}_{k\in\mathbb{R}^2,k\neq(0,0)}$, mostly in Section \ref{xthetash}. We start by defining the function $h''(\theta,h')$.
\begin{definition}\label{defh''} The function $h'':\mathbb{R}\times[-1,1]\to[-1,1]$ is
\[
h''(\theta,h'):=\sin\left(\frac{\theta+2\pi-2\arcsin(h')}{2}\right)\mathbbm{1}_{[2\arcsin(h')-3\pi,2\arcsin(h')-\pi)}(\theta).
\]
\end{definition}
Then we use $h''$ to define $f$.
\begin{definition}\label{defeffe}
The function $f:\mathbb{T}^1_{2\pi}\times[0,+\infty)\times[-1,1]\times[-1,1]\to[0,+\infty)$ is

{\footnotesize\begin{align*}
f(\theta,t,h|h'):=\sum_{\ell\in\mathbb{Z}}\frac{\partial h''(\theta+2\ell\pi,h')}{\partial\theta}\int_0^tdt' Q(t-t',h|h''(\theta+2\ell\pi,h'))Q(t',h''(\theta+2\ell\pi,h')|h').
\end{align*}}

\end{definition}
Notice that by integrating $f$ over $\theta$ one gets exactly $Q^{(2)}$, as in Definition \ref{def:Qn}:

{\scriptsize\begin{align}
\nonumber\int_{\mathbb{T}^1_{2\pi}}d\theta f(\theta,t,h|h')&=\int_{2\arcsin(h')-3\pi}^{2\arcsin(h')-\pi}d\theta\frac{\partial h''(\theta,h')}{\partial\theta}\int_0^tdt' Q(t-t',h|h''(\theta,h'))Q(t',h''(\theta,h')|h')
\\
\nonumber&=\int_{-1}^1dh''\int_0^tdt'Q(t-t',h|h'')Q(t',h''|h')\quad\textrm{ changing variables }h''=h''(\theta,h')
\\
\label{intfQ2}&=Q^{(2)}(t,h|h'),\quad\textrm{ by Definition \ref{def:Qn}}.
\end{align}}

Finally, we also use $h''$ to define $g^k$.
\begin{definition}\label{def:gk} For $k\in\mathbb{R}^2,k\neq(0,0)$, and $h''$ as in Definition \ref{defh''}, the functions
\[
g^k:\mathbb{T}^1_{2\pi}\times[0,+\infty)\times[-1,1]\times\mathbb{T}^1_{2\pi}\times[-1,1]\to\mathbb{C},
\]
are defined as

{\scriptsize\begin{align*}
 g^k(\theta,t,h|\theta',h'):=\sum_{\ell\in\mathbb{Z}}\frac{\partial h''}{\partial\theta}e^{2\pi itk\cdot v(\theta)}\int_0^tdt' Q(t-t',h|h'') Q(t',h''|h')e^{2\pi it' k\cdot[v(\theta'-\pi+2\arcsin(h'))-v(\theta)]},
\end{align*}}

with
\[
h'':=h''(\theta-\theta'+2\ell\pi,h')\textrm{ from Definition \ref{defh''}}.
\]
\end{definition}
Notice also that $g^{(0,0)}$ is not included in the previous Definition because one would have
\[
g^{(0,0)}(\theta,t,h|\theta',h')=f(\theta-\theta',t,h|h'),
\]
and also that for any $k$ it holds $|g^k(\theta,t,h|\theta',h')|\leq f(\theta-\theta',t,h|h')$, with $f$ from Definition \ref{defeffe}.

\newpage
\section{Properties of the collision operator and related functions.}
\label{app:funzionidiQ}
In this preliminary Section we study the properties of collision kernel $Q$ and functions derived from it.
\subsection{Properties of $Q$ and $Q^{(n)}$.}
First we recall some basic properties of $Q$ of Definition \ref{defQ}.
\newline
\subsubsection{Behavior of $ Q$ for large $s$.}
For fixed $h,h'\in(-1,1)$, $Q$ is compactly supported in $s$. But we need a bound not depending on $h,h'$, and therefore we assert that there exists a constant $C>0$ such that

\begin{align}\label{decQ}
 Q(s,h|h')\leq\frac{C}{s+1}\quad\forall s\in[0,+\infty),h,h'\in[-1,1].
\end{align}

As pointed out in \cite{CG2010,MS2008d2}, this can be easily derived from Definition \ref{defQ} because, if $|h'|\leq h$ then
\begin{itemize}
\item if $h-h'\leq\frac{1}{2}$, then $\frac{1}{1+h'}=\frac{1}{\underbrace{1+h}_{\geq1}-\underbrace{(h-h')}_{\leq\frac{1}{2}}}\leq 2$ and therefore
\[
Q(s,h|h')=0\quad\forall s\geq2,h>0,h'\in[h-\frac{1}{2},h],
\]
\item instead if $h-h'\geq\frac{1}{2}$ then, since $\frac{1}{1+h}\leq1,\quad\forall s\geq1$ one has
\[
Q(s,h|h')=\frac{6}{\pi^2}\frac{\frac{1}{s}-(1+h')}{h-h'}\leq\frac{6}{\pi^2}\frac{2}{s}.
\]
\end{itemize}
Since $Q$ is also bounded, this proves \eqref{decQ} for $|h'|\leq h$, and thanks to the symmetries of $Q$ this exhausts all the other cases and proves \eqref{decQ}.

Now we study the properties of $Q^{(n)}$, from Definition \ref{def:Qn}, which we collect in the following Lemma.
\begin{lemma}\label{proprQn} $Q^{(n)}$ has the following three properties:

\begin{align}
\label{Qnsimm}Q^{(n)}(s,h|h')&=Q^{(n)}(s,h'|h);
\\
\label{Qnsegno}Q^{(n)}(s,h|h')&=Q^{(n)}(s,-h|-h');
\\
\label{Qnint1}\int_0^{\infty}ds\int_{-1}^1dhQ^{(n)}(s,h|h')&=1\textrm{ and in particular }E^{(n)}\leq1.
\end{align}

\end{lemma}
\begin{proof}
To begin with, we prove \eqref{Qnsimm}. For this purpose, it is sufficient to observe that, if $s_0:=s$ and $s_n:=0$, since $Q$ is symmetric we can rewrite $Q^{(n)}$ as

\begin{align*}
Q^{(n)}(s,h_0|h_n)&=\int_{s_0>s_1>\dots>s_{n-1}>0}\prod_{i=1}^{n-1}ds_i\int_{[-1,1]^{n-1}}\prod_{i=1}^{n-1}dh_i\prod_{i=0}^{n-1}Q(s_i-s_{i+1},h_i|h_{i+1})
\\
&=\int_{s_0>s_1>\dots>s_{n-1}>0}\prod_{i=1}^{n-1}ds_i\int_{[-1,1]^{n-1}}\prod_{i=1}^{n-1}dh_i\prod_{i=0}^{n-1}Q(s_i-s_{i+1},h_{i+1}|h_i),
\end{align*}

and if we change variables $\tau_i:=s-s_{n-i},k_i=h_{n-i}$ we get

{\scriptsize\begin{align*}
Q^{(n)}(s,h_0|h_n)&=\int_{\tau_0>\tau_1>\dots>\tau_{n-1}>0}\prod_{i=1}^{n-1}d\tau_i\int_{[-1,1]^{n-1}}\prod_{i=1}^{n-1}dk_i\underbrace{\prod_{i=0}^{n-1}Q(\tau_{n-1-i}-\tau_{n-i},k_{n-i-1}|k_{n-i})}_{=\prod_{i=0}^{n-1}Q(\tau_i-\tau_{i+1},k_i|k_{i+1})}
\\
&=Q^{(n)}(s,h_n|h_0).
\end{align*}}

Then, \eqref{Qnsegno} can be proven inductively on $n$, indeed by changing variables $h''\mapsto-h''$ in Definition \ref{def:Qn}, we get

{\small\begin{align*}
Q^{(n)}(s,-h|-h')&=\int_0^sds'\int_{-1}^1dh''Q(s-s',-h|h'')Q^{(n-1)}(s',h''|-h')
\\
&=\int_0^sds'\int_{-1}^1dh''\underbrace{Q(s-s',-h|-h'')}_{=Q(s-s',h|h'')}\underbrace{Q^{(n-1)}(s',-h''|-h')}_{Q^{(n-1)}(s',h''|h') \textrm{ by inductive hypothesis}}
\\
&=\int_0^sds'\int_{-1}^1dh''Q(s-s',h|h'')Q^{(n-1)}(s',h''|h')
\\
&=Q^{(n)}(s,h|h')\textrm{ by Definition }\ref{def:Qn}.
\end{align*}}

Lastly, to prove \eqref{Qnint1}, we only use that $Q$ preserves $L^1$ norm, indeed

{\footnotesize\begin{align*}
\int_0^{\infty}ds\int_{-1}^1dhQ^{(n)}(s,h|h')&=\int_0^{\infty}ds\int_{-1}^1dh\int_0^sds'\int_{-1}^1dh''Q(s-s',h|h'')Q^{(n-1)}(s',h''|h')
\\
&=\int_0^{\infty}ds'\int_{-1}^1dh''Q^{(n-1)}(s',h''|h')\underbrace{\int_{s'}^{\infty}ds\int_{-1}^1dhQ(s-s',h|h'')}_{=1}
\\
&=\int_0^{\infty}ds'\int_{-1}^1dh''Q^{(n-1)}(s',h''|h')=1\textrm{ by inductive hypothesis}.
\end{align*}}

\end{proof}
\subsection{Properties of $E$ and $E^{(n)}$.}
To begin with, we recall some properties of $E$. We now focus on the $(s,h)$ where $E$ is supported on.
\newline
\subsubsection{Support of $E$.}
The structure of the support of $E$ is easily understood: for $h\geq0$ ($h<0$ is symmetric because $E(s,-h)=E(s,h)$), we have

\begin{align*}
E(s,h)=0\textrm{ if and only if } Q(s',h|h')=0\quad\forall s'\geq s,h'\in[-1,1],
\end{align*}

that is

{\scriptsize\begin{align*}
E(s,h)=0\textrm{ if and only if }\left\{\begin{array}{lcr}s'\geq\frac{1}{1+h'}\quad\forall h'\in[-h,h],s'\geq s,\\ s'\geq\frac{1}{1+h}\quad\forall s'\geq s,h'\in(h,1],\\s'\geq\frac{1}{1-h}\quad\forall h'\in[-1,-h),s'\geq s,\end{array}\right.\textrm{ that is if and only if }s\geq\frac{1}{1-h}.
\end{align*}}

Therefore 

\begin{align}\label{sptE}
\textrm{Support of $E$ }=\left\{(s,h):h\in[-1,1],0\leq s\leq\frac{1}{1-|h|}\right\},
\end{align}

and in particular $E(s,h)>0\quad\forall h\in[-1,1],0\leq s<1$.
\begin{figure}[h!]
\centering
\includegraphics[width=8cm]{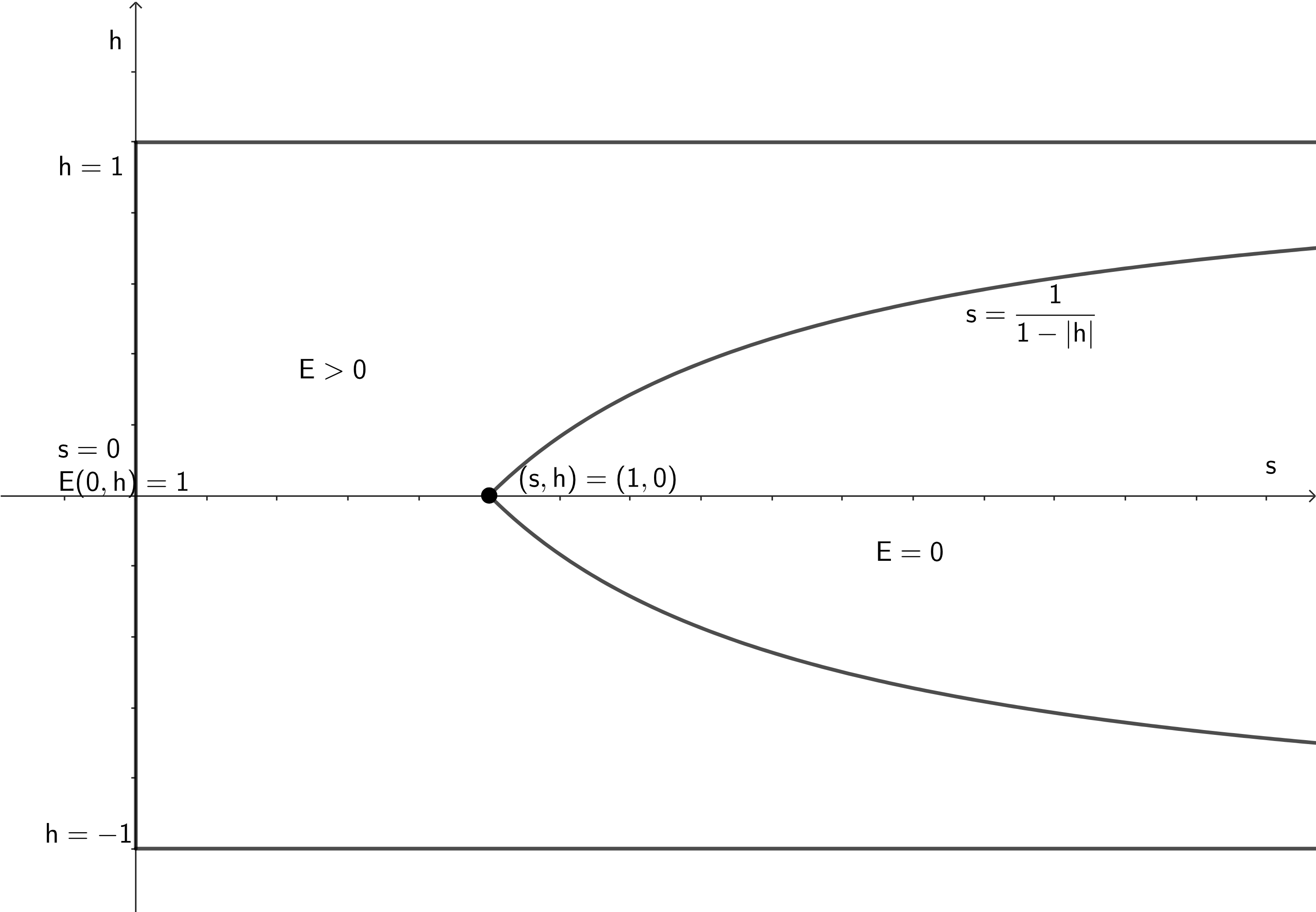}
\caption{The support of $E$ is defined by the curve
$
\{s\geq1,s=\frac{1}{1-|h|}\}\cup\{s=0,h\in[-1,1]\}\cup\{s\geq0,h=\pm1\}
$.}
\end{figure}
Moreover we have the following asymptotic estimate.
\begin{lemma}[\cite{MS2008d2}]\label{decE} There exists a constant $C>0$ such that

\begin{align*}
E(s,h)\leq\frac{C}{s+1}\mathbbm{1}_{s\leq\frac{1}{1-|h|}},\quad\forall s\in[0,+\infty),h\in[-1,1].
\end{align*}

\end{lemma}
The previous Lemma can be obtained by direct computations on $Q$. The main consequence of this Lemma is that, as a function of $s$, the support of $E$ is compact for any $h\in(-1,1)$ and that, for fixed $s$, $E$ is non zero only in an interval (in $h$) whose amplitude is $\frac{2}{s}$.

Moreover by \cite{CG2008,MS2008d2} we have

\begin{align}\label{decintE}
\int_{-1}^1dhE(s,h)\simeq\frac{1}{\pi^2s^2}\textrm{ and therefore }\int_s^{\infty}ds'\int_{-1}^1dhE(s',h)\simeq\frac{1}{\pi^2s},
\end{align}

but the rougher estimate $\int_s^{\infty}ds'\int_{-1}^1dhE(s',h)\leq\frac{C}{s+1}$ can also be proved by using Lemma \ref{decE} and is sufficient for our purposes.
Then, in the following Lemma we collect some properties of the function $E^{(n)}$, $n\geq1$.
\begin{lemma}\label{proprEn}
$\forall n\in\mathbb{N}$, $E^{(n)}$ of Definition \ref{def:En} writes also as:

\begin{align}
\label{espr:En1} E^{(n)}(s,h)&=E^{(n-1)}(s,h)+\int_0^sds'\int_{-1}^1dh'Q^{(n-1)}(s-s',h|h')E(s',h'),\quad n\geq2.
\\
\nonumber E^{(n)}(s,h)&=E^{(n-1)}(s,h)+\int_0^sds'\int_{-1}^1dh'Q(s-s',h|h')E^{(n-1)}(s',h')
\\
\label{espr:En2}&-\int_0^sds'\int_{-1}^1dh'Q(s-s',h|h')E^{(n-2)}(s',h'),\quad n\geq 3.
\end{align}

Moreover $E^{(n)}$ has the following properties: 

\begin{align}
\label{Ensimm}&E^{(n)}(s,h)=E^{(n)}(s,-h),
\\
\label{intEn}&\int_0^{\infty}ds\int_{-1}^1dhE^{(n)}(s,h)=n,
\\
\label{ubEn}&E^{(n)}(s,h)\leq \frac{c_n}{s+1},\quad c_n>0,
\\
\label{restoEn}&\int_s^{\infty}ds'\int_{-1}^1dhE^{(n)}(s',h)\leq\frac{c_n'}{s+1}.
\end{align}

\end{lemma}
\begin{proof}
To begin with, we prove property \eqref{espr:En1}. We first look at the identity

\begin{align*}
E(s,h)&=\int_s^{\infty}ds'\int_{-1}^1dh'Q(s',h|h')
\\
&=1-\int_0^sds'\int_{-1}^1dh'Q(s',h|h')
\\
&=1-\int_0^sds'\int_{-1}^1dh'Q(s-s',h|h').
\end{align*}

Taking the convolution with $Q^{(n-1)}$ of both sides, we get

\begin{align*}
&\int_0^sds'\int_{-1}^1dh'Q^{(n-1)}(s-s',h|h')E(s',h')
\\
&=\int_0^sds'\int_{-1}^1dh'Q^{(n-1)}(s-s',h|h')
\\
&-\int_0^sds'\int_{-1}^1dh'Q^{(n-1)}(s-s',h|h')\int_0^{s'}ds''\int_{-1}^1dh''Q(s'-s'',h'|h'')
\\
&=1-E^{(n-1)}(s,h)
\\
&-\int_0^sds''\int_{-1}^1dh''\int_0^{s-s''}ds'\int_{-1}^1dh'Q^{(n-1)}(s-s''-s',h|h')Q(s',h'|h'').
\end{align*}

The integral of the third summand in the right hand side of the above inequality can be also written as:

\begin{align*}
&\int_0^{s-s''}ds'\int_{-1}^1dh'Q^{(n-1)}(s-s''-s',h|h')Q(s',h'|h'')
\\
&=\int_0^{s-s''}ds'\int_{-1}^1dh'\underbrace{Q^{(n-1)}(s',h|h')}_{=Q^{(n-1)}(s',h'|h)\textrm{ by }\eqref{Qnsimm}}\underbrace{Q(s-s''-s',h'|h'')}_{=Q(s-s''-s',h''|h')\textrm{ by }\eqref{Qnsimm}}
\\
&=\int_0^{s-s''}ds'\int_{-1}^1dh'Q(s-s''-s',h''|h')Q^{(n-1)}(s',h'|h)
\\
&=Q^{(n)}(s-s'',h''|h)\textrm{ by Definition }\ref{def:Qn}
\\
&=Q^{(n)}(s-s'',h|h'')\textrm{ for the symmetry property }\eqref{Qnsimm},
\end{align*}

hence recalling the previous expression we have

\begin{align*}
&\int_0^sds'\int_{-1}^1dh'Q^{(n-1)}(s-s',h|h')E(s',h')
\\
&=1-E^{(n-1)}(s,h)-\underbrace{\int_0^sds''\int_{-1}^1dh''Q^{(n)}(s'',h|h'')}_{=1-E^{(n)}(s,h)\textrm{ thanks to }\eqref{Qnint1}\textrm{ and Definition }\ref{def:En}}
\\
&=E^{(n)}(s,h)-E^{(n-1)}(s,h),
\end{align*}

and this proves property \eqref{espr:En1}.

Property \eqref{espr:En2} follows from \eqref{espr:En1}: if $n\geq3$ we have:

\begin{align*}
&\int_0^sds'\int_{-1}^1dh'Q^{(n-1)}(s-s',h|h')E(s',h')
\\
&=\int_0^sds'\int_{-1}^1dh'E(s',h')\underbrace{\int_0^{s-s'}ds''\int_{-1}^1dh''Q(s-s'-s'',h|h'')Q^{(n-2)}(s'',h''|h')}_{\textrm{ by Definition \ref{def:Qn} of }Q^{(n-1)}}
\\
&=\int_0^sds'\int_{-1}^1dh'E(s',h')\int_{s'}^sds''\int_{-1}^1dh''Q(s-s'',h|h'')Q^{(n-2)}(s''-s',h''|h')
\\
&=\int_0^sds''\int_{-1}^1dh''Q(s-s'',h|h'')\underbrace{\int_0^{s''}ds'\int_{-1}^1dh'E(s',h')Q^{(n-2)}(s''-s',h''|h')}_{=E^{(n-1)}(s'',h'')-E^{(n-2)}(s'',h'')\textrm{ thanks to }\eqref{espr:En1}},
\end{align*}

that is, by using again property \eqref{espr:En1}, we have

{\small\begin{align*}
E^{(n)}(s,h)-E^{(n-1)}(s,h)=\int_0^sds'\int_{-1}^1dh'Q(s-s',h|h')\left[E^{(n-1)}(s',h')-E^{(n-2)}(s',h')\right],
\end{align*}}

that is \eqref{espr:En2}.

The other three properties follow from the first and the second ones, and also from the fact that they hold for $n=1$. We begin with the proof of \eqref{Ensimm}: by changing variables $h'\mapsto-h'$ in Definition \ref{def:En} of $E^{(n)}$, we have

{\scriptsize\begin{align*}
E^{(n)}(s,-h)&=\int_s^{\infty}ds'\int_{-1}^1dh'Q^{(n)}(s',-h|h')
\\
&=\int_s^{\infty}ds'\int_{-1}^1dh'\underbrace{Q^{(n)}(s',-h|-h')}_{=Q^{(n)}(s',h|h')\textrm{ thanks to }\eqref{Qnsegno}\textrm{ of Lemma \ref{proprQn}}}\textrm{ changing variables }h'\to-h'
\\
&=E^{(n)}(s,h).
\end{align*}}

To prove \eqref{intEn} we proceed by induction. It follows from \eqref{espr:En1}, indeed

{\small\begin{align*}
\int_0^{\infty}ds\int_{-1}^1dhE^{(n)}(s,h)&=\int_0^{\infty}ds\int_{-1}^1dhE^{(n-1)}(s,h)
\\
&+\int_0^{\infty}ds\int_{-1}^1dh\int_0^sds'\int_{-1}^1dh'Q^{(n-1)}(s-s',h|h')E(s',h')
\\
&=\underbrace{\int_0^{\infty}ds\int_{-1}^1dhE^{(n-1)}(s,h)}_{=n-1\textrm{ by inductive hypothesis}}
\\
&+\int_0^{\infty}ds'\int_{-1}^1dh'E(s',h')\underbrace{\int_{s'}^{\infty}ds\int_{-1}^1dhQ^{(n-1)}(s-s',h|h')}_{=1\textrm{ by }\eqref{Qnint1}}
\\
&=n-1+1=n.
\end{align*}}

To prove \eqref{ubEn} we proceed by induction too. By using again \eqref{espr:En2} we infer that

\begin{align*}
E^{(n)}(s,h)&\leq E^{(n-1)}(s,h)+\int_0^sds'\int_{-1}^1dh'Q(s-s',h|h')E^{(n-1)}(s',h')
\\
&\leq\underbrace{E^{(n-1)}(s,h)}_{\leq\frac{c_{n-1}}{s+1}}+\int_0^{s/2}ds'\int_{-1}^1dh'\underbrace{Q(s-s',h|h')}_{\leq\frac{C}{s-s'+1}\leq\frac{2C}{s+2}\textrm{ by \eqref{decQ}}}E^{(n-1)}(s',h')
\\
&+\int_{s/2}^sds'\int_{-1}^1dh'Q(s-s',h|h')\underbrace{E^{(n-1)}(s',h')}_{\leq\frac{c_{n-1}}{s'+1}\leq\frac{2c_{n-1}}{s+2}\textrm{ by inductive hypothesis }}
\\
&\leq\frac{c_{n-1}}{s+1}+\frac{2C}{s+1}\underbrace{\int_0^{s/2}ds'\int_{-1}^1dh'E^{(n-1)}(s',h')}_{\leq n-1\textrm{ by }\eqref{intEn}}
\\
&+\frac{2c_{n-1}}{s+1}\underbrace{\int_{s/2}^sds'\int_{-1}^1dh'Q(s-s',h|h')}_{\leq 1}
\\
&\leq\frac{3c_{n-1}+2C(n-1)}{s+1}=:\frac{c_n}{s+1}.
\end{align*}

Lastly, we are proving also \eqref{restoEn} by induction. Thanks to \eqref{espr:En1}, we have

{\small\begin{align*}
\int_s^{\infty}ds'\int_{-1}^1dhE^{(n)}(s',h)&=\underbrace{\int_s^{\infty}ds'\int_{-1}^1dhE^{(n-1)}(s',h)}_{\leq\frac{c_{n-1}'}{s}\textrm{ by inductive hypothesis}}
\\
&+\int_s^{\infty}ds'\int_{-1}^1dh\int_0^{s'}ds''\int_{-1}^1dh'Q^{(n-1)}(s'-s'',h|h')E(s'',h').
\end{align*}}

Now since we are applying the inductive hypothesis to the first of the previous two summands, we look at the second one, which we rewrite as

\begin{align*}
&\int_s^{\infty}ds'\int_{-1}^1dh\int_0^{s'}ds''\int_{-1}^1dh'Q^{(n-1)}(s'-s'',h|h')E(s'',h')
\\
&=\int_0^{\infty}ds''\int_{-1}^1dh'E(s'',h')\int_{\max\{s'',s\}}^{\infty}ds' \int_{-1}^1dhQ^{(n-1)}(s'-s'',h|h'),
\\
&=\int_0^sds''\int_{-1}^1dh'E(s'',h')\underbrace{\int_s^{\infty}ds' \int_{-1}^1dhQ^{(n-1)}(s'-s'',h|h')}_{=E^{(n-1)}(s-s'',h')}
\\
&+\int_s^{\infty}ds''\int_{-1}^1dh'E(s'',h')\underbrace{\int_{s''}^{\infty}ds' \int_{-1}^1dhQ^{(n-1)}(s'-s'',h|h')}_{=1\textrm{ by }\eqref{Qnint1}}
\\
&=\int_0^sds''\int_{-1}^1dh'E(s'',h')E^{(n-1)}(s-s'',h')+\int_s^{\infty}ds''\int_{-1}^1dh'E(s'',h')
\\
&=\int_0^{s/2}ds''\int_{-1}^1dh'\underbrace{E(s'',h')}_{\leq1}E^{(n-1)}(s-s'',h')
\\
&+\int_{s/2}^sds''\int_{-1}^1dh'E(s'',h')\underbrace{E^{(n-1)}(s-s'',h')}_{\leq1}+\int_s^{\infty}ds''\int_{-1}^1dh'E(s'',h')
\\
&\leq\int_{s/2}^{\infty}ds''\int_{-1}^1dh'E^{(n-1)}(s'',h')+\int_{s/2}^{\infty}ds''\int_{-1}^1dh'E(s'',h')
\\
&+\int_s^{\infty}ds''\int_{-1}^1dh'E(s'',h')
\\
&\leq\frac{2c_{n-1}'+2C}{s+2}+\frac{C}{s+1}\textrm{ by inductive hypothesis and with }C\textrm{ derived by }\eqref{decintE}.
\end{align*}

This concludes the proof of Lemma \ref{proprEn}.
\end{proof}
\subsection{Properties of $\Pi$, $f$ and $g^k$.}
In this Subsection we want to prove some properties of the functions written above.
\newline
\subsubsection{Properties of $\Pi$.}
$\Pi$ is not finite for any choice of $h$ and $h'$, indeed $\Pi(1|-1)=+\infty$, but we can prove that $\Pi$ diverges at most logarithmically at most.
\begin{lemma}\label{divPi}
For any $\varepsilon\in(0,1)$, the transition kernel $\Pi$ in Definition \ref{def:Pi} satisfies
\[
\Pi(h|h')\leq\frac{6}{\pi^2}\max\left\{\frac{1}{\varepsilon},\frac{\log 2-\log(1-|h|)}{2(1-\varepsilon)}\right\}.
\]
\end{lemma}
\begin{proof}
We prove this property case by case, and such cases are not necessarily disjoints.

We begin with the following case:

\begin{align}\label{caso1}
\textrm{If }h\textrm{ and }h'\textrm{ have the same sign, then } \Pi(h|h')\leq\frac{6}{\pi^2}.
\end{align}

To prove inequality \eqref{caso1} we can notice that, since the hypothesis is symmetric in $h$ and $h'$, we can restrict to the case $h\geq h'\geq0$, where we have

{\scriptsize\begin{align*}
\quad\Pi(h|h')=\frac{6}{\pi^2}\frac{\log(1+h)-\log(1+h')}{h-h'}=\frac{6}{\pi^2}\frac{1}{1+\xi},\quad\xi\in[h',h]\subseteq[0,1],\textrm{ and therefore }\Pi(h|h')\leq\frac{6}{\pi^2}.
\end{align*}}

Now let us fix $\varepsilon\in(0,1)$.

\begin{align}\label{caso2}
\textrm{If }|h|,|h'|\leq1-\varepsilon\textrm{ then }\Pi(h|h')\leq\frac{6}{\varepsilon\pi^2}.
\end{align}

This case is similar to the previous one, and since the hypothesis in \eqref{caso2} is symmetric with respect to exchanging $h$ and $h'$, we can restrict to the case $h\geq |h'|$, indeed by symmetry this also exhausts the other possibilities:

{\footnotesize\begin{align*}
\Pi(h|h')=\frac{6}{\pi^2}\frac{\log(1+h)-\log(1+h')}{h-h'}=\frac{6}{\pi^2}\frac{1}{1+\xi},\quad\xi\in[h',h]\subseteq[-h,h]\subseteq[-1+\varepsilon,1-\varepsilon],
\end{align*}}

therefore 
\[
\Pi(h|h')\leq\frac{1}{\varepsilon}\frac{6}{\pi^2}.
\]

Therefore in the previous cases $\Pi$ is bounded by a constant that depends on $\varepsilon$.
The same holds in the following case:

\begin{align}\label{caso3}
\textrm{If }|h|\leq1-\varepsilon\leq|h'|\textrm{ or }|h'|\leq1-\varepsilon\leq|h|,\textrm{ and }hh'\leq0,\textrm{ then }\Pi(h|h')\leq\frac{6}{\varepsilon\pi^2}.
\end{align}

To prove property \eqref{caso3} we look at the case $|h'|\leq1-\varepsilon\leq|h|$:

{\footnotesize\begin{align*}
\textrm{if }h>0\geq h',&\textrm{ then } h\geq1-\varepsilon\geq|h'|=-h',\textrm{ then }\Pi(h|h')=\frac{6}{\pi^2}\frac{\log(1+h)-\log(1+h')}{h-h'}
\\
&=\frac{6}{\pi^2(1+\xi)},\textrm{ with }\xi\in[h',h]\subseteq[-1+\varepsilon,1],\textrm{ therefore }\Pi(h|h')\leq\frac{1}{\varepsilon}\frac{6}{\pi^2},
\\
\textrm{if }h<0\leq h',&\textrm{ then }-h=|h|\geq1-\varepsilon\geq|h'|=|-h'|=h'\geq-h',\textrm{ then }\Pi(h|h')=\Pi(-h|-h')
\\
&=\frac{6}{\pi^2}\frac{\log(1-h)-\log(1-h')}{(-h)-(-h')}=\frac{1}{1+\xi},\textrm{ with }\xi\in[-h',-h]\subseteq[-1+\varepsilon,1],
\\
&\textrm{ therefore }\Pi(h|h')\leq\frac{6}{\varepsilon\pi^2}.
\end{align*}}

By symmetry this exhausts also the case $|h|\leq1-\varepsilon\leq|h'|$ and therefore \eqref{caso3} is proved.

It remains to be proven that

\begin{align}\label{caso4}
\textrm{If }|h|,|h'|\geq1-\varepsilon\textrm{ and }hh'<0,\textrm{ then }\Pi(h|h')\leq\frac{6(\log 2-\log(1-|h|))}{2\pi^2(1-\varepsilon)}.
\end{align}

We prove \eqref{caso4} case by case, as follows:

{\scriptsize\begin{align*}
\textrm{if }h\geq-h'\geq1-\varepsilon\quad\textrm{ then }\quad\Pi(h|h')&=\frac{6}{\pi^2}\frac{\log(1+h)-\log(1+h')}{\underbrace{h-h'}_{\geq2(1-\varepsilon)}}\leq\frac{6(\log 2-\log(1+h'))}{2\pi^2(1-\varepsilon)}
\\
&\leq\frac{6(\log 2-\log(1-|h|))}{2\pi^2(1-\varepsilon)},
\\
\textrm{if }-h'\geq h\geq1-\varepsilon\quad\textrm{ then }\quad\Pi(h|h')&=\Pi(-h|-h')=\Pi(-h'|-h)=\frac{6}{\pi^2}\frac{\log(1-h')-\log(1-h)}{\underbrace{(-h')-(-h)}_{\geq2(1-\varepsilon)}}
\\
&\leq\frac{6(\log 2-\log(1-h))}{2\pi^2(1-\varepsilon)}=\frac{6(\log 2-\log(1-|h|))}{2\pi^2(1-\varepsilon)},
\\
\textrm{if }h'\geq-h\geq1-\varepsilon\quad\textrm{ then }\quad\Pi(h|h')&=\Pi(h'|h)=\frac{6}{\pi^2}\frac{\log(1+h')-\log(1+h)}{\underbrace{h'-h}_{\geq2(1-\varepsilon)}}
\\
&\leq\frac{6(\log 2-\log(1+h))}{2\pi^2(1-\varepsilon)}=\frac{6(\log 2-\log(1-|h|))}{2\pi^2(1-\varepsilon)},
\\
\textrm{if }-h\geq h'\geq1-\varepsilon\quad\textrm{ then }\quad\Pi(h|h')&=\Pi(-h|-h')=\frac{6}{\pi^2}\frac{\log(1-h)-\log(1-h')}{\underbrace{(-h)-(-h')}_{\geq2(1-\varepsilon)}}
\\
&\leq\frac{6(\log 2-\log(1-h'))}{2\pi^2(1-\varepsilon)}\leq\frac{6(\log 2-\log(1+h))}{2\pi^2(1-\varepsilon)}\\
&=\frac{6(\log 2-\log(1-|h|))}{2\pi^2(1-\varepsilon)}.
\end{align*}}

By collecting the estimates \eqref{caso1}, \eqref{caso2}, \eqref{caso3} and \eqref{caso4}, since $\varepsilon\in(0,1)$ implies $\frac{1}{\varepsilon}>1$, we get the desired conclusion.
\end{proof}
\subsubsection{Properties of $f$.}
We study here the properties of the function $f$ in Definition \ref{defeffe}. 
\begin{lemma}\label{propreffe}
The function $f:\mathbb{T}^1_{2\pi}\times[0,+\infty)\times[-1,1]\times[-1,1]\to[0,+\infty)$ in Definition \ref{defeffe} has the following properties:

\begin{align}
\label{fint1}&\int_{\mathbb{T}^1_{2\pi}}d\theta\int_0^{\infty}dt\int_{-1}^1dhf(\theta,t,h|h')=1\quad\forall h'\in[-1,1],
\\
\label{fint1dinuovo}&\int_{\mathbb{T}^1_{2\pi}}d\theta\int_0^{\infty}dt\int_{-1}^1dh'f(\theta,t,h|h')=1\quad\forall h\in[-1,1],
\\
\label{deceffe}&f(\theta,t,h|h')\leq\frac{C}{t+1}\qquad\forall(\theta,t,h|h')\in\mathbb{T}^1_{2\pi}\times[0,+\infty)\times[-1,1]\times[-1,1],
\\
\label{decQf}&\int_0^tdt'\int_{-1}^1dh''Q(t-t',h|h'')f(\theta+\pi-2\arcsin(h''),t',h''|h')\leq \frac{C}{t+1}.
\end{align}

\end{lemma}
\begin{proof} We begin with the proof of \eqref{fint1}. We have

{\footnotesize\begin{align*}
\int_{\mathbb{T}^1_{2\pi}}d\theta\int_0^{\infty}dt\int_{-1}^1dhf(\theta,t,h|h')&=\int_0^{\infty}dt\int_{-1}^1dh\int_{\mathbb{T}^1_{2\pi}}d\theta f(\theta,t,h|h')\textrm{ changing integration order}
\\
&=\int_0^{\infty}dt\int_{-1}^1dhQ^{(2)}(t,h|h')\textrm{ by property \eqref{intfQ2}}
\\
&=1\textrm{ by property \eqref{Qnint1} of Lemma \ref{proprQn}.}
\end{align*}}

The proof of \eqref{fint1dinuovo} works in the same way:

{\footnotesize\begin{align*}
\int_{\mathbb{T}^1_{2\pi}}d\theta\int_0^{\infty}dt\int_{-1}^1dh'f(\theta,s,h|h')&=\int_0^{\infty}dt\int_{-1}^1dh'\int_{\mathbb{T}^1_{2\pi}}d\theta f(\theta,t,h|h')\textrm{ changing integration order}
\\
&=\int_0^{\infty}dt\int_{-1}^1dh'Q^{(2)}(t,h|h')\textrm{ by property \eqref{intfQ2}}
\\
&=\int_0^{\infty}dt\int_{-1}^1dh'Q^{(2)}(t,h'|h)=1\textrm{ by property \eqref{Qnsimm} of Lemma \ref{proprQn}}
\\
&=1\textrm{ by property \eqref{Qnint1} of Lemma \ref{proprQn}.}
\end{align*}}

As for \eqref{deceffe}, first recall $f$ is obtained as

{\footnotesize\begin{align*}
f(\theta,t,h|h')=\sum_{\ell\in\mathbb{Z}}\frac{\partial h''(\theta+2\ell\pi,h')}{\partial\theta}\int_0^tdt' Q(t-t',h|h''(\theta+2\ell\pi,h'))Q(t',h''(\theta+2\ell\pi,h')|h'),
\end{align*}}

with

\begin{align*}
h''(\theta,h')=\sin\left(\frac{\theta+2\pi-2\arcsin(h')}{2}\right)\mathbbm{1}_{[2\arcsin(h')-3\pi,2\arcsin(h')-\pi]}(\theta).
\end{align*}

Since $\cos^2+\sin^2=1$ and $f$ is obtained by extending periodically its definition for
\[
\theta\in[2\arcsin(h')-3\pi,2\arcsin(h')-\pi],
\]
where $\cos\left(\frac{\theta+2\pi-2\arcsin(h')}{2}\right)\geq0$, we have

{\footnotesize\begin{align*}
f(\theta,t,h|h')=\sum_{\ell\in\mathbb{Z}}\frac{\sqrt{1-h''(\theta+2\ell\pi,h')^2}}{2}\int_0^tdt' Q(t-t',h|h''(\theta+2\ell\pi,h'))Q(t',h''(\theta+2\ell\pi,h')|h').
\end{align*}}

It is then sufficient to prove that there exists a constant $C>0$ such that

\begin{align*}
\frac{\sqrt{1-h''^2}}{2}\int_0^tdt' Q(t-t',h|h'')Q(t',h''|h')\leq\frac{C}{t+1}\quad\forall h,h',h''\in[-1,1].
\end{align*}

This holds because

\begin{align*}
&\frac{\sqrt{1-h''^2}}{2}\int_0^tdt' Q(t-t',h|h'')Q(t',h''|h')
\\
&=\frac{\sqrt{1-h''^2}}{2}\int_0^{\frac{t}{2}}dt'  \underbrace{Q(t-t',h'|h'')}_{\leq\frac{C}{t-t'+1}\leq\frac{2C}{t+2}\textrm{ by }\eqref{decQ}}Q(t',h''|h)
\\
&+\frac{\sqrt{1-h''^2}}{2}\int_{\frac{t}{2}}^tdt' Q(t-t',h|h'')\underbrace{Q(t',h''|h)}_{\leq\frac{C}{t'+1}\leq\frac{2C}{t+2}\textrm{ by }\eqref{decQ}}
\\
&\leq\frac{C\sqrt{1-h''^2}}{t+2}\left[\int_0^{\frac{t}{2}}dt' Q(t',h''|h)+\int_{\frac{t}{2}}^tdt' Q(t-t',h'|h'')\right]
\\
&\leq\frac{C\sqrt{1-h''^2}}{t+2}\left[\Pi(h''|h)+\Pi(h'|h'')\right]
\\
&\leq\frac{C\sqrt{1-h''^2}}{t+2}\frac{6}{\pi^2}\max\left\{\frac{1}{\varepsilon},\frac{\log 2-\log(1-|h''|)}{2(1-\varepsilon)}\right\},
\end{align*}

where in the last inequality we used Lemma \ref{divPi}.

Therefore

{\tiny\begin{align*}
\frac{\sqrt{1-h''^2}}{2}\int_0^tdt' Q(t-t,h|h'')Q(t',h''|h')\leq\frac{C\sqrt{1+|h''|}\sqrt{1-|h''|}}{t+2}\frac{6}{\pi^2}\max\left\{\frac{1}{\varepsilon},\frac{\log 2-\log(1-|h''|)}{2(1-\varepsilon)}\right\},
\end{align*}}

hence the desired conclusion, since the function $\sqrt{x}\log x$ is bounded around $0$.

Lastly, \eqref{decQf} follows from property \eqref{deceffe}, indeed

\begin{align*}
&\int_0^tdt'\int_{-1}^1dh''Q(t-t',h|h'')f(\theta+\pi-2\arcsin(h''),t',h''|h')
\\
&=\int_0^{\frac{t}{2}}dt'\int_{-1}^1dh''Q(t-t',h|h'')\underbrace{f(\theta+\pi-2\arcsin(h''),t',h''|h')}_{\leq C\textrm{ by }\eqref{deceffe}}
\\
&+\int_{\frac{t}{2}}^tdt'\int_{-1}^1dh''Q(t-t',h|h'')\underbrace{f(\theta+\pi-2\arcsin(h''),t',h''|h')}_{\leq \frac{C}{t'+1}\leq\frac{2C}{t+2}\textrm{ by }\eqref{deceffe}}
\\
&\leq C\underbrace{\int_{\frac{t}{2}}^{\infty}dt'\int_{-1}^1dh''Q(t',h|h'')}_{=E(\frac{t}{2},h)\leq\frac{C}{\frac{t}{2}+1}\textrm{ by Lemma \ref{decE}}}+\frac{2C}{s+2}\underbrace{\int_0^{\frac{t}{2}}dt'\int_{-1}^1dh''Q(t',h|h'')}_{\leq1}
\\
&\leq\frac{C'}{s+2}.
\end{align*}

\end{proof}
\subsubsection{Properties of $g^k$.}
We state here some properties of the function $g^k$ introduced in Definition \ref{def:gk}.
\begin{lemma}\label{lemma:proprgk}
For $k\in\mathbb{R}^2,k\neq(0,0)$, the function $g^k:\mathbb{T}^1_{2\pi}\times[0,+\infty)\times[-1,1]\times\mathbb{T}^1_{2\pi}\times[-1,1]\to\mathbb{C}$ by Definition \ref{def:gk} has the following properties: for any $(\theta,t,h|\theta',h')$

{\footnotesize\begin{align}
\label{decgk}&|g^k(\theta,t,h|\theta',h')|\leq\frac{C}{t+1},
\\
\label{decQgk}&\left|\int_0^tdt'\int_{-1}^1dh''Q(t-t',h|h'')e^{2\pi i(t-t')k\cdot v(\theta)}g^k(\theta+\pi-2\arcsin(h''),t',h''|\theta',h')\right|\leq\frac{C}{t+1},
\\
\label{gkint<1}&\|g^k(\theta,\cdot,h|\cdot,\cdot)\|_{L^1}\leq1-C'\min\{1,|k|^2\}\qquad\forall(\theta,h)\in\mathbb{T}^1_{2\pi}\times[-1,1],
\end{align}}

where $C>0$ and $C'\in(0,1)$ do not depend on $k\in\mathbb{R}^2$, $\theta\in\mathbb{T}^1_{2\pi}$ or $h\in[-1,1]$. 
\end{lemma}
\begin{proof} The properties \eqref{decgk} and \eqref{decQgk} follow immediately from Definitions \ref{def:gk} and \ref{defeffe} of $g^k$ and $f$, indeed by definition we have
\[
|g^k(\theta,t,h|\theta',h')|\leq f(\theta-\theta',s,h|h'),
\]
and since the estimates \eqref{deceffe} and \eqref{decQf} apply to $f$ the proof of the first two statements is concluded.

To prove the other two properties, we write $\|g^k(\theta,\cdot,h|\cdot,\cdot)\|_{L^1}$ as

{\footnotesize\begin{align*}
&\|g^k(\theta,\cdot,h|\cdot,\cdot)\|_{L^1}
\\
&=\int_0^{\infty}dt\int_{-1}^1dh'\int_{-1}^1dh''\left|\int_0^tdt' Q(t-t',h|h'')Q(t',h''|h')e^{2\pi it' k\cdot(v(\theta+\pi-2\arcsin(h''))-v(\theta))}\right|.
\end{align*}}

Now we prove \eqref{gkint<1}. For this purpose we use the properties of $Q^{(2)}$ introduced by Definition \ref{def:Qn}.

Thanks to \eqref{Qnint1} of Lemma \ref{proprQn}, $Q^{(2)}$ has integral 1. Hence we have

{\footnotesize\begin{align*}
1-\|g^k(\theta,\cdot,h|\cdot,\cdot)\|_{L^1}&=\int_0^{\infty}dt\int_{-1}^1dh'\int_{-1}^1dh''\left[\int_0^tdt' Q(t-t',h|h'')Q(t',h''|h')\right.
\\
&-\left.\left|\int_0^tdt' Q(t-t',h|h'')Q(t',h''|h')e^{2\pi it' k\cdot(v(\theta+\pi-2\arcsin(h''))-v(\theta))}\right|\right]
\\
&\geq\int_0^{\frac{1}{2}}dt\underbrace{\int_{-1}^1dh'}_{=2}\int_{-1}^1dh''\left[\int_0^tdt' \underbrace{Q(t-t',h|h'')}_{=\frac{6}{\pi^2}}\underbrace{Q(t',h''|h')}_{=\frac{6}{\pi^2}}\right.
\\
&-\left.\left|\int_0^tdt' \underbrace{Q(t-t',h|h'')}_{=\frac{6}{\pi^2}}\underbrace{Q(t',h''|h')}_{=\frac{6}{\pi^2}}e^{2\pi it' k\cdot(v(\theta+\pi-2\arcsin(h''))-v(\theta))}\right|\right]
\\
&=\frac{72}{\pi^4}\int_0^{\frac{1}{2}}dt\int_{-1}^1dh''\left[t-\left|\int_0^tdt'e^{it'2\pi k\cdot(v(\theta+\pi-2\arcsin(h''))-v(\theta))}\right|\right],
\end{align*}}

and therefore, since
\[
\left|\int_0^tdt'e^{it'\omega}\right|=\left|\frac{e^{it\omega}-1}{\omega}\right|=t\frac{\sqrt{2(1-\cos(|\omega|t))}}{|\omega|t},
\]
if $x(\theta,t,h''):=2\pi|k\cdot(v(\theta+\pi-2\arcsin(h''))-v(\theta))|t$, we got

\begin{align}\label{primastima}
1-\|g^k(\theta,\cdot,h|\cdot,\cdot)\|_{L^1}\geq\frac{72}{\pi^4}\int_0^{\frac{1}{2}}dtt\int_{-1}^1dh''\left[1-\frac{\sqrt{2(1-\cos x(\theta,t,h''))}}{x(\theta,t,h'')}\right].
\end{align}

First we want to estimate from below the integrand. To this purpose, notice that since by direct computations one can find a constant $c''>0$ such that
\[
1-\frac{\sqrt{2(1-\cos x)}}{x}\geq c''\min\{x,c''\}^2,
\]
by \eqref{primastima} one gets

\begin{align}\label{secondastima}
1-\|g^k(\theta,\cdot,h|\cdot,\cdot)\|_{L^1}\geq\frac{72c''}{\pi^4}\int_0^{\frac{1}{2}}dtt\int_{-1}^1dh''\min\left\{x(\theta,t,h''),c''\right\}^2,
\end{align}

and therefore we finally have to bound from below $x(\theta,t,h'')$, at least in a suitable region that we will denote by $A^k(\theta)$. To this purpose, we first observe that

\begin{align*}
x(\theta,t,h'')&=2\pi|k\cdot(v(\theta+\pi-2\arcsin(h''))-v(\theta))|t
\\
&=4\pi t\sqrt{1-h''^2}|k|\left|\sqrt{1-h''^2}\hat k\cdot v(\theta)-h''\hat k\cdot v^{\perp}(\theta)\right|,
\end{align*}

with $\hat k:=\frac{k}{|k|}$ and $v^{\perp}(\theta)=(-\sin\theta,\cos\theta)$. 

Now
{\footnotesize\[
\textrm{ if }k\cdot v^{\perp}(\theta)=0,\textrm{ then }\hat k\cdot v(\theta)=\pm1\textrm{ and therefore }\left|\sqrt{1-h''^2}\hat k\cdot v(\theta)-h''\hat k\cdot v^{\perp}(\theta)\right|=\sqrt{1-h''^2},
\]}
thus we fix $\delta\in(0,\frac{1}{3})$ and in this case we define

\begin{align}\label{misuraAkkvperp=0}
A^k(\theta):=[-1+\delta,1-\delta],\textrm{ with measure }|A^k(\theta)|=2-2\delta,
\end{align}

 such that for any $h''\in A^k(\theta)$
 
{\small\begin{align}\label{kvperp=0}
x(\theta,t,h'')=4\pi t\sqrt{1-h''^2}|k|\left|\sqrt{1-h''^2}\hat k\cdot v(\theta)-h''\hat k\cdot v^{\perp}(\theta)\right|=4\pi t(1-h''^2)|k|\geq4\pi t\delta|k|.
\end{align}}

Therefore hereafter we can assume $\hat k\cdot v^{\perp}(\theta)\neq0$. In this case
\[
\left|\sqrt{1-h''^2}\hat k\cdot v(\theta)-h''\hat k\cdot v^{\perp}(\theta)\right|=0\Leftrightarrow h''=\hat k\cdot v(\theta)\textrm{sign}(\hat k\cdot v^{\perp}(\theta))=:h''(\theta),
\]
and we can rearrange the previous term as follows:

{\scriptsize\begin{align}
\nonumber\left|\sqrt{1-h''^2}\hat k\cdot v(\theta)-h''\hat k\cdot v^{\perp}(\theta)\right|&=\left|(\sqrt{1-h''^2}-\sqrt{1-h''^2(\theta)})\hat k\cdot v(\theta)-(h''-h''(\theta))\hat k\cdot v^{\perp}(\theta)\right|
\\
\nonumber&=|h''-h''(\theta)|\left|\frac{h''(\theta)+h''}{\sqrt{1-h''^2}+\sqrt{1-h''^2(\theta)}}\hat k\cdot v(\theta)+\hat k\cdot v^{\perp}(\theta)\right|
\\
\nonumber&=|h''-h''(\theta)|\left|\frac{(\hat k\cdot v(\theta))^2\textrm{sign}(\hat k\cdot v^{\perp}(\theta))+h''\hat k\cdot v(\theta)}{\sqrt{1-h''^2}+|\hat k\cdot v^{\perp}(\theta)|}+\hat k\cdot v^{\perp}(\theta)\right|
\\
\nonumber&=|h''-h''(\theta)|\left|\frac{\textrm{sign}(\hat k\cdot v^{\perp}(\theta))+h''\hat k\cdot v(\theta)+\sqrt{1-h''^2}\hat k\cdot v^{\perp}(\theta)}{\sqrt{1-h''^2}+|\hat k\cdot v^{\perp}(\theta)|}\right|
\\
\label{terzastima}&=|h''-h''(\theta)|\frac{1+h''h''(\theta)+\sqrt{1-h''^2}\sqrt{1-h''^2(\theta)}}{\sqrt{1-h''^2}+|\hat k\cdot v^{\perp}(\theta)|}.
\end{align}}

Now, for $\delta<\frac{1}{3}$ fixed as before, we define the region

\begin{align}\label{misuraAtheta}
A^k(\theta):=\left\{h''\in[-1,1]:h''h''(\theta)\geq0,|h''-h''(\theta)|\geq\delta,1-|h''|\geq\delta\right\},
\end{align}

with measure 
\[
|A^k(\theta)|\geq1-3\delta.
\]
Using the previous computation \eqref{terzastima}, for any $h''\in A^k(\theta)$ we have

{\footnotesize\begin{align}\label{stimabassox}
x(\theta,t,h'')=4\pi t\underbrace{\sqrt{1-h''^2}}_{\geq\sqrt\delta}|k|\underbrace{|h''-h''(\theta)|}_{\geq\delta}\frac{1+\overbrace{h''h''(\theta)}^{\geq0}+\overbrace{\sqrt{1-h''^2}\sqrt{1-h''^2(\theta)}}^{\geq0}}{\underbrace{\sqrt{1-h''^2}}_{\leq1}+\underbrace{|\hat k\cdot v^{\perp}(\theta)|}_{\leq1}}\geq2\pi t|k|\delta\sqrt\delta,
\end{align}}

and combining this estimate with \eqref{secondastima} we have

{\small\begin{align*}
1-\|g^k(\theta,\cdot,h|\cdot,\cdot)\|_{L^1}&\geq\frac{72c''}{\pi^4}\int_0^{\frac{1}{2}}dtt\int_{-1}^1dh''\min\left\{x(\theta,t,h''),c''\right\}^2
\\
&\geq\frac{72c''}{\pi^4}\int_0^{\frac{1}{2}}dtt\int_{A^k(\theta)}dh''\min\left\{x(\theta,t,h''),c''\right\}^2
\\
&\geq\frac{72c''}{\pi^4}\int_0^{\frac{1}{2}}dtt\int_{A^k(\theta)}dh''\min\left\{2\pi|k|\delta\sqrt\delta t,c''\right\}^2\textrm{ by \eqref{kvperp=0} and \eqref{stimabassox}}
\\
&=\frac{72c''}{\pi^4}\int_0^{\frac{1}{2}}dtt\min\left\{2\pi|k|\delta\sqrt\delta t,c''\right\}^2|A^k(\theta)|
\\
&\geq\frac{72c''(1-3\delta)}{\pi^4}\int_0^{\frac{1}{2}}dtt\min\left\{2\pi|k|\delta\sqrt\delta t,c''\right\}^2\textrm{ by \eqref{misuraAkkvperp=0} and \eqref{misuraAtheta}}
\\
&=\frac{72c''(1-3\delta)}{\pi^4}\left\{\begin{array}{lcr}\frac{(\pi\delta\sqrt\delta)^2|k|^2}{16},&&|k|\leq\frac{c''}{\pi\delta\sqrt\delta},\\\frac{c''^4}{(4\pi\delta\sqrt\delta|k|)^2}+\frac{c''^2}{8}(1-\frac{c''^2}{(\pi\delta\sqrt\delta|k|)^2}),&&|k|\geq\frac{c''}{\pi\delta\sqrt\delta},\end{array}\right.
\\
&\geq C\min\{1,|k|^2\},
\end{align*}}

up to another constant $C\in(0,1)$ not depending on $\theta$ or $h$. Therefore property \eqref{gkint<1} is proved.
\end{proof}

\newpage
\section{Existence and uniqueness in $L^p$.}\label{app:es/un}
In this Section we focus on the existence and the uniqueness of the mild solution of the three equations \eqref{eq:ev}, \eqref{eq:ev_x} and \eqref{eq:ev_x,v}, meaning \eqref{rapp:ev}, \eqref{rapp:ev_x} and \eqref{rapp:ev_x,v} respectively, as in Definition \ref{def:soluzione}. The main result we need is the following.
\begin{proposition}\label{prop:es/un} Let $\mathbb{X}=\mathbb{R}^2$ or $\mathbb{X}=\mathbb{T}^2=\mathbb{R}^2/\mathbb{Z}^2$, $p\in[1,+\infty]$ and $\mu_0\in \cap_{T>0}L^p(\mathbb{X}\times\mathbb{T}^1_{2\pi}\times[0,T]\times[-1,1])$. Then the mild solution of the equation \eqref{eq:ev} is unique in the class $\cap_{T>0}L^p(\mathbb{X}\times\mathbb{T}^1_{2\pi}\times[0,T]\times[-1,1])$. Moreover for any $T>0$ and at any time $t>0$ it holds

{\small\begin{align}\label{mutL1cpt}
\int_{\mathbb{X}}dx\int_{\mathbb{T}^1_{2\pi}}d\theta\int_0^Tds\int_{-1}^1dh|\mu_t(x,\theta,s,h)|\leq\int_{\mathbb{X}}dx\int_{\mathbb{T}^1_{2\pi}}d\theta\int_0^{T+t}ds\int_{-1}^1dh|\mu_0(x,\theta,s,h)|,
\end{align}}

and if $\mu_0\in L^1(\mathbb{X}\times\mathbb{T}^1_{2\pi}\times[0,+\infty)\times[-1,1])$ then also $\mu_t\in L^1(\mathbb{X}\times\mathbb{T}^1_{2\pi}\times[0,+\infty)\times[-1,1])$ and the following properties hold

{\small\begin{align}
\label{distanzadecr}\int_{\mathbb{X}}dx\int_{\mathbb{T}^1_{2\pi}}d\theta\int_0^{\infty}ds\int_{-1}^1dh|\mu_t(x,\theta,s,h)|&\leq\int_{\mathbb{X}}dx\int_{\mathbb{T}^1_{2\pi}}d\theta\int_0^{\infty}ds\int_{-1}^1dh|\mu_0(x,\theta,s,h)|,
\\
\label{consmassa}\int_{\mathbb{X}}dx\int_{\mathbb{T}^1_{2\pi}}d\theta\int_0^{\infty}ds\int_{-1}^1dh\mu_t(x,\theta,s,h)&=\int_{\mathbb{X}}dx\int_{\mathbb{T}^1_{2\pi}}d\theta\int_0^{\infty}ds\int_{-1}^1dh\mu_0(x,\theta,s,h).
\end{align}}

Finally it holds

\begin{align}
\nonumber\|\mu_t\|_{L^p(\mathbb{X}\times\mathbb{T}^1_{2\pi}\times[0,T]\times[-1,1])}&\leq C\|\mu_0\|_{L^p(\mathbb{R}^2\times\mathbb{T}^1_{2\pi}\times[0,t+T]\times[-1,1])}
\\
\label{normaLp}&+C\int_{\mathbb{T}^1_{2\pi}}d\theta'\int_0^tdt'\int_{-1}^1dh'\|\mu_0(\cdot,\theta',t',h')\|_{L^p(\mathbb{X})},
\end{align}

and therefore, if $\mu_0\in L^p(\mathbb{X}\times\mathbb{T}^1_{2\pi}\times[0,+\infty)\times[-1,1])$, also $\mu_t\in L^p(\mathbb{X}\times\mathbb{T}^1_{2\pi}\times[0,+\infty)\times[-1,1])$ at any time $t>0$ and

\begin{align}
\nonumber\|\mu_t\|_{L^p(\mathbb{X}\times\mathbb{T}^1_{2\pi}\times[0,+\infty)\times[-1,1])}&\leq C\|\mu_0\|_{L^p(\mathbb{R}^2\times\mathbb{T}^1_{2\pi}\times[0,+\infty)\times[-1,1])}
\\
\label{normaLpfinita}&+C\int_{\mathbb{T}^1_{2\pi}}d\theta'\int_0^tdt'\int_{-1}^1dh'\|\mu_0(\cdot,\theta',t',h')\|_{L^p(\mathbb{X})}.
\end{align}

Moreover we have:
\begin{itemize}
\item if $\mu_0\geq0$ then $\mu_t\geq0$ at any time $t>0$,
\item if $\mu_0$ does not depend on $x$ or $(x,\theta)$ neither $\mu_t$ does at any time $t>0$, and the same properties hold without integrating with respect to $x$ or $(x,\theta)$.
\end{itemize}
\end{proposition}
Before proving Proposition \ref{prop:es/un} we recall that the existence and the uniqueness of the mild solutions have already been proved in \cite{CG2010,MS2008bgl}, as well as the fact that the $L^1$ distance between two solutions is not increasing in time. Nevertheless, for self consistency of the paper we prove it, both because our proof is slightly different and because the main Lemma we need to prove it will be also useful for other results in this paper.

Notice also that the quantity in \eqref{normaLpfinita} ensures that $\mu_t\in L^p$, since for any $p\in[1,+\infty]$ using H\"older inequality we get

\begin{align*}
\int_{\mathbb{T}^1_{2\pi}}d\theta'\int_0^tdt'\int_{-1}^1dh'\|\mu_0(\cdot,\theta',t',h')\|_{L^p(\mathbb{R}^2)}\leq(4\pi t)^{1-\frac{1}{p}}\|\mu_0\|_{L^p(\mathbb{R}^2\times\mathbb{T}^1_{2\pi}\times[0,t]\times[-1,1])}.
\end{align*}

Of course, a priori this estimate does not ensure that the family $\{\mu_t\}_{t\geq0}$ is bounded in $L^p$ if $p\neq1$.

Finally, we point out that if $\mu_0$ does not depend on $x$ (respectively on $(x,\theta)$), the second summand in both \eqref{normaLp} and \eqref{normaLpfinita} can be expressed as $\|\mu_0\|_{L^1(\mathbb{T}^1_{2\pi}\times[0,t]\times[-1,1])}$ (respectively as $\|\mu_0\|_{L^1([0,t]\times[-1,1])}$).
\newline
\subsubsection{An intermediate step to prove Proposition \ref{prop:es/un}.}
To prove Proposition \ref{prop:es/un} we need an intermediate result. If we look at the equation \eqref{rapp:ev} and we evaluate $\mu_t$ for $s=0$, we get an equation that involves only $\mu_t(x,\theta,s=0,h)$ and $\mu_0$, that is

{\small\begin{align*}
\mu_t(x,\theta,0,h)&=\mu_0(x-v(\theta)t,\theta,t,h)
\\
&+\int_0^tdt'\int_{-1}^1dh'Q(t-t',h|h')\mu_{t'}(x-(t-t')v(\theta),\theta+\pi-2\arcsin(h'),0,h').
\end{align*}}

Therefore the crucial point is proving the existence and the uniqueness of $\mu_t(x,\theta,0,h)$, since it implies the ones of $\mu_t(x,\theta,s,h)$ by \eqref{rapp:ev}.

Therefore we first prove the following Lemma, where $\mu_t(x,\theta,0,h)$ is replaced by $\rho(x,\theta,t,h)$ and $\mu_0(x,\theta,t,h)$ is replaced by $\mu(x,\theta,t,h)$. We write it in a more general way, since we will need it in a slightly different formulation.
\begin{lemma}\label{lemma:rho}
Let $\mathbb{X}=\mathbb{R}^2$ or $\mathbb{X}=\mathbb{T}^2=\mathbb{R}^2/\mathbb{Z}^2$, $T>0$, $p\in[1,+\infty]$ and $\mu\in L^p(\mathbb{R}^2\times\mathbb{T}^1_{2\pi}\times[0,T]\times[-1,1])$. Then there exists a unique $\rho\in L^p(\mathbb{R}^2\times\mathbb{T}^1_{2\pi}\times[0,T]\times[-1,1])$ such that

{\small\begin{align}
\nonumber\rho(x,\theta,t,h)&=\mu(x-tv(\theta),\theta,t,h)
\\
\label{proprho}&+\int_0^tdt'\int_{-1}^1dh'Q(t-t',h|h')\rho(x-(t-t')v(\theta),\theta+\pi-2\arcsin(h'),t',h').
\end{align}}

Moreover for any $t\leq T$ we have

{\footnotesize\begin{align}
\label{propr1}&\int_{\mathbb{X}}dx\int_{\mathbb{T}^1_{2\pi}}d\theta \int_0^tdt'\int_{-1}^1dh|\rho(x,\theta,t',h)|E(t-t',h)\leq \int_{\mathbb{X}}dx\int_{\mathbb{T}^1_{2\pi}}d\theta \int_0^tdt'\int_{-1}^1dh|\mu(x,\theta,t',h)|,
\\
\label{propr2}&\int_{\mathbb{X}}dx\int_{\mathbb{T}^1_{2\pi}}d\theta \int_0^tdt'\int_{-1}^1dh\rho(x,\theta,t',h)E(t-t',h)=\int_{\mathbb{X}}dx\int_{\mathbb{T}^1_{2\pi}}d\theta \int_0^tdt'\int_{-1}^1dh\mu(x,\theta,t',h).
\end{align}}

Moreover, there exists a constant $C>0$ such that

\begin{align}
\nonumber\|\rho(\cdot,\theta,t,h)\|_{L^p(\mathbb{X})}&\leq\|\mu(\cdot,\theta,t,h)\|_{L^p(\mathbb{X})}
\\
\nonumber&+\int_0^tdt'\int_{-1}^1dh'Q(t-t',h|h')\|\mu(\cdot,\theta+\pi-2\arcsin(h'),t',h')\|_{L^p(\mathbb{X})}
\\
\label{normaLprho}&+C\int_{\mathbb{T}^1_{2\pi}}d\theta'\int_0^tdt'\int_{-1}^1dh'\|\mu(\cdot,\theta',t',h')\|_{L^p(\mathbb{X})}.
\end{align}

Furthermore, the following properties hold:
\begin{itemize}
\item if $\mu\geq0$ then $\rho\geq0$,
\item if $\mu$ does not depend on $x$ or $(x,\theta)$ neither does $\rho$. In such case, \eqref{propr1} and \eqref{propr2} hold also without integrating with respect to $x$ or $(x,\theta)$.
\end{itemize}
\end{lemma}
Before proving the Lemma, we point out that if $\mu$ does not depend on $x$ or $(x,\theta)$, \eqref{normaLprho} can be expressed by substituting the norm $\|\mu(\cdot,\theta,s,h)\|_{L^p(\mathbb{X})}$ with the modulus $|\mu(\theta,s,h)|$ or $|\mu(s,h)|$.
\begin{proof} Since the proof technique does not depend on whenever $\mathbb{X}=\mathbb{R}^2$ or $\mathbb{X}=\mathbb{T}^2$, $p\in[1,+\infty]$, we prove this result only in the case $\mathbb{X}=\mathbb{R}^2$ and $p\in[1,+\infty)$.

To prove the existence of such $\rho$ in the equation \eqref{proprho} we use the Banach fixed-point Theorem being careful that the integral of the function $Q(t,h|h')$ in the domain $\{(t,h):0\leq t\leq T,h\in[-1,1]\}$ is exactly 1 (and not $<1$) as soon as $T\geq\frac{1}{1-|h'|}$ (see \eqref{sptE}). Therefore, splitting $\mathbb{R}^2\times\mathbb{T}^1_{2\pi}\times[0,+\infty)\times[-1,1]$ into $\left\{\mathbb{R}^2\times\mathbb{T}^1_{2\pi}\times[\frac{k}{2},\frac{k+1}{2})\times[-1,1]\right\}_{k\in\mathbb{N}}$, we define

\begin{align*}
{\mathcal M}^k:=\mathbb{R}^2\times\mathbb{T}^1_{2\pi}\times\left[\frac{k}{2},\frac{k+1}{2}\wedge T\right)\times[-1,1],
\end{align*}

and we find step by step $\rho$ as a function $L^1({\mathcal M}^k)$, for $ k=0,1,\dots,\lfloor2T\rfloor$, going on by induction on $k$. We only prove the inductive step because it has no substantial differences from the basic step $k=0$. Therefore, for $t\in[\frac{k}{2},\frac{k+1}{2})$, we write the equation \eqref{proprho} splitting the integral with respect to $t'$ in two summands

{\small\begin{align*}
\rho(x,\theta,t,h)&=\mu(x-tv(\theta),\theta,t,h)
\\
&+\int_0^{\frac{k}{2}}dt'\int_{-1}^1dh'Q(t-t',h|h')\rho(x-(t-t')v(\theta),\theta+\pi-2\arcsin(h'),t',h')
\\
&+\int_{\frac{k}{2}}^tdt'\int_{-1}^1dh'Q(t-t',h|h')\rho(x-(t-t')v(\theta),\theta+\pi-2\arcsin(h'),t',h'),
\end{align*}}

that is, supposing that we have already defined $\bar\rho$ as the solution of \eqref{proprho} for any $t\in[0,\frac{k}{2})$, we denote

{\small\begin{align*}
{\mathcal F}&:L^p({\mathcal M}^k)\to L^p({\mathcal M}^k)
\\
{\mathcal F}[\rho](x,\theta,t,h)&:=\mu(x-tv(\theta),\theta,t,h)
\\
&+\int_0^{k/2}dt'\int_{-1}^1dh' Q(t-t',h|h')\bar\rho(x-(t-t')v(\theta),\theta+\pi-2\arcsin(h'),t',h')
\\
&+\int_{k/2}^tdt'\int_{-1}^1dh' Q(t-t',h|h')\rho(x-(t-t')v(\theta),\theta+\pi-2\arcsin(h'),t',h').
\end{align*}}

This map preserves the periodicity in $\theta$ and therefore it makes sense for $\theta\in\mathbb{T}^1_{2\pi}$. It preserves also the periodicity with respect to $x$, and therefore the argument works also when $\mathbb{X}=\mathbb{T}^2$.

${\mathcal F}$ is a contraction with respect to the canonical norm associated with $L^p({\mathcal M}^k)$. Indeed by Jensen's inequality \eqref{Jensencons}, one has

{\small\begin{align*}
&\|{\mathcal F}[\rho]-{\mathcal F}[\rho']\|_{L^p({\mathcal M}^k)}^p
\\
&=\int_{\mathbb{R}^2}dx\int_{\mathbb{T}^1_{2\pi}}d\theta\int_{\frac{k}{2}}^{\frac{k+1}{2}\wedge T}dt\int_{-1}^1dh\left|\int_{k/2}^tdt'\int_{-1}^1dh' Q(t-t',h|h')(\rho-\rho')(x',\theta',t',h')\right|^p,
\end{align*}}

with 
\[
x':=x-(t-t')v(\theta),\textrm{ and }\theta':=\theta+\pi-2\arcsin(h').
\]
Therefore

{\footnotesize\begin{align*}
&\|{\mathcal F}[\rho]-{\mathcal F}[\rho']\|_{L^p({\mathcal M}^k)}^p
\\
&=\int_{\mathbb{R}^2}dx\int_{\mathbb{T}^1_{2\pi}}d\theta\int_{\frac{k}{2}}^{\frac{k+1}{2}\wedge T}dt\int_{-1}^1dh\left|\int_{k/2}^tdt'\int_{-1}^1dh' Q(t-t',h|h')(\rho-\rho')(x',\theta',t',h')\right|^p
\\
&\leq\int_{\mathbb{R}^2}dx\int_{\mathbb{T}^1_{2\pi}}d\theta\int_{\frac{k}{2}}^{\frac{k+1}{2}\wedge T}dt\int_{-1}^1dh\int_{k/2}^tdt'\int_{-1}^1dh' Q(t-t',h|h')|\rho-\rho'|^p(x',\theta',t',h')\xi_p,
\end{align*}}

with
\[
\xi_p=\xi_p(t,h)=(1-E(t-\frac{k}{2},h))^{p-1}.
\]
Now if one defines

\begin{align*}
C_E:=\min \left\{E(t,h):t\in\left[0,\frac{1}{2}\right],h\in[-1,1]\right\}>0,
\end{align*}

in such a way to get by the previous computations

{\tiny\begin{align*}
&\|{\mathcal F}[\rho]-{\mathcal F}[\rho']\|_{L^p({\mathcal M}^k)}^p
\\
&\leq(1-C_E)^{p-1}\int_{\mathbb{R}^2}dx\int_{\mathbb{T}^1_{2\pi}}d\theta\int_{\frac{k}{2}}^{\frac{k+1}{2}\wedge T}dt'\int_{-1}^1dh'|\rho-\rho'|^p(x,\theta,t',h')\int_{t'}^{\frac{k+1}{2}\wedge T}dt\int_{-1}^1dh Q(t-t',h|h')
\\
&=(1-C_E)^{p-1}\int_{\mathbb{R}^2}dx\int_{\mathbb{T}^1_{2\pi}}d\theta\int_{\frac{k}{2}}^{\frac{k+1}{2}\wedge T}dt'\int_{-1}^1dh'|\rho-\rho'|^p(x,\theta,t',h')\left(1-E\left(\frac{k+1}{2}\wedge T-t',h'\right)\right)
\\
&\leq(1-C_E)^p\|\rho-\rho'\|_{L^p({\mathcal M}^k)}^p,
\end{align*}}

one can use the Banach fixed-point Theorem on $\mathcal{F}$ to extend $\rho$ to ${\mathcal M}^k$.

The uniqueness is also a direct consequence of the Banach fixed-point Theorem applied step by step to $L^p({\mathcal M}^k)$.

To prove now inequality \eqref{propr1} we simply apply the triangle inequality in the integral defining $\rho$ in \eqref{proprho} and for any $t\leq T$ we get

{\small\begin{align*}
|\rho(x,\theta,t,h)|&\leq|\mu(x-tv(\theta),\theta,t,h)|
\\
&+\int_0^tdt'\int_{-1}^1dh' Q(t-t',h|h')|\rho(x-(t-t')v(\theta),\theta+\pi-2\arcsin(h'),t',h')|,
\end{align*}}

Therefore, we integrate both sides of the above inequality with respect to $(x,\theta,t,h)$. In this way, in the right-hand side one can change variables
\[
x-tv(\theta)\to x,
\]
in the first summand and
\[
x-(t-t')v(\theta)\to x,\theta+\pi-2\arcsin(h')\to\theta,
\]
in the second one, we get

\begin{align}
\nonumber&\int_{\mathbb{R}^2}dx\int_{\mathbb{T}^1_{2\pi}}d\theta\int_0^tdt'\int_{-1}^1dh|\rho(x,\theta,t',h)|
\\
\nonumber&\leq\int_{\mathbb{R}^2}dx\int_{\mathbb{T}^1_{2\pi}}d\theta\int_0^tdt'\int_{-1}^1dh\left|\mu(x,\theta,t',h)\right|
\\
\label{integraleE}&+\int_{\mathbb{R}^2}dx\int_{\mathbb{T}^1_{2\pi}}d\theta\int_0^tdt'\int_{-1}^1dh\int_0^{t'}dt''\int_{-1}^1dh' Q(t'-t'',h|h')|\rho(x,\theta,t'',h')|.
\end{align}

Since the second summand \eqref{integraleE} in the right-hand side of the above inequality can also be written as

\begin{align*}
\eqref{integraleE}&=\int_{\mathbb{R}^2}dx\int_{\mathbb{T}^1_{2\pi}}d\theta\int_0^tdt''\int_{-1}^1dh'|\rho(x,\theta,t'',h')|\int_{t''}^tdt'\int_{-1}^1dhQ(t'-t'',h|h')
\\
&=\int_{\mathbb{R}^2}dx\int_{\mathbb{T}^1_{2\pi}}d\theta\int_0^tdt'\int_{-1}^1dh'|\rho(x,\theta,t',h')|(1-E(t-t',h')),
\end{align*}

the term in the left-hand side of the inequality containing \eqref{integraleE} gets deleted and we obtain

{\tiny\begin{align*}
\int_{\mathbb{R}^2}dx\int_{\mathbb{T}^1_{2\pi}}d\theta\int_0^tdt'\int_{-1}^1dh|\rho(x,\theta,t',h)|E(t-t',h')\leq\int_{\mathbb{R}^2}dx\int_{\mathbb{T}^1_{2\pi}}d\theta\int_0^tdt'\int_{-1}^1dh|\mu(x,\theta,t',h)|\quad\forall t\leq T.
\end{align*}}

The same steps, without the modulus and therefore without the triangle inequality, prove property \eqref{propr2}.

Before proving \eqref{normaLprho}, we focus on the last statement of the Theorem: $\mu\geq0$ implies $\rho\geq0$ because the map ${\mathcal F}$ preserves also the positivity of the argument and the metric space of non negative $L^1$ functions is complete with respect to $L^1$ norm, therefore the Banach fixed-point Theorem can be applied in this space. The same holds for the dependance on $x$ and on $(x,\theta)$, where the same argument works.

We finally have to prove \eqref{normaLprho}. To this purpose, we split $\rho$ in the sum of three contributions, that is

{\scriptsize\begin{align}
\nonumber&\rho(x,\theta,t,h)
\\
\nonumber&=\beta(x,\theta,t,h)+\mu(x-tv(\theta),\theta,t,h)
\\
\nonumber&+\int_0^tdt'\int_{-1}^1dh'Q(t-t',h|h')\mu(x-(t-t')v(\theta)-t'v(\theta+\pi-2\arcsin(h')),\theta+\pi-2\arcsin(h'),t',h'),
\\
\label{defbeta}
\end{align}}

with $\beta$ satisfying

\begin{align}
\nonumber&\beta(x,\theta,t,h)
\\
\nonumber&=\int_0^tdt'\int_{-1}^1dh' Q(t-t',h|h')\int_0^{t'}dt''\int_{-1}^1dh''Q(t'-t'',h'|h'')\mu(x',\theta'',t'',h'')
\\
\label{proprbeta}&+\int_0^tdt'\int_{-1}^1dh' Q(t-t',h|h')\beta(x-(t-t')v(\theta),\theta+\pi-2\arcsin(h'),t',h'),
\end{align}

and where
{\small\[
x':=x-(t-t')v(\theta)-(t'-t'')v(\theta+\pi-2\arcsin(h'))-t''v(\theta+2\pi-2\arcsin(h')-2\arcsin(h'')),
\]}
and
\[
\theta'':=\theta+2\pi-2\arcsin(h')-2\arcsin(h'').
\]
By applying Minkowski's inequality to equation \eqref{defbeta}, one gets

{\small\begin{align}
\nonumber\|\rho(\cdot,\theta,t,h)\|_{L^p(\mathbb{R}^2)}&\leq\|\beta(\cdot,\theta,t,h)\|_{L^p(\mathbb{R}^2)}+\|\mu(\cdot,\theta,t,h)\|_{L^p(\mathbb{R}^2)}
\\
\label{stimaconbeta}&+\int_0^tdt'\int_{-1}^1dh'Q(t-t',h|h')\|\mu(\cdot,\theta+\pi-2\arcsin(h'),t',h')\|_{L^p(\mathbb{R}^2)},
\end{align}}

and therefore we can conclude the proof if we estimate $\|\beta(\cdot,\theta,t,h)\|_{L^p(\mathbb{R}^2)}$.

\textbf{Step 1:} we want to prove that if

\begin{align*}
G(\theta,t,h)&:=\int_{\mathbb{T}^1_{2\pi}}d\theta'\int_0^tdt'\int_{-1}^1dh'\|\mu(\cdot,\theta',t',h')\|_{L^p(\mathbb{R}^2)}f(\theta-\theta',t-t',h|h'),
\end{align*}

with $f$ as in Definition \ref{defeffe}, then

{\small\begin{align}
\nonumber& \|\beta(\cdot,\theta,t,h)\|_{L^p(\mathbb{R}^2)}
\\
\label{normaLpcorrezioni1}&\leq G(\theta,t,h)+\int_0^tdt'\int_{-1}^1dh' Q(t-t',h|h')G(\theta+\pi-2\arcsin(h'),t',h')
\\
\label{normaLpcorrezioni2}&+2\frac{c}{2\pi}\int_{\mathbb{T}^1_{2\pi}}d\theta\int_0^tdt'\int_{-1}^1dh'G(\theta,t',h')
\\
\nonumber&+\int_{\mathbb{T}^1_{2\pi}}d\theta'\int_0^tdt'\int_{-1}^1dh'\left|f(\theta-\theta',t-t',h|h')-\frac{c}{2\pi}E^{(2)}(t-t',h')\right|\|\beta(\cdot,\theta',t',h')\|_{L^p(\mathbb{R}^2)}.
\end{align}}

To get this, we first shorten
\[
F(\theta,t,h):=\|\beta(\cdot,\theta,t,h)\|_{L^p(\mathbb{R}^2)},
\]
and then we notice that by using Minkowski's inequality in the definition of $\beta$ \eqref{proprbeta}, changing then the integration order and shortening again
\[
\theta'':=\theta+2\pi-2\arcsin(h')-2\arcsin(h''), 
\]
we get 

{\small\begin{align}
\nonumber &F(\theta,t,h)
\\
\nonumber&\leq\left\|\int_0^tdt'\int_{-1}^1dh' Q(t-t',h|h')\int_0^{t'}dt''\int_{-1}^1dh''Q(t'-t'',h'|h'')\mu(\cdot,\theta'',t'',h'')\right\|_{L^p(\mathbb{R}^2)}
\\
\nonumber&+\left\|\int_0^tdt'\int_{-1}^1dh' Q(t-t',h|h')\beta(\cdot,\theta+\pi-2\arcsin(h'),t',h')\right\|_{L^p(\mathbb{R}^2)}
\\
\nonumber&\leq\int_0^tdt'\int_{-1}^1dh' Q(t-t',h|h')\int_0^{t'}dt''\int_{-1}^1dh''Q(t'-t'',h'|h'')\|\mu(\cdot,\theta'',t'',h'')\|_{L^p(\mathbb{R}^2)}
\\
\nonumber&+\int_0^tdt'\int_{-1}^1dh' Q(t-t',h|h')\|\beta(\cdot,\theta+\pi-2\arcsin(h'),t',h')\|_{L^p(\mathbb{R}^2)}
\\
\label{ref:rel}&=G(\theta,t,h)+\int_0^tdt'\int_{-1}^1dh' Q(t-t',h|h')F(\theta+\pi-2\arcsin(h'),t',h').
\end{align}}

Now we argue as in the proof of Lemma \ref{lemma:phimeglio}: first notice that by iterating twice over \eqref{ref:rel} we get

\begin{align}
\nonumber F(\theta,t,h)&\leq G(\theta,t,h)+\int_0^tdt'\int_{-1}^1dh' Q(t-t',h|h')G(\theta+\pi-2\arcsin(h'),t',h')
\\
\label{normaLprhostep1}&+\int_{\mathbb{T}^1_{2\pi}}d\theta'\int_0^tdt'\int_{-1}^1dh'f(\theta-\theta',t-t',h|h')F(\theta',t',h'),
\end{align}

with $f$ as in Definition \ref{defeffe}. Notice also that with the same steps as Lemma \ref{lemma:phimeglio}, that is, integrating both sides of \eqref{normaLprhostep1} over $\theta\in\mathbb{T}^1_{2\pi},t\in[0,T],h\in[-1,1]$, we get

\begin{align}
\nonumber&\int_{\mathbb{T}^1_{2\pi}}d\theta\int_0^Tdt\int_{-1}^1dhF(\theta,t,h)E^{(2)}(T-t,h)
\\
\nonumber&\leq\int_{\mathbb{T}^1_{2\pi}}d\theta\int_0^Tdt\int_{-1}^1dhG(\theta,t,h)
\\
\nonumber&+\int_{\mathbb{T}^1_{2\pi}}d\theta\int_0^Tdt'\int_{-1}^1dh'G(\theta,t',h')\underbrace{\int_{t'}^Tdt\int_{-1}^1dhQ(t-t',h|h')}_{\leq1}
\\
\label{normaLprhostep2}&\leq2\int_{\mathbb{T}^1_{2\pi}}d\theta\int_0^Tdt\int_{-1}^1dhG(\theta,t,h).
\end{align}

Now if we multiply both the sides of \eqref{normaLprhostep2} by $\frac{c}{2\pi}$, with $c\in(0,1)$, and if we subtract and add back the left-hand side into \eqref{normaLprhostep1}, we get \textbf{Step 1}.

\textbf{Step 2:} to estimate $F$, we begin with estimating all the three terms concerning $G$ in \eqref{normaLpcorrezioni1} and \eqref{normaLpcorrezioni2}.

Notice that by definition

\begin{align}
\nonumber G(\theta,t,h)&=\int_{\mathbb{T}^1_{2\pi}}d\theta'\int_0^tdt'\int_{-1}^1dh'\|\mu(\cdot,\theta',t',h')\|_{L^p(\mathbb{R}^2)}\underbrace{f(\theta-\theta',t-t',h|h')}_{\leq C\textrm{ by \eqref{deceffe} of Lemma \ref{propreffe}}}
\\
\label{normaLprhostep6}&\leq C\int_{\mathbb{T}^1_{2\pi}}d\theta'\int_0^tdt'\int_{-1}^1dh'\|\mu(\cdot,\theta',t',h')\|_{L^p(\mathbb{R}^2)},
\end{align}

and therefore also the second one has a very similar bound, indeed

\begin{align}
\nonumber&\int_0^tdt'\int_{-1}^1dh' Q(t-t',h|h')\underbrace{G(\theta+\pi-2\arcsin(h'),t',h')}_{\leq C\int_{\mathbb{T}^1_{2\pi}}d\theta'\int_0^{t'}dt''\int_{-1}^1dh''\|\mu(\cdot,\theta',t',h')\|_{L^p(\mathbb{R}^2)}\textrm{ by \eqref{normaLprhostep6}}}
\\
\nonumber&\leq C\int_{\mathbb{T}^1_{2\pi}}d\theta'\int_0^tdt''\int_{-1}^1dh''\|\mu(\cdot,\theta',t',h')\|_{L^p(\mathbb{R}^2)}\underbrace{\int_0^tdt'\int_{-1}^1dh' Q(t-t',h|h')}_{\leq1}
\\
\label{normaLprhostep7}&\leq C\int_{\mathbb{T}^1_{2\pi}}d\theta'\int_0^tdt''\int_{-1}^1dh''\|\mu(\cdot,\theta',t',h')\|_{L^p(\mathbb{R}^2)}.
\end{align}

That is, by \eqref{normaLprhostep6} and \eqref{normaLprhostep7}, we get

\begin{align}\label{normaLpcorrezioni3}
\eqref{normaLpcorrezioni1}\leq C\int_{\mathbb{T}^1_{2\pi}}d\theta'\int_0^tdt''\int_{-1}^1dh''\|\mu(\cdot,\theta',t',h')\|_{L^p(\mathbb{R}^2)}.
\end{align}

Finally, \eqref{normaLpcorrezioni2} can be bounded as

{\scriptsize\begin{align}
\nonumber\eqref{normaLpcorrezioni2}&=\int_{\mathbb{T}^1_{2\pi}}d\theta\int_0^tdt'\int_{-1}^1dh'\int_{\mathbb{T}^1_{2\pi}}d\theta'\int_0^{t'}dt''\int_{-1}^1dh'' f(\theta-\theta',t-t'',h'|h'')\|\mu(\cdot,\theta',t'',h'')\|_{L^p(\mathbb{R}^2)}
\\
\nonumber&=\int_{\mathbb{T}^1_{2\pi}}d\theta'\int_0^tdt''\int_{-1}^1dh''\|\mu(\cdot,\theta',t'',h'')\|_{L^p(\mathbb{R}^2)}\underbrace{\int_{\mathbb{T}^1_{2\pi}}d\theta\int_{t''}^tdt'\int_{-1}^1dh'f(\theta-\theta',t'-t'',h'|h'')}_{\leq1\textrm{ by \eqref{fint1} of Lemma \ref{propreffe}}}
\\
\label{normaLprhostep8}&\leq\int_{\mathbb{T}^1_{2\pi}}d\theta'\int_0^tdt''\int_{-1}^1dh''\|\mu(\cdot,\theta',t'',h'')\|_{L^p(\mathbb{R}^2)}.
\end{align}}

Therefore, by substituting \eqref{normaLpcorrezioni3} and \eqref{normaLprhostep8} into \textbf{Step 1}, we can write

{\small\begin{align}
\nonumber F(\theta,t,h)&\leq C\int_{\mathbb{T}^1_{2\pi}}d\theta'\int_0^tdt'\int_{-1}^1dh'\|\mu(\cdot,\theta',t',h')\|_{L^p(\mathbb{R}^2)}
\\
\label{normaLprhostep4}&+\int_{\mathbb{T}^1_{2\pi}}d\theta'\int_0^tdt'\int_{-1}^1dh'\left|f(\theta-\theta',t-t',h|h')-\frac{c}{2\pi}E^{(2)}(t-t',h')\right|F(\theta',t',h').
\end{align}}

\textbf{Step 3: } we prove here that, up to another constant $C'>0$, the following estimate holds:

{\tiny\begin{align*}
\|\beta(\cdot,\theta,t,h)\|_{L^p(\mathbb{R}^2)}&=F(\theta,t,h)
\\
&\leq C'\int_{\mathbb{T}^1_{2\pi}}d\theta'\int_0^tdt''\int_{-1}^1dh''\|\mu(\cdot,\theta',t',h')\|_{L^p(\mathbb{R}^2)},\quad\forall(\theta,t,h)\in\mathbb{T}^1_{2\pi}\times[0,+\infty)\times[-1,1].
\end{align*}}

To this purpose, also define

\begin{align*}
\tilde F(\theta,t,h):=\frac{F(\theta,t,h)}{\int_{\mathbb{T}^1_{2\pi}}d\theta'\int_0^tdt'\int_{-1}^1dh'\|\mu(\cdot,\theta',t',h')\|_{L^p(\mathbb{R}^2)}},
\end{align*}

and notice that the quantity above is well-defined. Indeed, since if the denominator is zero, then the same holds for the numerator. To prove this claim, suppose that $\int_{\mathbb{T}^1_{2\pi}}d\theta'\int_0^Tdt'\int_{-1}^1dh'\|\mu(\cdot,\theta',t',h')\|_{L^p(\mathbb{R}^2)}=0$, then by \eqref{normaLprhostep4} we have
{\scriptsize
\[
F(\theta,t,h)\leq\int_{\mathbb{T}^1_{2\pi}}d\theta'\int_0^tdt'\int_{-1}^1dh' \left|f(\theta-\theta',t-t',h|h')-\frac{c}{2\pi}E^{(2)}(t-t',h')\right|F(\theta',t',h'),\quad\forall t\leq T.
\]}
Therefore denoting, for $c$ small enough,
\[
d:=\sup_{h\in[-1,1]}\left\|f(\cdot,\cdot,h|\cdot)-\frac{c}{2\pi}E^{(2)}\right\|_{L^1(\mathbb{T}^1_{2\pi}\times[0,+\infty)\times[-1,1])}<1,
\]
the constant provided by Lemma \ref{lemmaok}, which we will prove Section \ref{thetash}, if $t\leq T$ one gets

{\scriptsize\begin{align*}
&F(\theta,t,h)\leq\|F\|_{L^{\infty}(\mathbb{T}^1_{2\pi}\times[0,T]\times[-1,1])}\int_{\mathbb{T}^1_{2\pi}}d\theta'\int_0^tdt'\int_{-1}^1dh' \left|f(\theta-\theta',t-t',h|h')-\frac{c}{2\pi}E^{(2)}(t-t',h')\right|,
\end{align*}}

then

\begin{align*}
\|F\|_{L^{\infty}(\mathbb{T}^1_{2\pi}\times[0,T]\times[-1,1])}\leq d\|F\|_{L^{\infty}(\mathbb{T}^1_{2\pi}\times[0,T]\times[-1,1])},
\end{align*}

and therefore 

\begin{align*}
\|F\|_{L^{\infty}(\mathbb{T}^1_{2\pi}\times[0,T]\times[-1,1])}=0\textrm{ since }d<1.
\end{align*}

Now that we know that $\tilde F$ is finite, by noticing that the denominator in the definition of $\tilde F$ is increasing with respect to $t$, dividing both sides of \eqref{normaLprhostep4} by $\int_{\mathbb{T}^1_{2\pi}}d\theta''\int_0^tdt''\int_{-1}^1dh''\|\mu(\cdot,\theta'',t'',h'')\|_{L^p(\mathbb{R}^2)}$, we get

{\footnotesize\begin{align*}
\tilde F(\theta,t,h)\leq C+\int_{\mathbb{T}^1_{2\pi}}d\theta'\int_0^tdt'\int_{-1}^1dh' \left|f(\theta-\theta',t-t',h|h')-\frac{c}{2\pi}E^{(2)}(t-t',h')\right|\tilde F(\theta',t',h').
\end{align*}}

Now let $T>0$: if $t\leq T$, using again $d:=\sup_{h\in[-1,1]}\|f(\cdot,\cdot,h|\cdot)-\frac{c}{2\pi}E^{(2)}\|_{L^1}<1$ provided by Lemma \ref{lemmaok}, by the previous inequality we have

{\tiny\begin{align*}
\tilde F(\theta,t,h)&\leq C+\|\tilde F\|_{L^{\infty}(\mathbb{T}^1_{2\pi}\times[0,T]\times[-1,1])}\int_{\mathbb{T}^1_{2\pi}}d\theta'\int_0^tdt'\int_{-1}^1dh' \left|f(\theta-\theta',t-t',h|h')-\frac{c}{2\pi}E^{(2)}(t-t',h')\right|
\\
&\leq C+\|\tilde F\|_{L^{\infty}(\mathbb{T}^1_{2\pi}\times[0,T]\times[-1,1])}d,
\end{align*}}

then
\[
\|\tilde F\|_{L^{\infty}(\mathbb{T}^1_{2\pi}\times[0,T]\times[-1,1])}\leq C+\|\tilde F\|_{L^{\infty}(\mathbb{T}^1_{2\pi}\times[0,T]\times[-1,1])}d,
\]
and therefore

{\small\[
\|\tilde F\|_{L^{\infty}(\mathbb{T}^1_{2\pi}\times[0,T]\times[-1,1])}\leq\frac{C}{1-d},\quad\forall T>0,\textrm{ that is, }\|\tilde F\|_{L^{\infty}(\mathbb{T}^1_{2\pi}\times[0,+\infty)\times[-1,1])}\leq\frac{C}{1-d}.
\]}

Hence, by the definition of $F$ and $\tilde F$, one gets
\[
\|\beta(\cdot,\theta,t,h)\|_{L^p(\mathbb{R}^2)}=F(\theta,t,h)\leq\frac{C}{1-d}\int_{\mathbb{T}^1_{2\pi}}d\theta'\int_0^tdt'\int_{-1}^1dh'\|\mu(\cdot,\theta',t',h')\|_{L^p(\mathbb{R}^2)},
\]
that is the content of \textbf{Step 3}. Combining the previous estimate with \eqref{stimaconbeta} we can conclude the proof of \eqref{normaLprho}.
\end{proof}
When studying the dependance on $x$, a slightly different version of Lemma \ref{lemma:rho} is needed:
\begin{lemma}\label{lemma:rhovarianteC}
Let $k\in\mathbb{R}^2,k\neq(0,0)$, $T>0$, $p\in[1,+\infty]$ and $\mu\in L^p(\mathbb{T}^1_{2\pi}\times[0,T]\times[-1,1];\mathbb{C})$. Then there exists a unique $\rho\in L^p(\mathbb{T}^1_{2\pi}\times[0,T]\times[-1,1];\mathbb{C})$ such that

{\footnotesize\begin{align*}
\rho(\theta,t,h)=\mu(\theta,t,h)+\int_0^tdt'\int_{-1}^1dh'Q(t-t',h|h')e^{2\pi i(t-t')k\cdot v(\theta)}\rho(\theta+\pi-2\arcsin(h'),t',h').
\end{align*}}

\end{lemma}
\subsubsection{Proof of Proposition \ref{prop:es/un}.}
\begin{proof}
We start by proving the existence and the uniqueness: Lemma \ref{lemma:rho} with $\mathbb{X}=\mathbb{R}^2$ or $\mathbb{X}=\mathbb{T}^2$, $p\in[1,+\infty]$ and $\mu=\mu_0$ provides the existence of a function $\rho$ satisfying \eqref{proprho}. Therefore, if we define $\mu_t(x,\theta,0,h)=\rho(x,\theta,t,h)$, and we then use equation \eqref{rapp:ev} to obtain $\mu_t(x,\theta,s,h)$, we get a mild solution of the equation \eqref{eq:ev}. The uniqueness is a consequence of the fact that the relation \eqref{rapp:ev} forces $\mu_t(x,\theta,0,h)$ to be a solution of \eqref{proprho}, which is unique thanks to Lemma \ref{lemma:rho}, and therefore also $\mu_t(x,\theta,s,h)$ is uniquely defined for any $s\geq0$.

The next step is to prove property \eqref{mutL1cpt}.

By representation formula \eqref{rapp:ev} we get

{\tiny\begin{align*}
&\int_{\mathbb{X}}dx\int_{\mathbb{T}^1_{2\pi}}d\theta\int_0^Tds\int_{-1}^1dh|\mu_t(x,\theta,s,h)|
\\
&\leq\int_{\mathbb{X}}dx\int_{\mathbb{T}^1_{2\pi}}d\theta\int_0^Tds\int_{-1}^1dh|\mu_0(x-tv(\theta),\theta,s+t,h)|
\\
&+\int_{\mathbb{X}}dx\int_{\mathbb{T}^1_{2\pi}}d\theta\int_0^Tds\int_{-1}^1dh\int_0^tdt'\int_{-1}^1dh' Q(s+t-t',h|h')|\mu_{t'}(x-(t-t')v(\theta),\theta+\pi-2\arcsin(h'),0,h')|
\\
&=\int_{\mathbb{X}}dx\int_{\mathbb{T}^1_{2\pi}}d\theta\int_t^{t+T}ds\int_{-1}^1dh|\mu_0(x,\theta,s,h)|
\\
&+\int_{\mathbb{X}}dx\int_{\mathbb{T}^1_{2\pi}}d\theta\int_0^Tds\int_{-1}^1dh\int_0^tdt'\int_{-1}^1dh' Q(s+t-t',h|h')|\mu_{t'}(x,\theta,0,h')|
\\
&=\int_{\mathbb{X}}dx\int_{\mathbb{T}^1_{2\pi}}d\theta\int_t^{t+T}ds\int_{-1}^1dh|\mu_0(x,\theta,s,h)|
\\
&+\int_{\mathbb{X}}dx\int_{\mathbb{T}^1_{2\pi}}d\theta\int_0^tdt'\int_{-1}^1dh' |\mu_{t'}(x,\theta,0,h')|\int_0^Tds\int_{-1}^1dhQ(s+t-t',h|h').
\end{align*}}

Now we look at the last term, and since
\[
\int_0^Tds\int_{-1}^1dhQ(s+t-t',h|h')=E(t-t',h')-E(t+T-t',h')\leq E(t-t',h'),
\]
we can write

\begin{align*}
&\int_{\mathbb{X}}dx\int_{\mathbb{T}^1_{2\pi}}d\theta\int_0^Tds\int_{-1}^1dh|\mu_t(x,\theta,s,h)|
\\
&\leq\int_{\mathbb{X}}dx\int_{\mathbb{T}^1_{2\pi}}d\theta\int_t^{t+T}ds\int_{-1}^1dh|\mu_0(x,\theta,s,h)|
\\
&+\underbrace{\int_{\mathbb{X}}dx\int_{\mathbb{T}^1_{2\pi}}d\theta\int_0^tdt'\int_{-1}^1dh' |\mu_{t'}(x,\theta,0,h')|E(t-t',h')}_{\leq\int_{\mathbb{X}}dx\int_{\mathbb{T}^1_{2\pi}}d\theta\int_0^tdt'\int_{-1}^1dh' |\mu_0(x,\theta,t',h')|\textrm{ by }\eqref{propr1}\textrm{ of Lemma \ref{lemma:rho}}}
\\
&\leq\int_{\mathbb{X}}dx\int_{\mathbb{T}^1_{2\pi}}d\theta\int_t^{t+T}ds\int_{-1}^1dh|\mu_0(x,\theta,s,h)|
\\
&+\int_{\mathbb{X}}dx\int_{\mathbb{T}^1_{2\pi}}d\theta\int_0^tdt'\int_{-1}^1dh' |\mu_0(x,\theta,t',h')|
\\
&=\int_{\mathbb{X}}dx\int_{\mathbb{T}^1_{2\pi}}d\theta\int_0^{t+T}ds\int_{-1}^1dh|\mu_0(x,\theta,s,h)|,
\end{align*}

and this proves property \eqref{mutL1cpt}.

As for \eqref{distanzadecr}, it is sufficient to take the limit $T\to+\infty$ in the equation \eqref{mutL1cpt}.

Then, to prove property \eqref{consmassa}, the same argument as the one used to prove property \eqref{mutL1cpt} works: the only differences are that \eqref{propr2} should be used instead of  \eqref{propr1} and that all the computations have to be meant in the limit $T\to+\infty$.

The positivity of $\mu_t$ is preserved in time since Lemma \ref{lemma:rho} ensures that $\mu_0\geq0$ a.e. implies $\rho(x,\theta,t,h)=\mu_t(x,\theta,0,h)\geq0$ a.e. and therefore thanks to \eqref{rapp:ev} we have also $\mu_t\geq0$ a.e..

Moreover, if $\mu_0$ does not depend on $x$ or on $(x,\theta)$ neither $\rho(x,\theta,t,h)=\mu_t(x,\theta,0,h)$ does thanks to Lemma \ref{lemma:rho}. Therefore, thanks to the representation formula \eqref{rapp:ev}, the same holds for $\mu_t$.

Now we only have to prove \eqref{normaLp}. To this purpose, if we use the representation \eqref{rapp:ev}, with triangle inequality and by extending the $L^p$ norm of the second summand to all $s\in[0,+\infty)$, we get

{\scriptsize\begin{align}
\label{normaLpstep1}\|\mu_t\|_{L^p(\mathbb{X}\times\mathbb{T}^1_{2\pi}\times[0,T]\times[-1,1])}&\leq\|\mu_0(\cdot,\cdot,\cdot+t,\cdot)\|_{L^p(\mathbb{X}\times\mathbb{T}^1_{2\pi}\times[0,T]\times[-1,1])}
\\
\label{normaLpstep2}&+\left\|\int_0^tdt'\int_{-1}^1dh'Q(\cdot+t-t',\cdot|h')\mu_{t'}(\cdot,\cdot,0,h')\right\|_{L^p(\mathbb{X}\times\mathbb{T}^1_{2\pi}\times[0,+\infty)\times[-1,1])}.
\end{align}}

Now since the first term in \eqref{normaLpstep1} can be expressed as
\[
\eqref{normaLpstep1}=\|\mu_0(\cdot,\cdot,\cdot+t,\cdot)\|_{L^p(\mathbb{X}\times\mathbb{T}^1_{2\pi}\times[0,T]\times[-1,1])}=\|\mu_0\|_{L^p(\mathbb{X}\times\mathbb{T}^1_{2\pi}\times[t,T+t]\times[-1,1])},
\]
we focus on the second summand. Since
\[
\|\cdot\|_{L^p(\mathbb{X}\times\mathbb{T}^1_{2\pi}\times[0,+\infty)\times[-1,1])}=\left\|\|\cdot\|_{L^p(\mathbb{X})}\right\|_{L^p(\mathbb{T}^1_{2\pi}\times[0,+\infty)\times[-1,1])},
\]
using again Minkowski's inequality, we can write

\begin{align*}
\eqref{normaLpstep2}&=\left\|\int_0^tdt'\int_{-1}^1dh'Q(\cdot+t-t',\cdot|h')\mu_{t'}(\cdot,\cdot,0,h')\right\|_{L^p(\mathbb{X}\times\mathbb{T}^1_{2\pi}\times[0,+\infty)\times[-1,1])}
\\
&\leq\left\|\int_0^tdt'\int_{-1}^1dh'Q(\cdot+t-t',\cdot|h')\|\mu_{t'}(\cdot,\cdot,0,h')\|_{L^p(\mathbb{X})}\right\|_{L^p(\mathbb{T}^1_{2\pi}\times[0,+\infty)\times[-1,1])},
\end{align*}

and using \eqref{normaLprho} of Lemma \ref{lemma:rho} and applying triangle inequality to the external norm $\|\|_{L^p(\mathbb{T}^1_{2\pi}\times[0,+\infty)\times[-1,1])}$, we have

\begin{align*}
\eqref{normaLpstep2}&\leq\left\|\int_0^tdt'\int_{-1}^1dh'Q(\cdot+t-t',\cdot|h')\left[\|\mu_0(\cdot,\cdot,t',h')\|_{L^p(\mathbb{X})}\right.\right.
\\
&+\int_0^{t'}dt''\int_{-1}^1dh''Q(t'-t'',h'|h'')\|\mu_0(\cdot,\cdot+\pi-2\arcsin(h''),t'',h'')\|_{L^p(\mathbb{X})}
\\
&+C\left.\left.\int_{\mathbb{T}^1_{2\pi}}d\theta'\int_0^{t'}dt''\int_{-1}^1dh''\|\mu_0(\cdot,\theta',t'',h'')\|_{L^p(\mathbb{X})}\right]\right\|_{L^p(\mathbb{T}^1_{2\pi}\times[0,+\infty)\times[-1,1])}.
\end{align*}

By triangle inequality we get

{\footnotesize\begin{align}
\nonumber&\eqref{normaLpstep2}
\\
\label{normaLpstep3}&\leq\left\|\int_0^tdt'\int_{-1}^1dh'Q(\cdot+t-t',\cdot|h')\|\mu_0(\cdot,\cdot,t',h')\|_{L^p(\mathbb{X})}\right\|_{L^p}
\\
\label{normaLpstep5}&+\left\|\int_0^tdt'\int_{-1}^1dh'Q(\cdot+t-t',\cdot|h')\int_0^{t'}dt''\int_{-1}^1dh'' Q(t'-t'',h'|h'')\|\mu_0(\cdot,\theta',t'',h'')\|_{L^p(\mathbb{X})}\right\|_{L^p}
\\
\label{normaLpstep4}&+C\left\|\int_0^tdt'\int_{-1}^1dh'Q(\cdot+t-t',\cdot|h')\int_0^{t'}dt''\int_{-1}^1dh''\|\mu_0(\cdot,\cdot,t'',h'')\|_{L^p(\mathbb{X})}\right\|_{L^p},
\end{align}}

where in the previous inequality we shortened
\[
L^p:=L^p(\mathbb{T}^1_{2\pi}\times[0,+\infty)\times[-1,1]),
\]
and
\[
\theta':=\theta+\pi-2\arcsin(h'').
\]
Hereafter, we will assume that $p$ is finite, but the case $p=\infty$ can be studied in the same way.

By using Jensen's inequality in the inner integral as in \eqref{Jensencons} and changing the integration order, one has

{\tiny\begin{align*}
&\eqref{normaLpstep3}
\\
&\leq\left(\int_{\mathbb{T}^1_{2\pi}}d\theta\int_0^{\infty}ds\int_{-1}^1dh\int_0^tdt'\int_{-1}^1dh'Q(s+t-t',h|h')\|\mu_0(\cdot,\theta,t',h')\|_{L^p(\mathbb{X})}^p\underbrace{(E(s,h)-E(s+t,h))^{p-1}}_{\leq1}\right)^{\frac{1}{p}}
\\
&=\left(\int_{\mathbb{T}^1_{2\pi}}d\theta\int_0^tdt'\int_{-1}^1dh'\|\mu_0(\cdot,\theta,t',h')\|_{L^p(\mathbb{X})}^p\underbrace{\int_0^{\infty}ds\int_{-1}^1dhQ(s+t-t',h|h')}_{=E(t-t',h')\leq1}\right)^{\frac{1}{p}}
\\
&=\|\mu_0\|_{L^p(\mathbb{X}\times\mathbb{T}^1_{2\pi}\times[0,t]\times[-1,1])},
\end{align*}}

and, first applying again twice Jensen's inequality as in \eqref{Jensencons}, and then changing the integration order, one has also

{\tiny\begin{align*}
\eqref{normaLpstep5}^p\leq\int d\theta dsdh\int_0^tdt'\int_{-1}^1dh'Q(s+t-t',h|h')\int_0^{t'}dt''\int_{-1}^1dh''Q(t'-t'',h'|h'')\|\mu_0(\cdot,\theta,t'',h'')\|_{L^p(\mathbb{X})}^p\xi,
\end{align*}}

with
\[
\xi=\xi(t,s,h,t',h'):=\underbrace{(1-E(t',h'))^{p-1}}_{\leq1}\underbrace{(E(s,h)-E(s+t,h))^{p-1}}_{\leq 1})^{\frac{1}{p}}\leq1,
\]
and
\[
\int d\theta dsdh:=\int_{\mathbb{T}^1_{2\pi}}d\theta\int_0^{\infty}ds\int_{-1}^1dh.
\]
Therefore

{\small\begin{align*}
\eqref{normaLpstep5}\leq\left(\int_{\mathbb{T}^1_{2\pi}}d\theta\int_0^tdt''\int_{-1}^1dh''\|\mu_0(\cdot,\theta,t'',h'')\|_{L^p(\mathbb{X})}^p\int_{t''}^tdt'\int_{-1}^1dh'Q(t'-t'',h'|h'')\xi_1\right)^{\frac{1}{p}}
\end{align*}}

with
\[
\xi_1=\xi_1(t,t',h'):=\int_0^{\infty}ds\int_{-1}^1dhQ(s+t-t',h|h')=E(t-t',h')\leq1.
\]
Hence

{\footnotesize\begin{align*}
\eqref{normaLpstep5}&\leq\left(\int_{\mathbb{T}^1_{2\pi}}d\theta\int_0^tdt''\int_{-1}^1dh''\|\mu_0(\cdot,\theta,t'',h'')\|_{L^p(\mathbb{X})}^p\underbrace{\int_{t''}^tdt'\int_{-1}^1dh'Q(t'-t'',h'|h'')}_{=E(t-t'',h'')\leq1}\right)^{\frac{1}{p}}
\\
&\leq\left(\int_{\mathbb{T}^1_{2\pi}}d\theta\int_0^tdt''\int_{-1}^1dh''\|\mu_0(\cdot,\theta,t'',h'')\|_{L^p(\mathbb{X})}^p\right)^{\frac{1}{p}}
\\
&=\|\mu_0\|_{L^p(\mathbb{X}\times\mathbb{T}^1_{2\pi}\times[0,t]\times\mathbb[-1,1])}.
\end{align*}}

Finally, one has

{\footnotesize\begin{align*}
\eqref{normaLpstep4}&=C\left\|\int_0^tdt'\int_{-1}^1dh'Q(\cdot+t-t',\cdot|h')\underbrace{\int_{\mathbb{T}^1_{2\pi}}d\theta'\int_0^{t'}dt''\int_{-1}^1dh''\|\mu_0(\cdot,\theta',t'',h'')\|_{L^p(\mathbb{X})}}_{\leq\int_{\mathbb{T}^1_{2\pi}}d\theta'\int_0^tdt''\int_{-1}^1dh''\|\mu_0(\cdot,\theta',t'',h'')\|_{L^p(\mathbb{X})}}\right\|_{L^p}
\\
&\leq C\int_{\mathbb{T}^1_{2\pi}}d\theta'\int_0^tdt''\int_{-1}^1dh''\|\mu_0(\cdot,\theta',t'',h'')\|_{L^p(\mathbb{X})}\left\|\underbrace{\int_0^tdt'\int_{-1}^1dh'Q(\cdot+t-t',\cdot|h')}_{=E-E(\cdot+t,\cdot)\leq E}\right\|_{L^p}
\\
&\leq(2\pi)^{\frac{1}{p}} C\int_{\mathbb{T}^1_{2\pi}}d\theta'\int_0^tdt''\int_{-1}^1dh''\|\mu_0(\cdot,\theta',t'',h'')\|_{L^p(\mathbb{X})},
\end{align*}}

where, in the first two lines of the previous inequality, we shortened
\[
L^p:=L^p(\mathbb{T}^1_{2\pi}\times[0,+\infty)\times[-1,1]).
\]

Summing \eqref{normaLpstep1} and \eqref{normaLpstep2}, which is smaller than \eqref{normaLpstep3}+\eqref{normaLpstep5}+\eqref{normaLpstep4}, we get \eqref{normaLp}.

Finally, \eqref{normaLpfinita} is obtained by sending $T\to+\infty$ in \eqref{normaLp}.
\end{proof}
\subsection{Stationary solutions.} Now we briefly focus on the stationary solutions of the equations \eqref{eq:ev}, \eqref{eq:ev_x} and \eqref{eq:ev_x,v}. It is straightforward to see that the only stationary solutions (in $L^1$) of the three equations are:
\begin{itemize}
\item $\mu_t(x,\theta,s,h)\equiv0$ for $x\in\mathbb{X}$, equation \eqref{eq:ev};
\item $\mu_t(x,\theta,s,h)\equiv\frac{c}{2\pi}E(s,h)$ for $x\in\mathbb{T}^2$, equation \eqref{eq:ev};
\item $\mu_t(\theta,s,h)\equiv\frac{c}{2\pi}E(s,h)$, equation \eqref{eq:ev_x};
\item $\mu_t(s,h)\equiv cE(s,h)$, equation \eqref{eq:ev_x,v}.
\end{itemize}
It is immediately checked that the previous mild solutions are stationary, because they correspond to $\mu_t(x,\theta,0,h)=c$, or $\mu_t(\theta,0,h)=c$ or $\mu_t(0,h)=c$, in the case of not dependance on $x$ or $(x,\theta)$. Following \cite{CG2010}, the reverse implication can also be proved: these solutions are the only stationary ones. An alternative way to prove it is exactly to use Theorems \ref{thm:convergenza_v,s,h}, \ref{thm:conv_x,v,s,h} and Theorem \ref{thm:conv_R2}, since if a stationary solution converges to another solution, either strongly or weakly, then it coincides with it.

\newpage
\section{The long time evolution of a density depending on $(\theta,s,h)$.}\label{thetash}
The aim of this Section is to prove Theorem \ref{thm:convergenza_v,s,h}, which states that if $\mu_0\in L^1\cap L^p(\mathbb{T}^1_{2\pi}\times[0,+\infty)\times[-1,1])$, with $p\in[1,+\infty]$, and $\mu_t$ is a solution of

{\small\begin{align}\label{rapp:ev_nuova}
\mu_t(\theta,s,h)=\mu_0(\theta,s+t,h)+\int_0^tdt'\int_{-1}^1dh'Q(s+t-t',h|h')\mu_{t'}(\theta+\pi-2\arcsin(h'),0,h'),
\end{align}}

then

{\tiny\begin{align*}
\left\|\mu_t-\frac{\langle\mu_0\rangle}{2\pi}E\right\|_{L^p}&\leq C\left[\frac{\|\mu_0\|_{L^1}+\|\mu_0\|_{L^p}}{t+1}+\|\mu_0\|_{L^1(\mathbb{T}^1_{2\pi}\times[t/4,+\infty)\times[-1,1])}+\|\mu_0\|_{L^p(\mathbb{T}^1_{2\pi}\times[t/4,+\infty)\times[-1,1])}\right],
\end{align*}}

up for a constant $C>0$ depending only on $Q$.

As we said in the introduction, the first term of \eqref{rapp:ev_nuova} represents the probability density of the particles for which no collision has occurred until time $t$. The second one can be understood as the probability density that at least a collision has occurred before time $t$. From equation \eqref{rapp:ev_nuova}, we can preliminary notice that the long time behavior of $\mu_t$ is determined  by that of $\mu_t(\theta,0,h)$, which is in turn determined by the equation

\begin{align}\label{eq:mut0}
\mu_t(\theta,0,h)=\mu_0(\theta,t,h)+\int_0^tdt'\int_{-1}^1dh'Q(t-t',h|h')\mu_{t'}(\theta+\pi-2\arcsin(h'),0,h').
\end{align}

Now if $\mu_t(\theta,0,h)=c$, by \eqref{eq:mut0} we get $\mu_0(\theta,t,h)=cE(t,h)$ and therefore by \eqref{rapp:ev_x} $\mu_t\equiv cE$ at any time $t$. Therefore the rate of convergence of $\mu_t(\theta,s,h)$ to the equilibrium state is determined by the rate of convergence of $\mu_t(\theta,0,h)$ to the constant $\frac{\langle\mu_0\rangle}{2\pi}$. This is our purpose in the next steps.

\subsection{Writing $\mu_t(\theta,0,h)$ as a linear function of $\mu_0$.}
We collect here some results that allow us to prove Theorem \ref{thm:convergenza_v,s,h}. In particular we want to write $\mu_t(\theta,0,h)$ as a linear function of $\mu_0$ to get a better estimate of its long time behavior. Therefore the first purpose is to prove now the following intermediate result.
\begin{proposition}\label{prop:musvarphi}
There exists a function $\varphi\in L^{\infty}(\mathbb{T}^1_{2\pi}\times[0,+\infty)\times[-1,1]^2;\mathbb{R})$ depending only on the kernel $Q$ such that for every $\mu_0\in L^1(\mathbb{T}^1_{2\pi}\times[0,+\infty)\times[-1,1])$ the function $(\theta,h)\mapsto\mu_t(\theta,0,h)$ writes as an affine function of $\varphi$ and a linear function of $\mu_0$ as

\begin{align}
\nonumber\mu_t(\theta,0,h)&=\mu_0(\theta,t,h)+\int_0^tdt'\int_{-1}^1dh'Q(t-t',h|h')\mu_0(\theta+\pi-2\arcsin(h'),t',h')
\\
\nonumber&+\frac{1}{2\pi}\int_{\mathbb{T}^1_{2\pi}}d\theta'\int_0^tdt'\int_{-1}^1dh'\mu_0(\theta',t',h')
\\
\label{prop:musvarphiineq}&+\int_{\mathbb{T}^1_{2\pi}}d\theta'\int_0^tdt'\int_{-1}^1dh'\varphi(\theta-\theta',t-t',h|h')\mu_0(\theta',t',h'),
\end{align}

and moreover there exists a constant $C>0$ such that

\begin{align*}
|\varphi(\theta,t,h|h')|\leq\frac{C}{t+1}\quad\forall(\theta,t,h|h')\in\mathbb{T}^1_{2\pi}\times[0,+\infty)\times[-1,1]^2.
\end{align*}

\end{proposition}
This Proposition is all what we need to prove Theorem \ref{thm:convergenza_v,s,h}. We observe that the most similar to $\frac{\langle\mu_0\rangle}{2\pi}$ is the third term in the right-hand side of \eqref{prop:musvarphiineq}. We will prove that the other terms vanish in the long time limit.

We are finally splitting the probability density $\mu_t(\theta,0,h)$ of a collision occurring at time $t$ in four contributions. The first and the second one are understood respectively as the probability that the first and the second collision occur exactly at time $t$. The sum of third and the fourth one as the probability that at least two collisions happened before time $t$.

To prove Proposition \ref{prop:musvarphi} we make several steps by some preliminary Lemmas. We first recall the Definition \ref{defeffe} of the function $f$, that is,

{\footnotesize\begin{align*}
f(\theta,t,h|h'):=\sum_{\ell\in\mathbb{Z}}\frac{\partial h''(\theta+2\ell\pi,h')}{\partial\theta}\int_0^tdt' Q(t-t',h|h''(\theta+2\ell\pi,h'))Q(t',h''(\theta+2\ell\pi,h')|h'),
\end{align*}}

with $h''$ as in Definition \ref{defh''}

\begin{align*}
h''(\theta,h'):=\sin\left(\frac{\theta+2\pi-2\arcsin(h')}{2}\right)\mathbbm{1}_{[2\arcsin(h')-3\pi,2\arcsin(h')-\pi)}(\theta).
\end{align*}

Lemma \ref{propreffe} in the Section \ref{app:funzionidiQ} establishes some properties of $f$: it has integral $1$ and it is bounded by $\frac{C}{t}$.

We shall start the proof of Proposition \ref{prop:musvarphi} by proving the following Lemma.
\begin{lemma}\label{lemma:esphi}
There exists a unique function $\varphi\in L^{\infty}_{loc}(\mathbb{T}^1_{2\pi}\times[0,+\infty)\times[-1,1]^2)$ such that

\begin{align}
\nonumber\varphi(\theta,t,h|h')&=f(\theta,t,h|h')-\frac{1}{2\pi}E(t,h)
\\
\label{espr:phi}&+\int_0^tdt'\int_{-1}^1dh''Q(t-t',h|h'')\varphi(\theta+\pi-2\arcsin(h''),t',h''|h'),
\end{align}

and moreover, for every $\mu_0\in L^1(\mathbb{T}^1_{2\pi}\times[0,+\infty)\times[-1,1])$, $\mu_t(\theta,0,h)$ is an affine function of $\varphi$ and a linear function of $\mu_0$:

\begin{align}
\nonumber\mu_t(\theta,0,h)&=\mu_0(\theta,t,h)+\int_0^tdt'\int_{-1}^1dh'Q(t-t',h|h')\mu_0(\theta+\pi-2\arcsin(h'),t',h')
\\
\nonumber&+\frac{1}{2\pi}\int_{\mathbb{T}^1_{2\pi}}d\theta'\int_0^tdt'\int_{-1}^1dh'\mu_0(\theta',t',h')
\\
\label{rapp:musphi}&+\int_{\mathbb{T}^1_{2\pi}}d\theta'\int_0^tdt'\int_{-1}^1dh'\varphi(\theta-\theta',t-t',h|h')\mu_0(\theta',t',h').
\end{align}

\end{lemma}
We recall that $\varphi$ is exactly the one in Proposition \ref{prop:musvarphi}. Therefore, once we will have proved it, we will only need to prove that, for large $t$, $\varphi$ behaves as $\frac{1}{t}$ at most.

We prove now Lemma \ref{lemma:esphi}
\begin{proof} We prove it in two steps.

\textbf{Step 1:} we start by looking for a function $\gamma:\mathbb{T}^1_{2\pi}\times[0,+\infty)\times[-1,1]^2\to[0,+\infty)$ such that for any $\mu_0$, the function $(\theta,h)\mapsto\mu_t(\theta,0,h)$ satisfies, for any $t>0,\theta,h$

\begin{align}
\nonumber\mu_t(\theta,0,h)&=\mu_0(\theta,t,h)+\int_0^tdt'\int_{-1}^1dh'Q(t-t',h|h')\mu_0(\theta+\pi-2\arcsin(h'),t',h')
\\
\label{rapp_gamma}&+\int_{\mathbb{T}^1_{2\pi}}d\theta'\int_0^tdt'\int_{-1}^1dh'\gamma(\theta-\theta',t-t',h|h')\mu_0(\theta',t',h').
\end{align}

We want to separate the dependance on the initial datum $\mu_0$, and the last term, that is the convolution of $\gamma$ and $\mu_0$, represents the probability that at least two collisions have occurred before time $t$.

To find such a function $\gamma$ we write $\mu_t$ as in the equation \eqref{rapp_gamma} and we substitute it in both sides of the equation \eqref{eq:mut0}. By doing this we get

{\scriptsize\begin{align}
\label{esprA0}&\mu_0(\theta,t,h)+\int_0^tdt'\int_{-1}^1dh'Q(t-t',h|h')\mu_0(\theta+\pi-2\arcsin(h'),t',h')
\\
\label{esprA}&+\int_{\mathbb{T}^1_{2\pi}}d\theta'\int_0^tdt'\int_{-1}^1dh'\gamma(\theta-\theta',t-t',h|h')\mu_0(\theta',t',h')
\\
\label{esprA0copia}&=\mu_0(\theta,t,h)+\int_0^tdt'\int_{-1}^1dh'Q(t-t',h|h')\mu_0(\theta+\pi-2\arcsin(h'),t',h')
\\
\label{esprB}&+\int_0^tdt'\int_{-1}^1dh'Q(t-t',h|h')\int_0^{t'}dt''\int_{-1}^1dh''Q(t'-t'',h'|h'')\mu_0(\theta'',t'',h'')
\\
\label{esprC}&+\int_0^tdt'\int_{-1}^1dh'Q(t-t',h|h')\int_{\mathbb{T}^1_{2\pi}}d\theta'\int_0^{t'}dt''\int_{-1}^1dh''\gamma(\tilde\theta-\theta',t'-t'',h'|h'')\mu_0(\theta',t'',h''),
\end{align}}

with
{\footnotesize\[
\theta''=\theta''(\theta,h',h''):=\theta+2\pi-2\arcsin(h')-2\arcsin(h'')\textrm{ and }\tilde\theta=\tilde\theta(\theta,h'):=\theta+\pi-2\arcsin(h').
\]}

The first two summands in \eqref{esprA0} and \eqref{esprA0copia} cancel the one each other, therefore we just have to write better the equality
\[
\eqref{esprA}=\eqref{esprB}+\eqref{esprC}.
\]
\eqref{esprA} is fine and does not need to be expressed in another way.

Instead we write better \eqref{esprB}: first we change the integration order, and then we rename the variables (exchange $t''$ with $t'$ and $h'$ with $h''$). By doing this we get

\begin{align*}
\eqref{esprB}=\int_0^tdt'\int_{-1}^1dh'\int_{-1}^1dh''\mu_0(t',h',\theta'')\int_0^{t-t'}dt''Q(t-t'-t'',h|h'')Q(t'',h''|h'),
\end{align*}

and now we change variables from $h''$ to $\theta'$ in such a way to get

\begin{align*}
\theta''=\theta+2\pi-2\arcsin(h')-2\arcsin(h'')=\theta',
\end{align*}

that is

\begin{align*}
h''=\sin\left(\frac{\theta-\theta'+2\pi-2\arcsin(h')}{2}\right)\mathbbm{1}_{2\arcsin(h')+[-3\pi,-\pi)}(\theta-\theta')=h''(\theta-\theta',h')
\end{align*}

as in Definition \ref{defh''}.

Therefore, shortening $h''=h''(\theta-\theta',h')$, we have obtained

{\scriptsize\begin{align}
\nonumber \eqref{esprB}&=\int_0^tdt'\int_{-1}^1dh'\int_{\theta+\pi-2\arcsin(h')}^{\theta+3\pi-2\arcsin(h')}d\theta'\mu_0(\theta',t',h')\frac{\partial h''}{\partial\theta}\int_0^{t-t'}dt''Q(t-t'-t'',h|h'')Q(t'',h''|h')
\\
\nonumber&=\int_0^tdt'\int_{-1}^1dh'\int_{\theta+\pi-2\arcsin(h')}^{\theta+3\pi-2\arcsin(h')}d\theta'\mu_0(\theta',t',h')f(\theta-\theta',t-t',h|h')\textrm{ by Definition \ref{defeffe} of }f
\\
\label{nuovaB}&=\int_{\mathbb{T}^1_{2\pi}}d\theta'\int_0^tdt'\int_{-1}^1dh'\mu_0(\theta',t',h')f(\theta-\theta',t-t',h|h'),
\end{align}}

since the integrand is periodic with respect to $\theta'$ and therefore we can choose any period to compute the integral.

As for \eqref{esprC} we operate similarly to \eqref{esprB}: first we change the integration order and we integrate before with respect to $(\theta',t'',h'')$ and then with respect to $(t',h')$, and then we exchange the variables names, that is, we exchange $t''$ with $t'$ and $h''$ with $h'$. In this way we get

 {\scriptsize\begin{align}
\label{nuovaC}\eqref{esprC}=\int_0^tdt'\int_{-1}^1dh'\int_{\mathbb{T}^1_{2\pi}}d\theta'\mu_0(\theta',t',h')\int_0^{t-t'}dt''\int_{-1}^1dh''Q(t-t'-t'',h|h'')\gamma(t'',h'',\tilde\theta-\theta'|h').
\end{align}}

Therefore we obtained

\begin{align*}
\eqref{esprA}=\eqref{esprB}+\eqref{esprC}=\eqref{nuovaB}+\eqref{nuovaC},
\end{align*}

which, for any $\mu_0\in L^1$, writes as

\begin{align*}
&\int_{\mathbb{T}^1_{2\pi}}d\theta'\int_0^tdt'\int_{-1}^1dh'\mu_0(\theta',t',h')\gamma(\theta-\theta',t-t',h|h')
\\
&=\int_{\mathbb{T}^1_{2\pi}}d\theta'\int_0^tdt'\int_{-1}^1dh'\mu_0(\theta',t',h')\biggl[f(\theta-\theta',t-t',h|h')
\\
&+\int_0^{t-t'}dt''\int_{-1}^1dh''Q(t-t'-t'',h|h'')\gamma(\theta+\pi-2\arcsin(h'')-\theta',t'',h''|h')\biggr].
\end{align*}

Therefore property \eqref{rapp_gamma} holds for any $\mu_0\in L^1$ if and only if, for any $(\theta,t,h|h')\in\mathbb{T}^1_{2\pi}\times[0,+\infty)\times[-1,1]^2$, $\gamma$ verifies

{\small\begin{align}
\label{espr:gamma}\gamma(\theta,t,h|h')=f(\theta,t,h|h')+\int_0^tdt'\int_{-1}^1dh''Q(t-t',h|h'')\gamma(\theta+\pi-2\arcsin(h''),t',h''|h').
\end{align}}

Now we show that, for fixed $h'$, $\gamma(\cdot,\cdot,\cdot|h')$ exists in $\cap_{T>0}L^{\infty}(\mathbb{R}\times[0,T]\times[-1,1])$ and is non negative. To prove it, it is sufficient to use Lemma \ref{lemma:rho}, with $h'$ fixed, in the case: no dependance on $x$, with $p=\infty$ and $\mu(\theta,t,h)=f(\theta,t,h|h')$.

We need an estimate in the space $L^{\infty}(\mathbb{T}^1_{2\pi}\times[0,T]\times[-1,1])$ uniformly with respect to $h'$ with a possible dependance on $T$, that is an estimate in the space $L^{\infty}(\mathbb{T}^1_{2\pi}\times[0,T]\times[-1,1]^2)$. To get it, we recall that $\gamma$ is obtained by using Lemma \ref{lemma:rho}, which is based on the Banach fixed-point Theorem, hence $\|\gamma(\cdot,\cdot,\cdot|h')\|_{L^{\infty}(\mathbb{T}^1_{2\pi}\times[0,T]\times[-1,1])}$ is bounded by a linear function of $\|f(\cdot,\cdot,\cdot|h')\|_{L^{\infty}(\mathbb{T}^1_{2\pi}\times[0,T]\times[-1,1])}$. By \eqref{deceffe} of Lemma \ref{propreffe}, $\|f(\cdot,\cdot,\cdot|h')\|_{L^{\infty}(\mathbb{T}^1_{2\pi}\times[0,T]\times[-1,1])}$ is uniformly bounded with respect to $h'$, hence $\|\gamma(\cdot,\cdot,\cdot|h')\|_{L^{\infty}(\mathbb{T}^1_{2\pi}\times[0,T]\times[-1,1])}$ is too. 

\textbf{Step 2:} we want to use $\gamma-\frac{1}{2\pi}$ instead of $\gamma$: since we would like to state that $\mu_t(\theta,0,h)\xrightarrow[t\to+\infty]{}\frac{\langle\mu_0\rangle}{2\pi}$, what we expect from $\gamma$ by looking the equation \eqref{rapp_gamma} is that $\gamma(\theta,t,h|h')\xrightarrow[t\to+\infty]{}\frac{1}{2\pi}$. We intentionally do not precise the setting which this convergence is to be understood in, we will clarify it later.

Let us then define
\[
\varphi:=\gamma-\frac{1}{2\pi}.
\]
With this notation, by \eqref{rapp_gamma}, $(\theta,h)\mapsto\mu_t(\theta,0,h)$ as a function of $\varphi$ verifies:

{\small\begin{align*}
\mu_t(\theta,0,h)&=\mu_0(\theta,t,h)+\int_0^tdt'\int_{-1}^1dh'Q(t-t',h|h')\mu_0(\theta+\pi-2\arcsin(h'),t',h')
\\
&+\frac{1}{2\pi}\int_0^{2\pi}d\theta'\int_0^tdt'\int_{-1}^1dh'\mu_0(\theta',t',h')
\\
&+\int_0^{2\pi}d\theta'\int_0^tdt'\int_{-1}^1dh'\varphi(\theta-\theta',t-t',h|h')\mu_0(\theta',t',h'),
\end{align*}}

and, using \eqref{espr:gamma}, $\varphi$ solves

\begin{align*}
\varphi(\theta,t,h|h')&=f(\theta,t,h|h')-\frac{1}{2\pi}E(t,h)
\\
&+\int_0^tdt'\int_{-1}^1dh''Q(t-t',h|h'')\varphi(\theta+\pi-2\arcsin(h''),t',h''|h'),
\end{align*}

with $f$ of Definition \ref{defeffe}.

The equations above are exactly the properties \eqref{espr:phi} and \eqref{rapp:musphi}.
\end{proof}
Now, as we said, the next step would be to deduce the long time behavior of $\varphi$ from its definition in \eqref{espr:phi}. But there are two problems in using directly \eqref{espr:phi}. The first one is that the convolution kernel in the integral is $Q$, and the integral of $Q$ is 1. If it was $<1$, one could have used Banach fixed-point Theorem to get long time estimates on $\varphi$. The second one is that the integral in \eqref{espr:phi} involves only the variables $(t',h')$, and the integral of $f-\frac{1}{2\pi}E$ only with respect to $(t,h)$ is non zero. Instead, by \eqref{fint1} of Lemma \ref{propreffe}, the integral of $f-\frac{1}{2\pi}E$ with respect to $(\theta,t,h)$ is zero. This suggests to iterate over \eqref{espr:phi} twice, so that, changing variables, we integrate also with respect to $\theta$.

Therefore, in the next two Lemmas, we get a more useful writing for $\varphi$ and a property of $f-\frac{1}{2\pi}E$. To show their next use, we state both of them together.

We recall that $Q^{(n)}$ and $E^{(n)}$ for $n\in\mathbb{N}$ are defined in Definitions \ref{def:Qn} and \ref{def:En}.
\begin{lemma}\label{lemma:phimeglio}
For any $c\in\mathbb{R}$ the function $\varphi$ defined in \eqref{espr:phi} of Lemma \ref{lemma:esphi} verifies:

{\scriptsize\begin{align*}
\varphi(\theta,t,h|h')&=f(\theta,t,h|h')+\int_0^tdt'\int_{-1}^1dh''Q(t-t',h|h'')f(\theta+\pi-2\arcsin(h''),t',h''|h')
\\
&-\frac{1}{2\pi}E^{(2)}(t,h)+\frac{c}{2\pi}\left[\int_t^{\infty}dt'\int_{-1}^1dh''E^{(2)}(t',h'')-E^{(2)}(t,h')-E^{(3)}(t,h')\right]
\\
&+\int_{\mathbb{T}^1_{2\pi}}d\theta'\int_0^tdt'\int_{-1}^1dh'''\varphi(\theta',t',h'''|h')\left[f(\theta-\theta',t-t',h|h''')-\frac{c}{2\pi}E^{(2)}(t-t',h''')\right].
\end{align*}}

\end{lemma}
\begin{lemma}\label{lemmaok}
There exists a constant $\bar c>0$ such that for any $c\in(0,\bar c]$ it holds:
\begin{align*}
d:=\sup_{h\in[-1,1]}\left\|f(\cdot,\cdot,h|\cdot)-\frac{c}{2\pi}E^{(2)}\right\|_{L^1(\mathbb{T}^1_{2\pi}\times[0,+\infty)\times[-1,1])}<1.
\end{align*}

\end{lemma}
Before proving both lemmas above, we first show how they lead to the following estimate of $\varphi$ that will be useful for the proof of Proposition \ref{prop:musvarphi}.
\begin{lemma}\label{decphi}
There exists a constant $C>0$ such that the function $\varphi$ defined by \eqref{espr:phi} of Lemma \ref{lemma:esphi} verifies
\[
|\varphi(\theta,t,h|h')|\leq\frac{C}{t+1}\quad\forall(\theta,t,h|h')\in\mathbb{T}^1_{2\pi}\times[0,+\infty)\times[-1,1]^2.
\]
\end{lemma}
\begin{proof} Fix $c\in(0,\bar c]$, with $\bar c$ provided by Lemma \ref{lemmaok}, and rewrite $\varphi$ as in Lemma \ref{lemma:phimeglio} with such $c$. Also denote

{\footnotesize\begin{align}
\label{addendieffe}J(\theta,t,h|h')&:=f(\theta,t,h|h')+\int_0^tdt'\int_{-1}^1dh''Q(t-t',h|h'')f(\theta+\pi-2\arcsin(h''),t',h''|h')
\\
\label{addendiEn}&-\frac{1}{2\pi}E^{(2)}(t,h)+\frac{c}{2\pi}\left[\int_t^{\infty}dt'\int_{-1}^1dh''E^{(2)}(t',h'')-E^{(2)}(t,h')-E^{(3)}(t,h')\right],
\end{align}}

and we use this notation only in this proof. With $J$ the expression in Lemma \ref{lemma:phimeglio} has the equivalent formulation:

{\scriptsize\begin{align}
\nonumber\varphi(\theta,t,h|h')&=J(\theta,t,h|h')
\\
\label{ciao}&+\int_{\mathbb{T}^1_{2\pi}}d\theta'\int_0^tdt'\int_{-1}^1dh'''\varphi(\theta',t',h'''|h')\left[f(\theta-\theta',t-t',h|h''')-\frac{c}{2\pi}E^{(2)}(t-t',h''')\right].
\end{align}}

We estimate the summands in \eqref{addendieffe} with Lemma \ref{propreffe}. The summands in \eqref{addendiEn} can be bounded by Lemma \ref{proprEn}. In this way we can state that the following quantities concerning $J$ verify:

\begin{align*}
\|J\|_{L^{\infty}}<+\infty\quad\textrm{ and }\quad A_J:=\sup_{(\theta,t,h|h')\in\mathbb{T}^1_{2\pi}\times[0,+\infty)\times[-1,1]^2}|tJ(\theta,t,h|h')|<+\infty.
\end{align*}

\textbf{Step 1:} to begin with, we prove that $\varphi$ is bounded. Denote
\[
B_T^{\varphi}:=\|\varphi\|_{L^{\infty}(\mathbb{T}^1_{2\pi}\times[0,T]\times[-1,1]^2)},
\]
which is finite for any $T>0$ thanks to Lemma \ref{lemma:esphi}. If we now look at equation \eqref{ciao} we can write for $t\leq T$

{\scriptsize\begin{align*}
|\varphi(\theta,t,h|h')|&\leq\underbrace{|J(\theta,t,h|h')|}_{\leq\|J\|_{L^{\infty}}}
\\
&+\int_{\mathbb{T}^1_{2\pi}}d\theta'\int_0^tdt'\int_{-1}^1dh'''\underbrace{|\varphi(\theta',t',h'''|h')|}_{\leq B_{t'}^{\varphi}\leq B_T^{\varphi}}\left|f(\theta-\theta',t-t',h|h''')-\frac{c}{2\pi}E^{(2)}(t-t',h''')\right|
\\
&\leq\|J\|_{L^{\infty}}+B_T^{\varphi}\int_{\mathbb{T}^1_{2\pi}}d\theta'\int_0^tdt'\int_{-1}^1dh'''\left|f(\theta-\theta',t-t',h|h''')-\frac{c}{2\pi}E^{(2)}(t-t',h''')\right|
\\
&\leq\|J\|_{L^{\infty}}+B_T^{\varphi}\underbrace{\left\|f(\cdot,\cdot,h|\cdot)-\frac{c}{2\pi}E^{(2)}\right\|_{L^1(\mathbb{T}^1_{2\pi}\times[0,+\infty)\times[-1,1])}}_{\leq d<1\quad\forall h\textrm{ by Lemma }\ref{lemmaok}}
\\
&\leq\|J\|_{L^{\infty}}+B_T^{\varphi}d,
\end{align*}}

and therefore since the previous inequality holds for any $t\leq T$, the same holds for $B_T^{\varphi}=\|\varphi\|_{L^{\infty}(\mathbb{T}^1_{2\pi}\times[0,T]\times[-1,1]^2)}$, that is 

\begin{align*}
B_T^{\varphi}\leq\|J\|_{L^{\infty}}+B_T^{\varphi}d\quad\textrm{ and therefore}\quad B_T^{\varphi}\leq\frac{\|J\|_{L^{\infty}}}{1-d},
\end{align*}

and since this inequality holds in turn regardless of $T$, this proves that

\begin{align*}
\|\varphi\|_{L^{\infty}(\mathbb{T}^1_{2\pi}\times[0,T]\times[-1,1]^2)}\leq\frac{\|J\|_{L^{\infty}}}{1-d}\quad\forall T>0\quad\textrm{ and therefore }\quad\varphi\in L^{\infty}.
\end{align*}

\textbf{Step 2:} we prove now that $\varphi$ decays at most as $\frac{1}{t}$. Denote again
\[
C_T^{\varphi}:=\sup_{(\theta,t,h|h')\in\mathbb{T}^1_{2\pi}\times[0,T]\times[-1,1]^2}|t\varphi(\theta,t,h|h')|,
\]
which is again finite for fixed $T$ thanks to Lemma \ref{lemma:esphi}. Also fix  $1<\alpha<\frac{1}{d}$, with $d<1$ defined in Lemma \ref{lemmaok}. If we use again the expression \eqref{ciao}, for $t\leq T$ we get

{\tiny\begin{align*}
|t\varphi(\theta,t,h|h')|&\leq\underbrace{|tJ(\theta,t,h|h')|}_{\leq A_J}
\\
&+t\int_{\mathbb{T}^1_{2\pi}}d\theta'\int_0^{\frac{t}{\alpha}}dt'\int_{-1}^1dh'''\underbrace{|\varphi(\theta',t',h'''|h')|}_{\leq \|\varphi\|_{L^{\infty}}}\left|f(\theta-\theta',t-t',h|h''')-\frac{c}{2\pi}E^{(2)}(t-t',h''')\right|
\\
&+t\int_{\mathbb{T}^1_{2\pi}}d\theta'\int_{\frac{t}{\alpha}}^tdt'\int_{-1}^1dh'''\underbrace{|\varphi(\theta',t',h'''|h')|}_{\leq \frac{C_{t'}^{\varphi}}{t'}\leq \frac{\alpha C_T^{\varphi}}{t}}\left|f(\theta-\theta',t-t',h|h''')-\frac{c}{2\pi}E^{(2)}(t-t',h''')\right|
\\
&\leq A_J+\|\varphi\|_{L^{\infty}}t\int_{\mathbb{T}^1_{2\pi}}d\theta'\int_0^{\frac{t}{\alpha}}dt'\int_{-1}^1dh'''\left|f(\theta-\theta',t-t',h|h''')-\frac{c}{2\pi}E^{(2)}(t-t',h''')\right|
\\
&+\alpha C_T^{\varphi}\int_{\mathbb{T}^1_{2\pi}}d\theta'\int_{\frac{t}{\alpha}}^tdt'\int_{-1}^1dh'''\left|f(\theta-\theta',t-t',h|h''')-\frac{c}{2\pi}E^{(2)}(t-t',h''')\right|
\\
&\leq A_J+\|\varphi\|_{L^{\infty}}t\int_{\mathbb{T}^1_{2\pi}}d\theta'\int_{t(1-\frac{1}{\alpha})}^tdt'\int_{-1}^1dh'''\left|f(\theta',t',h|h''')-\frac{c}{2\pi}E^{(2)}(t',h''')\right|
\\
&+\alpha C_T^{\varphi}\underbrace{\left\|f(\cdot,\cdot,h|\cdot)-\frac{c}{2\pi}E^{(2)}\right\|_{L^1(\mathbb{T}^1_{2\pi}\times[0,+\infty)\times[-1,1])}}_{\leq d<1\quad\forall h\textrm{ by Lemma }\ref{lemmaok}}
\\
&\leq A_J+\|\varphi\|_{L^{\infty}}t\int_{\mathbb{T}^1_{2\pi}}d\theta'\int_{t(1-\frac{1}{\alpha})}^{\infty}dt'\int_{-1}^1dh'''\left[f(\theta',t',h|h''')+\frac{c}{2\pi}E^{(2)}(t',h''')\right]+\alpha dC_T^{\varphi}.
\end{align*}}

We study here the second one of these three summands. The first term can be estimated as:

\begin{align*}
\int_{\mathbb{T}^1_{2\pi}}d\theta'\int_{t(1-\frac{1}{\alpha})}^{\infty}t'\int_{-1}^1dh'''f(\theta',t',h|h''')&=\int_{t(1-\frac{1}{\alpha})}^{\infty}dt'\int_{-1}^1dh'''Q^{(2)}(t',h|h''')
\\
&=E^{(2)}\left(\frac{t(\alpha-1)}{\alpha},h\right)
\\
&\leq\frac{\alpha c_2}{t(\alpha-1)}\textrm{ by }\eqref{ubEn}\textrm{ of Lemma \ref{proprEn}}.
\end{align*}

The second part can be estimated as:

\begin{align*}
\int_{\mathbb{T}^1_{2\pi}}d\theta'\int_{t(1-\frac{1}{\alpha})}^{\infty}dt'\int_{-1}^1dh'''\frac{c}{2\pi}E^{(2)}(t',h''')&=c\int_{t(1-\frac{1}{\alpha})}^{\infty}dt'\int_{-1}^1dh'''E^{(2)}(t',h''')
\\
&\leq\frac{\alpha cc_2'}{t(\alpha-1)}\textrm{ by }\eqref{restoEn}\textrm{ of Lemma \ref{proprEn}}.
\end{align*}

Finally we got

\begin{align*}
|t\varphi(\theta,t,h|h')|\leq A_J+\frac{\alpha\|\varphi\|_{L^{\infty}}(c_2+cc_2')}{\alpha-1}+\alpha dC_T^{\varphi}\quad\forall t\leq T.
\end{align*}

Passing the inequality above to the supremum on all variables, we obtain:
{\footnotesize\[
C_T^{\varphi}\leq A_J+\frac{\alpha\|\varphi\|_{L^{\infty}}(c_2+cc_2')}{\alpha-1}+\alpha dC_T^{\varphi},\quad\textrm{ therefore }\quad C_T^{\varphi}\leq\frac{A_J+\frac{\alpha\|\varphi\|_{L^{\infty}}(c_2+cc_2')}{\alpha-1}}{1-\alpha d}\quad\forall T>0,
\]}
i.e.
\[
\sup_{(\theta,t,h|h')\in\mathbb{T}^1_{2\pi}\times[0,+\infty)\times[-1,1]^2}|t\varphi(\theta,t,h|h')|\leq\frac{A_J+\frac{\alpha\|\varphi\|_{L^{\infty}}(c_2+cc_2')}{\alpha-1}}{1-\alpha d},
\]
which concludes the proof.
\end{proof}
We show now how both Lemmas \ref{lemma:phimeglio} and \ref{lemmaok} imply Proposition \ref{prop:musvarphi}.
\newline
\subsubsection{Proof of Proposition \ref{prop:musvarphi}.}
\begin{proof}
The function $\varphi$ in the statement of Proposition \ref{prop:musvarphi} is exactly the function $\varphi$ defined by Lemma \ref{lemma:esphi}. By Lemma \ref{decphi}, such function $\varphi$ satisfies:
\[
|\varphi(\theta,t,h|h')|\leq\frac{C}{t+1}\quad\forall(\theta,t,h|h')\in\mathbb{T}^1_{2\pi}\times[0,+\infty)\times[-1,1]^2,
\]
that is, the estimate in the statement of Proposition \ref{prop:musvarphi}.
\end{proof}

We now prove Lemmas \ref{lemma:phimeglio} and \ref{lemmaok}.
\newline
\subsubsection{Proof of Lemma \ref{lemma:phimeglio}.}
\begin{proof} We split up the proof into three steps.
\newline
\textbf{Step 1:} First we prove here that the function $\varphi$ defined through equation \eqref{espr:phi} of Lemma \ref{lemma:esphi} satisfies the following equality:

\begin{align}
\nonumber\varphi(\theta,t,h|h')&=f(\theta,t,h|h')-\frac{1}{2\pi}E^{(2)}(t,h)
\\
\nonumber&+\int_0^tdt'\int_{-1}^1dh''Q(t-t',h|h'')f(\theta+\pi-2\arcsin(h''),t',h''|h')
\\
\label{passo1}&+\int_{\mathbb{T}^1_{2\pi}}d\theta'\int_0^tdt'\int_{-1}^1dh'''\varphi(\theta',t',h'''|h')f(\theta-\theta',t-t',h|h'''),
\end{align}

with $E^{(2)}$ from Definition \ref{def:En}.

We begin by noticing that we can iterate twice the map defining  $\varphi$ in Lemma \ref{lemma:esphi}. This can be done substituting the whole right-hand side of \eqref{espr:phi} into the $\varphi$ inside the integral. In this way we get

{\small\begin{align}
\label{termvarphi}&\varphi(\theta,t,h|h')
\\
\label{termok}&=f(\theta,t,h|h')+\int_0^tdt'\int_{-1}^1dh''Q(t-t',h|h'')f(\theta+\pi-2\arcsin(h''),t',h''|h')
\\
\label{termE2}&-\frac{1}{2\pi}E(t,h)-\frac{1}{2\pi}\int_0^tdt'\int_{-1}^1dh''Q(t-t',h|h'')E(t',h'')
\\
\label{termphi}&+\int_0^tdt'\int_{-1}^1dh''Q(t-t',h|h'')\int_0^{t'}dt''\int_{-1}^1dh'''Q(t'-t'',h''|h''')\varphi(\tilde\theta,t'',h'''),
\end{align}}

with
\[
\tilde\theta=\tilde\theta(\theta,h'',h'''):=\theta+2\pi-2\arcsin(h'')-2\arcsin(h''').
\]
Now the term in \eqref{termvarphi} is fine, but the one in \eqref{termE2} can be written better than that by using $E^{(2)}$ of Definition \ref{def:En}. Indeed property \eqref{espr:En1} of Lemma \ref{proprEn} implies that

\begin{align*}
\eqref{termE2}=-\frac{1}{2\pi}E^{(2)}(t,h).
\end{align*}

As for the term in \eqref{termphi}, by Fubini-Tonelli Theorem, we can change again the integration order: before we integrate with respect to $(t'',h''',h'')$ and then with respect to $t'$. Then we exchange the names of the variables, this time only $t''$ with $t'$. In this way, we obtain

{\small\begin{align*}
\eqref{termphi}=\int_0^tdt'\int_{-1}^1dh'''\int_{-1}^1dh''\varphi(\tilde\theta,t',h'''|h')\int_0^{t-t'}dt''Q(t-t'-t'',h|h'')Q(t'',h''|h''').
\end{align*}}

After a change of variables
\[
\tilde\theta=\theta+2\pi-2\arcsin(h''')-2\arcsin(h'')=\theta',
\]
in the right-hand side, we obtain

{\tiny\begin{align*}
\eqref{termphi}&=\int_0^tdt'\int_{-1}^1dh'''\int_{\theta+\pi-2\arcsin(h''')}^{\theta+3\pi-2\arcsin(h''')}d\theta'\varphi(\theta',t',h'''|h')f(\theta-\theta',t-t',h|h''')\textrm{ with }f\textrm{ of Definition \ref{defeffe}}
\\
&=\int_{\mathbb{T}^1_{2\pi}}d\theta'\int_0^tdt'\int_{-1}^1dh'''\varphi(\theta',t',h'''|h')f(\theta-\theta',t-t',h|h'''),
\end{align*}}

because the integrand is again periodic in $\theta'$ and therefore any interval of length $2\pi$ is fine.

Summing \eqref{termok} with reformulated \eqref{termE2} and \eqref{termphi}, we get property \eqref{passo1}.

\textbf{Step 2:} We prove now that

{\scriptsize\begin{align}\label{passo2}
\int_{\mathbb{T}^1_{2\pi}}d\theta\int_0^tdt'\int_{-1}^1dh\varphi(\theta,t',h|h')E^{(2)}(t-t',h)
=\int_t^{\infty}dt'\int_{-1}^1dhE^{(2)}(t',h)-E^{(2)}(t,h')-E^{(3)}(t,h').
\end{align}}

To prove property \eqref{passo2}, we compute the integral $\int_{\mathbb{T}^1_{2\pi}}d\theta\int_0^tdt'\int_{-1}^1dh$ of both sides of equation \eqref{passo1}. We focus on the summands in the right-hand side of the identity \eqref{passo1}, since the integral of the left-hand side does not need to be changed.

The first one is immediate:

{\scriptsize\begin{align*}
\int_{\mathbb{T}^1_{2\pi}}d\theta\int_0^tdt'\int_{-1}^1dh f(\theta,s,h|h')&=\int_0^tdt'\int_{-1}^1dh\int_{\mathbb{T}^1_{2\pi}}d\theta f(\theta,t',h|h')\textrm{ by changing integration order}
\\
&=\int_0^tdt'\int_{-1}^1dhQ^{(2)}(t',h|h')\textrm{ by property \eqref{intfQ2}}
\\
&=1-E^{(2)}(t,h')\textrm{ by Definition }\ref{def:En},
\end{align*}}

and the second one is very similar

{\small\begin{align*}
&\int_{\mathbb{T}^1_{2\pi}}d\theta\int_0^tdt'\int_{-1}^1dh\int_0^{t'}dt''\int_{-1}^1dh''Q(t'-t'',h|h'')f(\theta+\pi-2\arcsin(h''),t'',h''|h')
\\
&=\int_0^tdt'\int_{-1}^1dh\int_0^{t'}t''\int_{-1}^1dh''Q(t'-t'',h|h'')\underbrace{\int_{\mathbb{T}^1_{2\pi}}d\theta f(\theta,t'',h''|h')}_{=Q^{(2)}(t'',h''|h')\textrm{ by property \eqref{intfQ2}}}
\\
&=\int_0^tdt'\int_{-1}^1dh\underbrace{\int_0^{t'}dt''\int_{-1}^1dh''Q(t'-t'',h|h'')Q^{(2)}(t'',h''|h')}_{=Q^{(3)}(t',h|h')\textrm{ by Definition }\ref{def:Qn}}
\\
&=\int_0^tdt'\int_{-1}^1dhQ^{(3)}(t',h|h')=1-E^{(3)}(t,h')\textrm{ by Definition }\ref{def:En}.
\end{align*}}

Notice that by property \eqref{intEn} of Lemma \ref{proprEn}, $E^{(2)}$ has integral $2$ with respect to variables $(s,h)$, so the third term of the sum above writes as:

\begin{align*}
\int_{\mathbb{T}^1_{2\pi}}d\theta\int_0^tdt'\int_{-1}^1dh\frac{1}{2\pi}E^{(2)}(t',h)=2-\int_t^{\infty}dt'\int_{-1}^1dhE^{(2)}(t',h).
\end{align*}

The integral of the fourth summand instead can be written as

{\tiny\begin{align*}
&\int_{\mathbb{T}^1_{2\pi}}d\theta\int_0^tdt'\int_{-1}^1dh\int_0^{t'}dt''\int_{-1}^1dh'''\int_{\mathbb{T}^1_{2\pi}}d\theta'\varphi(\theta',t'',h'''|h')f(\theta-\theta',t'-t'',h|h''')
\\
&=\int_0^tdt'\int_{-1}^1dh\int_0^{t'}dt''\int_{-1}^1dh'''\int_{\mathbb{T}^1_{2\pi}}d\theta'\varphi(\theta',t'',h'''|h')\underbrace{\int_{\mathbb{T}^1_{2\pi}}d\theta f(\theta-\theta',t'-t'',h|h''')}_{=Q^{(2)}(t'-t'',h|h''') \textrm{ by property \eqref{intfQ2}}}
\\
&=\int_{\mathbb{T}^1_{2\pi}}d\theta'\int_0^tdt''\int_{-1}^1dh'''\varphi(\theta',t'',h'''|h')\int_{t''}^tdt'\int_{-1}^1dhQ^{(2)}(t'-t'',h|h''')
\\
&=\int_{\mathbb{T}^1_{2\pi}}d\theta'\int_0^tdt''\int_{-1}^1dh'''\varphi(\theta',t'',h'''|h')\left(1-\int_{t-t''}^{\infty}dt'\int_{-1}^1dhQ^{(2)}(t',h|h''')\right)\textrm{ by \eqref{Qnint1} of Lemma \ref{proprQn}}
\\
&=\int_{\mathbb{T}^1_{2\pi}}d\theta\int_0^tdt''\int_{-1}^1dh'''\varphi(\theta,t'',h'''|h')(1-E^{(2)}(t-t'',h'''))\textrm{ by Definition \ref{def:En}}.
\end{align*}}

Summing upon all the integral of the terms in the right-hand side of equation \eqref{passo1} we get equation \eqref{passo2}.

\textbf{Step 3:} we can conclude. If we subtract and add back to both sides of \eqref{passo1} the quantity

\begin{align*}
\frac{c}{2\pi}\int_{\mathbb{T}^1_{2\pi}}d\theta'\int_0^tdt'\int_{-1}^1dh'''\varphi(\theta',t',h'''|h')E^{(2)}(t-t',h'''),
\end{align*}

where $c\in\mathbb{R}$, we get

{\footnotesize\begin{align*}
\varphi(\theta,t,h|h')&=f(\theta,t,h|h')-\frac{1}{2\pi}E^{(2)}(t,h)
\\
&+\int_0^tdt'\int_{-1}^1dh''Q(t-t',h|h'')f(\theta+\pi-2\arcsin(h''),t',h''|h')
\\
&+\int_{\mathbb{T}^1_{2\pi}}d\theta'\int_0^tdt'\int_{-1}^1dh'''\varphi(\theta',t',h'''|h')f(\theta-\theta',t-t',h|h''')
\\
&=f(\theta,t,h|h')-\frac{1}{2\pi}E^{(2)}(t,h)
\\
&+\int_0^tdt'\int_{-1}^1dh''Q(t-t',h|h'')f(\theta+\pi-2\arcsin(h''),t',h''|h')
\\
&+\int_{\mathbb{T}^1_{2\pi}}d\theta'\int_0^tdt'\int_{-1}^1dh'''\varphi(\theta',t',h'''|h')\left[f(\theta-\theta',t-t',h|h''')-\frac{c}{2\pi}E^{(2)}(t-t',h''')\right]
\\
&+\frac{c}{2\pi}\underbrace{\int_{\mathbb{T}^1_{2\pi}}d\theta'\int_0^tdt'\int_{-1}^1dh'''\varphi(\theta',t',h'''|h')E^{(2)}(t-t',h''')}_{=\eqref{passo2}},
\end{align*}}

and this concludes the proof of Lemma \ref{lemma:phimeglio}.
\end{proof}
Now prove Lemma \ref{lemmaok}.
\newline
\subsubsection{Proof of Lemma \ref{lemmaok}.}
\begin{proof} We denote here $f^h(\theta,t,h')=f(\theta,t,h|h')$, and we also write $\int d\theta dtdh':=\int_{\mathbb{T}^1_{2\pi}}d\theta\int_0^{\infty}dt\int_{-1}^1dh'$. So we get

{\small\begin{align*}
\left\|f^h-\frac{c}{2\pi}E^{(2)}\right\|_{L^1}&=\int_{f^h>\frac{c}{2\pi}E^{(2)}}d\theta dtdh' f^h-\frac{c}{2\pi}\int_{f^h>\frac{c}{2\pi}E^{(2)}}d\theta dtdh' E^{(2)}
\\
&-\int_{f^h\leq\frac{c}{2\pi}E^{(2)}}d\theta dtdh' f^h+\frac{c}{2\pi}\int_{f^h\leq\frac{c}{2\pi}E^{(2)}}d\theta dtdh' E^{(2)}
\\
&=\underbrace{\int d\theta dtdh' f^h}_{=1\textrm{ by }\eqref{fint1dinuovo}\textrm{ of Lemma \ref{propreffe}}}-2\underbrace{\int_{f^h\leq\frac{c}{2\pi}E^{(2)}}d\theta dtdh' f^h}_{\geq0}
\\
&-c\left[\frac{2}{2\pi}\int_{f^h>\frac{c}{2\pi}E^{(2)}}d\theta dtdh' E^{(2)}-\underbrace{\frac{1}{2\pi}\int d\theta dtdh' E^{(2)}}_{=2\textrm{  by \eqref{intEn} of Lemma \ref{proprEn}}}\right]
\\
&\leq1-2c\left[\frac{1}{2\pi}\int_{f^h>\frac{c}{2\pi}E^{(2)}}d\theta dtdh' E^{(2)}-1\right]
\\
&=1-2c\left[\frac{1}{2\pi}\int_{f^h>0}d\theta dtdh' E^{(2)}-1\right]+2c\int_{0<f^h\leq\frac{c}{2\pi}E^{(2)}} d\theta dtdh' E^{(2)}.
\end{align*}}

Now we first prove that the second of these three summands divided by $2c$ is uniformly negative for $h\in[-1,1]$ (next \textbf{Step 1}), and then we prove that the third one divided by $c$ is infinitesimal as $c\to0$ (next \textbf{Step 2}). Once done this, we could conclude, indeed we would have

{\scriptsize\begin{align}
\label{anteconclusione}\left\|f^h-\frac{c}{2\pi}E^{(2)}\right\|_{L^1}&\leq1-2c\underbrace{\left[\frac{1}{2\pi}\int_{f^h>0}d\theta dtdh' E^{(2)}-1\right]}_{\geq C>0\quad\forall h\in[-1,1]\textrm{ by using }\textbf{Step 1}}+2c\underbrace{\int d\theta dtdh'_{0<f^h\leq\frac{c}{2\pi}E^{(2)}}E^{(2)}}_{\leq\frac{C}{2}\textrm{ for }c\leq\bar c<<1\quad\forall h\in[-1,1]\textrm{ by }\textbf{Step 2}}
\\
\label{conclusione}&\leq1-cC\textrm{ for }c<<1.
\end{align}}

Therefore we begin with the proof of \textbf{Step 1}: we have to prove that the first summand in \eqref{anteconclusione} is uniformly positive as $c\to0$.

\textbf{Step 1:} we prove here that

\begin{align}\label{derivataneg}
\inf_{h\in[-1,1]}\frac{1}{2\pi}\int_{f^h>0}d\theta dtdh' E^{(2)}-1>0.
\end{align}

The easier way to prove this is by finding a domain
\[
\Omega\subseteq\left\{(\theta,t,h'):t\geq0,h\in[-1,1],\theta\in[2\arcsin(h')-3\pi,2\arcsin(h')-\pi]\right\},
\]
where each $f^h$ is positive. We are abusing a little bit the notation, since we are treating $f^h$ as if it was defined for $\theta\in[2\arcsin(h)-3\pi,2\arcsin(h)-\pi]$. To find such a domain, we notice that

\begin{align}\label{inf1}
\frac{\partial h''(\theta,h')}{\partial\theta}>0\quad\forall\theta\in(2\arcsin(h')-3\pi,2\arcsin(h')-\pi),h'\in[-1,1],
\end{align}

that is, the Jacobian determinant does not affect the region where $f^h$ is positive.

We also observe that for $t\in(0,\frac{1}{2}]$ we have

\begin{align}\label{inf2}
\int_0^tdt' \underbrace{Q(t-t',h|h'')}_{=\frac{6}{\pi^2}\textrm{ since }t-t'\leq\frac{1}{2}}\underbrace{Q(t',h''|h')}_{=\frac{6}{\pi^2}\textrm{ since }t'\leq\frac{1}{2}}=\frac{36}{\pi^4}t>0,
\end{align}

and for $t\in(\frac{1}{2},1)$ we have

{\footnotesize\begin{align}\label{inf2bis}
\int_0^tdt' Q(t-t',h|h'')Q(t',h''|h')\geq\int_{t-\frac{1}{2}}^{\frac{1}{2}}dt' \underbrace{Q(t-t',h|h'')}_{=\frac{6}{\pi^2}\textrm{ since }t-t'\leq\frac{1}{2}}\underbrace{Q(t',h''|h')}_{=\frac{6}{\pi^2}\textrm{ since }t'\leq\frac{1}{2}}=\frac{36}{\pi^4}(1-t)>0.
\end{align}}

Combining the properties \eqref{inf1}, \eqref{inf2} and \eqref{inf2bis}, we get
\[
f^h(\theta,t,h')>0\quad\forall t\in(0,1),\theta\in(2\arcsin(h')-3\pi,2\arcsin(h')-\pi),h,h'\in[-1,1].
\]
Therefore, if we define
\[
\Omega:=\left\{(\theta,t,h'):t\in(0,1),h'\in[-1,1],\theta\in(2\arcsin(h')-3\pi,2\arcsin(h')-\pi)\right\},
\]
we obtain

\begin{align}
\nonumber\frac{1}{2\pi}\int_{f^h>0}d\theta dtdh' E^{(2)}&\geq\frac{1}{2\pi}\int_{\Omega}d\theta dtdh' E^{(2)}
\\
\nonumber&=\frac{1}{2\pi}\int_0^1dt\int_{-1}^1dh'\int_{2\arcsin(h')-3\pi}^{2\arcsin(h')-\pi}d\theta E^{(2)}(t,h')
\\
\label{lb}&=\int_0^1dt\int_{-1}^1dh'E^{(2)}(t,h').
\end{align}

Now we have to bound from below the quantity \eqref{lb}. To this end, we observe that since $Q\leq\frac{6}{\pi^2}$, then

{\small\begin{align}
\label{ub} Q^{(2)}(t,h'|h'')=\int_0^tdt'\int_{-1}^1dh'''Q(t-t',h'|h''')Q(t',h'''|h'')\leq\int_0^tdt'\int_{-1}^1dh''\frac{36}{\pi^4}=\frac{72}{\pi^4}t,
\end{align}}

and therefore, as for the right-hand side of \eqref{lb}, we have

\begin{align*}
\int_0^1dt\int_{-1}^1dh'E^{(2)}(t,h')&=\int_0^1dt\int_{-1}^1dh'\int_t^{\infty}dt'\int_{-1}^1dh''Q^{(2)}(t',h'|h'')
\\
&=\int_0^{\infty}dt'\int_{-1}^1dh'\int_{-1}^1dh''Q^{(2)}(t',h'|h'')\min\{1,t'\}
\\
&=\int_0^1dt'\int_{-1}^1dh'\int_{-1}^1dh''Q^{(2)}(t',h'|h'')t'
\\
&+\underbrace{\int_1^{\infty}dt'\int_{-1}^1dh'\int_{-1}^1dh''Q^{(2)}(t',h'|h'')}_{=2-\int_0^1dt'\int_{-1}^1dh'\int_{-1}^1dh''Q^{(2)}(t',h'|h'')\textrm{ by property }\eqref{Qnint1}}
\\
&=2-\int_0^1dt'\int_{-1}^1dh'\int_{-1}^1dh''\underbrace{Q^{(2)}(t',h'|h'')}_{\leq\frac{72}{\pi^4}t'\textrm{ thanks to }\eqref{ub}}(1-t')
\\
&\geq2-\int_0^1dt'\int_{-1}^1dh'\int_{-1}^1dh''\frac{72}{\pi^4}t'(1-t')
\\
&=2-\frac{288}{\pi^4}\frac{1}{6}=2-\frac{48}{\pi^4}>1,
\end{align*}

and this concludes the proof of property \eqref{derivataneg}, i.e., \textbf{Step 1}.

Now we can prove \textbf{Step 2}, that is, the second summand in \eqref{anteconclusione} divided by $c$ is infinitesimal as $c\to0$.

\textbf{Step 2: } we want to prove that

\begin{align*}
\sup_{h\in[-1,1]}\int_{0<f^h\leq\frac{c}{2\pi}E^{(2)}}d\theta dtdh'E^{(2)}\xrightarrow[c\to0]{}0.
\end{align*}

Arguing by contradiction, there would exist $\varepsilon>0$, a sequence $c_n\in[0,\frac{1}{n}]$ and another sequence $h^n\in[-1,1]$ such that

\begin{align*}
\varepsilon\leq\int_{0< f^{h^n}\leq\frac{c_n}{2\pi}E^{(2)}}E^{(2)}.
\end{align*}

By compactness, there would exist a subsequence ${h}^{n_k}$ of $\{h^n\}$ converging to $\bar h\in[-1,1]$ such that
\[
E^{(2)}\mathbbm{1}_{0< f^{h^{n_k}}\leq \frac{c_{n_k}}{2\pi}}\xrightarrow[k\to+\infty]{}E^{(2)}\mathbbm{1}_{0< f^{\bar h}\leq0}=0\textrm{ everywhere in }(\theta,t,h'),
\]
moreover $E^{(2)}\mathbbm{1}_{0<f^{h^{n_k}}\leq \frac{c_{n_k}}{2\pi}}\leq E^{(2)}$, and therefore this can not occur for dominated convergence of $\{E^{(2)}\mathbbm{1}_{0<f^{h^{n_k}}\leq \frac{c_{n_k}}{2\pi}}\}_k$.

By \eqref{conclusione}, that concludes the proof.
\end{proof}
\subsection{Proof of Theorem \ref{thm:convergenza_v,s,h}.}
Now we have all the intermediate results that we need to prove Theorem \ref{thm:convergenza_v,s,h}.
\begin{proof}
The first result we prove is the statement \eqref{thm1:st1} of the Theorem. That is, we want to prove that for any function $\mu_0\in L^1\cap L^p(\mathbb{T}^1_{2\pi}\times[0,+\infty)\times[-1,1])$

{\tiny\begin{align*}
\left\|\mu_t-\frac{\langle\mu_0\rangle}{2\pi}E\right\|_{L^p(\mathbb{T}^1_{2\pi}\times[0,+\infty)\times[-1,1])}&\leq C\frac{\|\mu_0\|_{L^1(\mathbb{T}^1_{2\pi}\times[0,+\infty)\times[-1,1])}+\|\mu_0\|_{L^p(\mathbb{T}^1_{2\pi}\times[0,+\infty)\times[-1,1])}}{t+1}
\\
&+C\left[\|\mu_0\|_{L^1(\mathbb{T}^1_{2\pi}\times[t/4,+\infty)\times[-1,1])}+\|\mu_0\|_{L^p(\mathbb{T}^1_{2\pi}\times[t/4,+\infty)\times[-1,1])}\right].
\end{align*}}

We prove this only for $p\in[1,+\infty)$, since the proof of the statement in the case $p=\infty$ follows in the same way as the proof for finite $p$. When not better specified, in this proof we will shorten 
\[
\|\cdot\|_{L^p}:=\|\cdot\|_{L^p(\mathbb{T}^1_{2\pi}\times[0,+\infty)\times[-1,1])},
\]
and we will often use also the notation
\[
\int d\theta dsdh:=\int_{\mathbb{T}^1_{2\pi}}d\theta\int_0^{\infty}ds\int_{-1}^1dh.
\]
To begin with, we recall Proposition \ref{prop:musvarphi}, and inserting it into \eqref{rapp:ev_x} we get

{\scriptsize\begin{align}
\nonumber&\mu_t(\theta,s,h)
\\
\nonumber&=\mu_0(\theta,s+t,h)+\int_0^tdt'\int_{-1}^1dh'Q(s+t-t',h|h')\mu_0(\theta+\pi-2\arcsin(h'),t',h')
\\
\nonumber&+\int_0^tdt'\int_{-1}^1dh'Q(s+t-t',h|h')\int_0^{t'}dt''\int_{-1}^1dh'' Q(t'-t'',h'|h'')\mu_0(\theta_1,t'',h'')
\\
\label{addendofico}&+\frac{1}{2\pi}\int_0^tdt'\int_{-1}^1dh'Q(s+t-t',h|h')\int_{\mathbb{T}^1_{2\pi}}d\theta'\int_0^{t'}dt''\int_{-1}^1dh''\mu_0(\theta',t'',h'')
\\
\nonumber&+\int_0^tdt'\int_{-1}^1dh'Q(s+t-t',h|h')\int_{\mathbb{T}^1_{2\pi}}d\theta'\int_0^{t'}dt''\int_{-1}^1dh''\varphi(\theta_2-\theta',t'-t'',h'|h'')\mu_0(\theta',t'',h''),
\end{align}}

with
\begin{align*}
\theta_1&=\theta_1(\theta,h',h''):=\theta+2\pi-2\arcsin(h')-2\arcsin(h''),
\\
\theta_2&=\theta_2(\theta,h'):=\theta+\pi-2\arcsin(h').
\end{align*}

Since the fourth summand in \eqref{addendofico} writes as

\begin{align*}
\eqref{addendofico}&=\frac{1}{2\pi}\int_0^tdt''\int_{-1}^1dh''\int_{\mathbb{T}^1_{2\pi}}d\theta'\mu_0(\theta',t'',h'')\underbrace{\int_{t''}^tdt'\int_{-1}^1dh'Q(s+t-t',h|h')}_{=E(s,h)-E(s+t-t'',h)}
\\
&=\frac{1}{2\pi}\int_{\mathbb{T}^1_{2\pi}}d\theta'\int_0^tdt'\int_{-1}^1dh''\mu_0(\theta',t',h'')[E(s,h)-E(s+t-t',h)],
\end{align*}

 we have
 
{\scriptsize\begin{align}
\nonumber&\mu_t(\theta,s,h)
\\
\label{addendo1}&=\mu_0(\theta,s+t,h)
\\
\label{addendo2}&+\int_0^tdt'\int_{-1}^1dh'Q(s+t-t',h|h')\mu_0(\theta+\pi-2\arcsin(h'),t',h')
\\
\label{addendo3}&+\int_0^tdt'\int_{-1}^1dh'Q(s+t-t',h|h')\int_0^{t'}dt''\int_{-1}^1dh''Q(t'-t'',h'|h'')\mu_0(\theta_1,t'',h'')
\\
\label{addendo4}&+\frac{E(s,h)}{2\pi}\int_{\mathbb{T}^1_{2\pi}}d\theta'\int_0^tdt'\int_{-1}^1dh''\mu_0(\theta',t',h'')
\\
\label{addendo5}&-\frac{1}{2\pi}\int_0^tdt' E(s+t-t',h)\int_{-1}^1dh''\int_{\mathbb{T}^1_{2\pi}}d\theta'\mu_0(\theta',t',h'')
\\
\label{addendo6}&+\int_0^tdt'\int_{-1}^1dh'Q(s+t-t',h|h')\int_{\mathbb{T}^1_{2\pi}}d\theta'\int_0^{t'}dt''\int_{-1}^1dh''\varphi(\theta_2-\theta',t'-t'',h'|h'')\mu_0(\theta',t'',h'').
\end{align}}

To prove inequality \eqref {thm1:st1} in the first part of the Theorem we have to estimate all the terms in the expression above.

As for the first one we have

\begin{align*}
\|\eqref{addendo1}\|_{L^p}^p&=\int_{\mathbb{T}^1_{2\pi}}d\theta\int_0^{\infty}ds\int_{-1}^1dh|\mu_0(\theta,s+t,h)|^p
\\
&=\int_{\mathbb{T}^1_{2\pi}}d\theta\int_t^{\infty}ds\int_{-1}^1dh|\mu_0(\theta,s,h)|^p=\|\mu_0\|_{L^p(\mathbb{T}^1_{2\pi}\times[t,+\infty)\times[-1,1])}^p,
\end{align*}

and therefore the estimate in \eqref{thm1:st1} applies to it.

A very simple but crucial inequality in the next estimates is the following: if $f\geq0$, then by Jensen's inequality applied when the measure is $dxf(x)$ over the set $X$ we have

\begin{align*}
\left|\int_{X}dxf(x)g(x)\right|^p\leq\int_{X}dxf(x)|g(x)|^p\left(\int_{X}dxf(x)\right)^{p-1}.
\end{align*}

Applying Jensen's inequality with $f=Q$, $g=\mu_0$ and $X=[t_1,t_2]\times[-1,1]$, one gets

{\small\begin{align}
\nonumber&\left|\int_{t_1}^{t_2}dt'\int_{-1}^1dh'Q(t-t',h|h')\mu_0(\theta,t',h')\right|
\\
\nonumber&\leq\int_{t_1}^{t_2}dt'\int_{-1}^1dh'Q(t-t',h|h')|\mu_0(\theta,t',h')|^p\left(\int_{t_1}^{t_2}dt'\int_{-1}^1dh'Q(t-t',h|h')\right)^{p-1}
\\
\label{Jensencons}&=\int_{t_1}^{t_2}dt'\int_{-1}^1dh'Q(t-t',h|h')|\mu_0(\theta,t',h')|^p\left(E(t-t_2,h)-E(t-t_1,h)\right)^{p-1}.
\end{align}}

Going back to the proof, the second term in the sum can be bounded as follows

{\tiny\begin{align}
\nonumber\|\eqref{addendo2}\|_{L^p}&\leq\left(\int_{\mathbb{T}^1_{2\pi}}d\theta\int_0^{\infty}ds\int_{-1}^1dh\left|\int_0^tdt'\int_{-1}^1dh'Q(s+t-t',h|h')|\mu_0(\theta+\pi-2\arcsin(h'),t',h')|\right|^p\right)^{\frac{1}{p}}
\\
\label{addendo2first}&\leq\left(\int_{\mathbb{T}^1_{2\pi}}d\theta\int_0^{\infty}ds\int_{-1}^1dh\left|\int_0^{\frac{t}{2}}dt'\int_{-1}^1dh'Q(s+t-t',h|h')|\mu_0(\theta+\pi-2\arcsin(h'),t',h')|\right|^p\right)^{\frac{1}{p}}
\\
\label{addendo2second}&+\left(\int_{\mathbb{T}^1_{2\pi}}d\theta\int_0^{\infty}ds\int_{-1}^1dh\left|\int_{\frac{t}{2}}^tdt'\int_{-1}^1dh'Q(s+t-t',h|h')|\mu_0(\theta+\pi-2\arcsin(h'),t',h')|\right|^p\right)^{\frac{1}{p}}.
\end{align}}

We shall estimate separately \eqref{addendo2first} and \eqref{addendo2second}. As for the first one, thanks to Jensen's inequality used as in \eqref{Jensencons}, we have

{\small\begin{align*}
&\eqref{addendo2first}^p\leq\int_{\mathbb{T}^1_{2\pi}}d\theta\int_0^{\frac{t}{2}}dt'\int_{-1}^1dh'|\mu_0(\theta,t',h')|^p\int_0^{\infty}ds\int_{-1}^1dhQ(s+t-t',h|h')\xi_0
\end{align*}}

with

\begin{align*}
\xi_0&=\xi_0(s,t,h):=\left(E(s+\frac{t}{2},h)-E(s+t,h)\right)^{p-1}\leq E(s+\frac{t}{2},h)^{p-1}
\\
&\leq\frac{(2C)^{p-1}}{(t+2)^{p-1}}\textrm{ by Lemma \ref{decE}}.
\end{align*}

Hence

{\small\begin{align*}\eqref{addendo2first}^p&\leq\frac{(2C)^{p-1}}{(t+2)^{p-1}}\int_{\mathbb{T}^1_{2\pi}}d\theta\int_0^{\frac{t}{2}}dt'\int_{-1}^1dh'\underbrace{E(t-t',h')}_{\leq\frac{C}{t-t'+1}\leq\frac{2C}{t+2}\textrm{ by Lemma \ref{decE}}}|\mu_0(\theta,{t'},h')|^p
\\
&\leq\frac{(2C)^p}{(t+2)^p}\|\mu_0\|_{L^p(\mathbb{T}^1_{2\pi}\times[0,+\infty)\times[-1,1])}^p,
\end{align*}}

while the second one satisfies

{\tiny\begin{align*}
&\eqref{addendo2second}^p
\\
&\leq\int_{\mathbb{T}^1_{2\pi}}d\theta\int_{\frac{t}{2}}^tdt'\int_{-1}^1dh'|\mu_0(\theta,t',h')|^p\int_0^{\infty}ds\int_{-1}^1dhQ(s+t-t',h|h')\underbrace{\left(E(s,h)-E(s+\frac{t}{2},h)\right)^{p-1}}_{\leq1}
\\
&\leq\int_{\mathbb{T}^1_{2\pi}}d\theta\int_{\frac{t}{2}}^tdt'\int_{-1}^1dh'\underbrace{E(t-t',h')}_{\leq1}|\mu_0(\theta,{t'},h')|^p\leq\|\mu_0\|_{L^p(\mathbb{T}^1_{2\pi}\times[\frac{t}{2},+\infty)\times[-1,1])}^p.
\end{align*}}

which ends the estimate of \eqref{addendo2} since
\[
\|\eqref{addendo2}\|_{L^p}\leq\eqref{addendo2first}+\eqref{addendo2second}\leq\frac{2C}{t+2}\|\mu_0\|_{L^p}+\|\mu_0\|_{L^p(\mathbb{T}^1_{2\pi}\times[\frac{t}{2},+\infty)\times[-1,1])}.
\]
Then, as for \eqref{addendo3}, if we split $\{t'\in[0,t]\}$ in $\{t'\in[0,\frac{t}{2}]\}\cup\{t'\in[\frac{t}{2},t]\}$ and $\{t''\in[0,t']\}$ in $\{t''\in[0,\frac{t'}{2}]\}\cup\{t''\in[\frac{t'}{2},t']\}$, we get

{\tiny\begin{align}
\nonumber\|&\eqref{addendo3}\|_{L^p}
\\
\nonumber&\leq\left\|\int_0^tdt'\int_{-1}^1dh'Q(\cdot+t-t',\cdot|h')\int_0^{t'}dt''\int_{-1}^1dh''Q(t'-t'',h'|h'')|\mu_0(\theta_1,t'',h'')|\right\|_{L^p}
\\
\label{addendo3first}&\leq\left(\int d\theta dsdh\left|\int_0^{\frac{t}{2}}dt'\int_{-1}^1dh'Q(s+t-t',h|h')\int_0^{t'}dt''\int_{-1}^1dh''Q(t'-t'',h'|h'')|\mu_0(\theta_1,t'',h'')|\right|^p\right)^{\frac{1}{p}}
\\
\label{addendo3second}&+\left(\int d\theta dsdh\left|\int_{\frac{t}{2}}^tdt'\int_{-1}^1dh'Q(s+t-t',h|h')\int_0^{\frac{t'}{2}}dt''\int_{-1}^1dh''Q(t'-t'',h'|h'')|\mu_0(\theta_1,t'',h'')|\right|^p\right)^{\frac{1}{p}}
\\
\label{addendo3third}&+\left(\int d\theta dsdh\left|\int_{\frac{t}{2}}^tdt'\int_{-1}^1dh'Q(s+t-t',h|h')\int_{\frac{t'}{2}}^{t'}dt''\int_{-1}^1dh''Q(t'-t'',h'|h'')|\mu_0(\theta_1,t'',h'')|\right|^p\right)^{\frac{1}{p}}.
\end{align}}

with $\theta_1$ as before. Now estimate the three summands separately.

Let us begin with the first one: using twice Jensen's inequality \eqref{Jensencons}, one gets

{\tiny\begin{align*}
\eqref{addendo3first}^p\leq\int d\theta dsdh\int_0^{\frac{t}{2}}dt'\int_{-1}^1dh'Q(s+t-t',h|h')\int_0^{t'}dt''\int_{-1}^1dh''Q(t'-t'',h'|h'')|\mu_0(\theta,t'',h'')|^p\xi_1,
\end{align*}}

with

\begin{align*}
\xi_1&=\xi_1(s,t,h,t',h'):=\left(E(s+\frac{t}{2},h)-E(s+t,h)\right)^{p-1}(1-E(t',h'))^{p-1}
\\
&\leq E(s+\frac{t}{2},h)^{p-1}\leq\frac{(2C)^{p-1}}{(t+2)^{p-1}}\textrm{ by Lemma \ref{decE}}.
\end{align*}

Therefore

{\footnotesize\begin{align*}
\eqref{addendo3first}^p\leq\frac{(2C)^{p-1}}{(t+2)^{p-1}}\int_{\mathbb{T}^1_{2\pi}}d\theta\int_0^{\frac{t}{2}}dt''\int_{-1}^1dh''|\mu_0(\theta,t'',h'')|^p\int_{t''}^{\frac{t}{2}}dt'\int_{-1}^1dh''Q(t'-t'',h'|h'')\xi_{1,2}
\end{align*}}

with

\begin{align*}
\xi_{1,2}&=\xi_{1,2}(t,t',h'):=\int_0^{\infty}ds\int_{-1}^1dhQ(s+t-t',h|h')=E(t-t',h')\leq E(\frac{t}{2},h')
\\
&\leq\frac{2C}{t+2}\textrm{ by Lemma \ref{decE}}.
\end{align*}

Hence
{\small\begin{align*}
\eqref{addendo3first}^p&\leq\frac{(2C)^p}{(t+2)^p}\int_{\mathbb{T}^1_{2\pi}}d\theta\int_0^{\frac{t}{2}}dt''\int_{-1}^1dh''|\mu_0(\theta,t'',h'')|^p\underbrace{\int_{t''}^{\frac{t}{2}}dt'\int_{-1}^1dh''Q(t'-t'',h'|h'')}_{\leq1}
\\
&\leq\frac{(2C)^p}{(t+2)^p}\|\mu_0\|_{L^p}^p.
\end{align*}}

As for \eqref{addendo3second}, we use twice Jensen's inequality \eqref{Jensencons} to get

{\scriptsize\begin{align*}
\eqref{addendo3second}^p\leq\int d\theta dsdh\int_{\frac{t}{2}}^tdt'\int_{-1}^1dh'Q(s+t-t',h|h')\int_0^{\frac{t'}{2}}dt''\int_{-1}^1dh''Q(t'-t'',h'|h'')|\mu_0(\theta,t'',h'')|^p\xi_2,
\end{align*}}

with

\begin{align*}
\xi_2&=\xi_2(s,t,h,t',h'):=\left(E(s,h)-E\left(s+\frac{t}{2},h\right)\right)^{p-1}\left(E\left(\frac{t'}{2},h'\right)-E(t',h')\right)^{p-1}
\\
&\leq E\left(\frac{t'}{2},h'\right)^{p-1}\leq\frac{(2C)^{p-1}}{(t'+2)^{p-1}}\leq\frac{(4C)^{p-1}}{(t+4)^{p-1}}\textrm{ by Lemma \ref{decE}}.
\end{align*}

Therefore, if
\[
\bar t:=\max\left\{2t'',\frac{t}{2}\right\},
\]
we have

{\tiny\begin{align*}
&\eqref{addendo3second}^p
\\
&\leq\frac{(4C)^{p-1}}{(t+4)^{p-1}}\int_{\mathbb{T}^1_{2\pi}}d\theta\int_0^tdt''\int_{-1}^1dh''|\mu_0(\theta,t'',h'')|^p\int_{\bar t}^tdt'\int_{-1}^1dh'E(t-t',h')\underbrace{Q(t'-t'',h'|h'')}_{\leq\frac{2C}{t'+2}\leq\frac{4C}{t+4}\textrm{ from \eqref{decQ}}}
\\
&\leq\frac{(4C)^p}{(t+4)^p}\int_{\mathbb{T}^1_{2\pi}}d\theta\int_0^tdt''\int_{-1}^1dh''|\mu_0(\theta,t'',h'')|^p\underbrace{\int_{\max\{2t'',\frac{t}{2}\}}^tdt'\int_{-1}^1dh'E(t-t',h')}_{\leq1}
\\
&\leq\frac{(4C)^p}{(t+4)^p}\|\mu_0\|_{L^p}^p.
\end{align*}}

For \eqref{addendo3third}, applying twice Jensen's inequality as in \eqref{Jensencons}, we also have

{\tiny\begin{align*}
\eqref{addendo3third}^p\leq\int d\theta dsdh\int_{\frac{t}{2}}^tdt'\int_{-1}^1dh'Q(s+t-t',h|h')\int_{\frac{t'}{2}}^{t'}dt''\int_{-1}^1dh''Q(t'-t'',h'|h'')|\mu_0(\theta,t'',h'')|^p\xi_3,
\end{align*}}

with

\begin{align*}
\xi_3=\xi_3(s,t,h,t',h'):=\left(E(s,h)-E\left(s+\frac{t}{2},h\right)\right)^{p-1}\left(1-E\left(\frac{t'}{2},h'\right)\right)^{p-1}\leq1.
\end{align*}

Hence

{\scriptsize\begin{align*}
\eqref{addendo3third}^p&\leq\int_{\mathbb{T}^1_{2\pi}}d\theta\int_{\frac{t}{4}}^tdt''\int_{-1}^1dh''|\mu_0(\theta,t'',h'')|^p\underbrace{\int_{\max\{\frac{t}{2},t''\}}^{\min\{t,2t''\}}dt'\int_{-1}^1dh'\underbrace{E(t-t',h')}_{\leq1}Q(t'-t'',h'|h'')}_{\leq1}
\\
&\leq\|\mu_0\|_{L^p(\mathbb{T}^1_{2\pi}\times[\frac{t}{4},+\infty)\times[-1,1])}^p.
\end{align*}}

Thus 
{\small\[
\|\eqref{addendo3}\|_{L^p}\leq\eqref{addendo3first}+\eqref{addendo3second}+\eqref{addendo3third}\leq\left[\frac{2C}{t+2}+\frac{4C}{t+4}\right]\|\mu_0\|_{L^p}+\|\mu_0\|_{L^p(\mathbb{T}^1_{2\pi}\times[\frac{t}{4},+\infty)\times[-1,1])},
\]}
and we have ended with the estimate of \eqref{addendo3}.

Now we shall bound \eqref{addendo4}, and, since $E\leq1$ implies $E^p\leq E$, we get

{\small\begin{align*}
&\left\|\eqref{addendo4}-\frac{\langle\mu_0\rangle}{2\pi}E\right\|_{L^p}^p
\\
&=\frac{1}{2\pi}\int_{\mathbb{T}^1_{2\pi}}d\theta\int_0^{\infty}ds\int_{-1}^1dh\left|\int_{\mathbb{T}^1_{2\pi}}d\theta'\int_0^tdt'\int_{-1}^1dh''\mu_0(\theta',{t'},h'')-\langle\mu_0\rangle\right|^pE(s,h)^p
\\
&=\left|\int_{\mathbb{T}^1_{2\pi}}d\theta'\int_0^tdt'\int_{-1}^1dh''\mu_0(\theta',{t'},h'')-\langle\mu_0\rangle\right|^p\int_0^{\infty}ds\int_{-1}^1dhE(s,h)
\\
&\leq\left(\int_{\mathbb{T}^1_{2\pi}}d\theta'\int_t^{\infty}dt'\int_{-1}^1dh''|\mu_0(\theta',{t'},h'')|\right)^p=\|\mu_0\|_{L^1(\mathbb{T}^1_{2\pi}\times[t,+\infty)\times[-1,1])}^p.
\end{align*}}

For the fifth summand one can write

{\footnotesize\begin{align}
\nonumber\|\eqref{addendo5}\|_{L^p}&\leq\left(\int_{\mathbb{T}^1_{2\pi}}d\theta\int_0^{\infty}ds\int_{-1}^1dh\left|\frac{1}{2\pi}\int_0^tdt' E(s+t-t',h)\int_{-1}^1dh''\int_{\mathbb{T}^1_{2\pi}}d\theta'|\mu_0(\theta',{t'},h'')|\right|^p\right)^{\frac{1}{p}}
\\
\label{addendo5first}&\leq(2\pi)^{\frac{1}{p}-1}\left(\int_0^{\infty}ds\int_{-1}^1dh\left|\int_0^{\frac{t}{2}}dt' E(s+t-t',h)\int_{-1}^1dh''\int_{\mathbb{T}^1_{2\pi}}d\theta'|\mu_0(\theta',{t'},h'')|\right|^p\right)^{\frac{1}{p}}
\\
\label{addendo5second}&+(2\pi)^{\frac{1}{p}-1}\left(\int_0^{\infty}ds\int_{-1}^1dh\left|\int_{\frac{t}{2}}^tdt' E(s+t-t',h)\int_{-1}^1dh''\int_{\mathbb{T}^1_{2\pi}}d\theta'|\mu_0(\theta',{t'},h'')|\right|^p\right)^{\frac{1}{p}},
\end{align}}

and we bound again \eqref{addendo5first} and \eqref{addendo5second} separately. This way we get

{\small\begin{align*}
\eqref{addendo5first}^p&=(2\pi)^{1-p}\int_0^{\infty}ds\int_{-1}^1dh\left|\int_0^{\frac{t}{2}}dt' \underbrace{E(s+t-t',h)}_{\leq E(s+\frac{t}{2},h)}\int_{-1}^1dh''\int_{\mathbb{T}^1_{2\pi}}d\theta'|\mu_0(\theta',{t'},h'')|\right|^p
\\
&\leq(2\pi)^{1-p}\int_0^{\infty}ds\int_{-1}^1dh\underbrace{E(s+\frac{t}{2},h)^p}_{\leq\frac{C^p}{(s+\frac{t}{2}+1)^p}\mathbbm{1}(s+\frac{t}{2}\leq\frac{1}{1-|h|})\textrm{ by Lemma \ref{decE}}}\|\mu_0\|_{L^1}^p
\\
&\leq C^p\|\mu_0\|_{L^1}^p\int_0^{\infty}ds\frac{1}{(s+\frac{t}{2}+1)^{p+1}}\leq \frac{(2C)^p}{p(t+2)^p}\|\mu_0\|_{L^1}^p,
\end{align*}}

and, since $E(s,h)\leq1$ implies $ E(s,h)^p\leq E(s,h)$, we also have

{\scriptsize\begin{align*}
\eqref{addendo5second}^p&=(2\pi)^{1-p}\int_0^{\infty}ds\int_{-1}^1dh\left|\int_{\frac{t}{2}}^tdt' \underbrace{E(s+t-t',h)}_{\leq E(s,h)}\int_{-1}^1dh''\int_{\mathbb{T}^1_{2\pi}}d\theta'|\mu_0(\theta',{t'},h'')|\right|^p
\\
&\leq(2\pi)^{1-p}\int_0^{\infty}ds\int_{-1}^1dhE(s,h)^p\|\mu_0\|_{L^1(\mathbb{T}^1_{2\pi}\times[\frac{t}{2},+\infty)\times[-1,1])}^p\leq \|\mu_0\|_{L^1(\mathbb{T}^1_{2\pi}\times[\frac{t}{2},+\infty)\times[-1,1])}^p,
\end{align*}}

and therefore

\begin{align*}
\|\eqref{addendo5}\|_{L^p}\leq\eqref{addendo5first}+\eqref{addendo5second}\leq \frac{2C}{t+2}\|\mu_0\|_{L^1}+C\|\mu_0\|_{L^1(\mathbb{T}^1_{2\pi}\times[\frac{t}{2},+\infty)\times[-1,1])}.
\end{align*}

We can estimate the sixth term in the same way as we estimated \eqref{addendo3}: we split the integral with respect to $t',t''$ in
\[
[0,t]=\left[0,\frac{t}{2}\right]\cup\left[\frac{t}{2},t\right]\textrm{ and }[0,t']=\left[0,\frac{t'}{2}\right]\cup\left[\frac{t'}{2},t'\right],
\]
in such a way to have

{\tiny\begin{align}
\nonumber&\|\eqref{addendo6}\|_{L^p}
\\
\nonumber&\leq\left\|\int_0^tdt'\int_{-1}^1dh'Q(\cdot+t-t',\cdot|h')\int_{\mathbb{T}^1_{2\pi}}d\theta'\int_0^{t'}dt''\int_{-1}^1dh''|\varphi(\theta_2-\theta',t'-t'',h'|h'')||\mu_0(\theta',t'',h'')|\right\|_{L^p}
\\
\label{addendo6first}&\leq\left\|\int_0^{\frac{t}{2}}dt'\int_{-1}^1dh'Q(\cdot+t-t',\cdot|h')\int_{\mathbb{T}^1_{2\pi}}d\theta'\int_0^{t'}dt''\int_{-1}^1dh''|\varphi(\theta_2-\theta',t'-t'',h'|h'')||\mu_0(\theta',t'',h'')|\right\|_{L^p}
\\
\label{addendo6second}&+\left\|\int_{\frac{t}{2}}^tdt'\int_{-1}^1dh'Q(\cdot+t-t',\cdot|h')\int_{\mathbb{T}^1_{2\pi}}d\theta'\int_0^{\frac{t'}{2}}dt''\int_{-1}^1dh''|\varphi(\theta_2-\theta',t'-t'',h'|h'')||\mu_0(\theta',t'',h'')|\right\|_{L^p}
\\
\label{addendo6third}&+\left\|\int_{\frac{t}{2}}^tdt'\int_{-1}^1dh'Q(\cdot+t-t',\cdot|h')\int_{\mathbb{T}^1_{2\pi}}d\theta'\int_{\frac{t'}{2}}^{t'}dt''\int_{-1}^1dh''|\varphi(\theta_2-\theta',t'-t'',h'|h'')||\mu_0(\theta',t'',h'')|\right\|_{L^p},
\end{align}}

and then we use again Jensen's inequality \eqref{Jensencons}.

We begin with \eqref{addendo6first}. Arguing as we said, we get

{\scriptsize\begin{align*}
\eqref{addendo6first}^p\leq\int d\theta dsdh\int_0^{\frac{t}{2}}dt'\int_{-1}^1dh'Q(s+t-t',h|h')\left|\int_{\mathbb{T}^1_{2\pi}}d\theta'\int_0^{t'}dt''\int_{-1}^1dh''|\varphi||\mu_0(\theta',t'',h'')|\right|^p\xi_4,
\end{align*}}

with
\[
|\varphi|=|\varphi(\theta_2-\theta',t'-t'',h'|h'')|\leq C\textrm{ by Proposition \ref{prop:musvarphi}},
\]
and 

\begin{align*}
\xi_4&:=\xi_4(s,t,h)=\left(E\left(s+\frac{t}{2},h\right)-E(s+t,h)\right)^{p-1}
\\
&\leq\frac{C^{p-1}}{(s+\frac{t}{2}+1)^{p-1}}\mathbbm{1}\left(s+\frac{t}{2}\leq\frac{1}{1-|h|}\right)\textrm{ by Lemma \ref{decE}}.
\end{align*}

In this way, we have

{\tiny\begin{align*}
&\eqref{addendo6first}^p
\\
&\leq2\pi C^{2p-1}\int_0^{\infty}ds\int_{-1}^1dh\int_0^{\frac{t}{2}}dt'\int_{-1}^1dh'Q(s+t-t',h|h')\frac{1}{(s+\frac{t}{2}+1)^{p-1}}\mathbbm{1}(s+\frac{t}{2}\leq\frac{1}{1-|h|})\|\mu_0\|_{L^1}^p
\\
&=2\pi C^{2p-1}\|\mu_0\|_{L^1}^p\int_0^{\infty}ds\frac{1}{(s+\frac{t}{2}+1)^{p-1}}\int_{-1}^1dh\mathbbm{1}(1-|h|\leq\frac{1}{s+\frac{t}{2}})\underbrace{(E(s+\frac{t}{2},h)-E(s+t,h))}_{\leq E(s+\frac{t}{2},h)\leq\frac{C}{s+\frac{t}{2}+1}\textrm{ by Lemma \ref{decE}}}
\\
&\leq2\pi C^{2p}\|\mu_0\|_{L^1}^p\int_{\frac{t}{2}+1}^{\infty}ds\frac{1}{s^{p+1}}=2\pi \frac{2^pC^{2p}}{p(t+2)^p}\|\mu_0\|_{L^1}^p.
\end{align*}}

We still have to estimate the other two terms: the second one can be bounded using first \eqref{Jensencons} and then Proposition \ref{prop:musvarphi}, as follows. We have

{\scriptsize\begin{align*}
\eqref{addendo6second}^p\leq\int d\theta dsdh\int_{\frac{t}{2}}^tdt'\int_{-1}^1dh'Q(s+t-t',h|h')\left|\int_{\mathbb{T}^1_{2\pi}}d\theta'\int_0^{\frac{t'}{2}}dt''\int_{-1}^1dh''|\varphi||\mu_0(\theta',t'',h'')|\right|^p\xi_5
\end{align*}}

with
\[
\xi_5=\xi_5(t,s,h):=\left(E(s,h)-E\left(s+\frac{t}{2},h\right)\right)^{p-1}\leq E(s,h)^{p-1}\leq1.
\]
Since
\[
|\varphi|=|\varphi(\theta_2-\theta',t'-t'',h'|h'')|\leq \frac{C}{t'-t''+1}\leq\frac{2C}{t'+2}\leq\frac{4C}{t+4}\textrm{ by Proposition \ref{prop:musvarphi}},
\]
we have

\begin{align*}
\eqref{addendo6second}^p&\leq\frac{(4C)^p}{(t+4)^p}2\pi\|\mu_0\|_{L^1}^p\int_0^{\infty}ds\int_{-1}^1dh\underbrace{\int_{\frac{t}{2}}^tdt'\int_{-1}^1dh'Q(s+t-t',h|h')}_{=E(s,h)-E(s+\frac{t}{2},h)\leq E(s,h)}
\\
&\leq\frac{(4C)^p}{(t+4)^p}2\pi\|\mu_0\|_{L^1}^p\int_0^{\infty}ds\int_{-1}^1dhE(s,h)=\frac{(4C)^p}{(t+4)^p}2\pi\|\mu_0\|_{L^1}^p.
\end{align*}

Finally, only \eqref{addendo6third} is missing: we estimate it below.

{\scriptsize\begin{align*}
&\eqref{addendo6third}^p\leq\int d\theta dsdh\int_{\frac{t}{2}}^tdt'\int_{-1}^1dh'Q(s+t-t',h|h')\xi_6\left|\int_{\mathbb{T}^1_{2\pi}}d\theta'\int_{\frac{t'}{2}}^{t'}dt''\int_{-1}^1dh''|\varphi||\mu_0(\theta',t'',h'')|\right|^p,
\end{align*}}

with
\[
\xi_6=\xi_6(t,s,h):=\left(E(s,h)-E(s+\frac{t}{2},h)\right)^{p-1}\leq1,
\]
and
\[
|\varphi|=|\varphi(\theta_2-\theta',t'-t'',h'|h'')|\leq C\textrm{ by Proposition \ref{prop:musvarphi}},
\]
and therefore

{\small\begin{align*}
\\
\eqref{addendo6third}^p&\leq2\pi C^p\int_0^{\infty}ds\int_{-1}^1dh\int_{\frac{t}{2}}^tdt'\int_{-1}^1dh'Q(s+t-t',h|h')\|\mu_0\|_{L^1(\mathbb{T}^1_{2\pi}\times[\frac{t'}{2},+\infty)\times[-1,1])}^p
\\
&\leq2\pi C^p\|\mu_0\|_{L^1(\mathbb{T}^1_{2\pi}\times[\frac{t}{4},+\infty)\times[-1,1])}^p\int_0^{\infty}ds\int_{-1}^1dh\underbrace{(E(s,h)-E(s+\frac{t}{2},h))}_{\leq E(s,h)}
\\
&\leq 2\pi C^p\|\mu_0\|_{L^1(\mathbb{T}^1_{2\pi}\times[\frac{t}{4},+\infty)\times[-1,1])}^p.
\end{align*}}

Thus we get

\begin{align*}
\|\eqref{addendo6}\|_{L^p}&\leq\eqref{addendo6first}+\eqref{addendo6second}+\eqref{addendo6third}
\\
&\leq (2\pi)^\frac{1}{p}\left[\left(\frac{2C^2}{t+2}+\frac{4C}{t+4}\right)\|\mu_0\|_{L^1}+C\|\mu_0\|_{L^1(\mathbb{T}^1_{2\pi}\times[\frac{t}{4},+\infty)\times[-1,1])}\right],
\end{align*}

which terminates the proof of the first statement \eqref{thm1:st1} of Theorem \ref{thm:convergenza_v,s,h}.

Now it remains to prove the second part of the Theorem, that is, \eqref{thm1:st2}, which states that in the particular case $\mu_0(\theta,s,h)=\mu_{in}(\theta)E(s,h)$
\[
\left\|\mu_t-\frac{\langle\mu_0\rangle}{2\pi}E\right\|_{L^p(\mathbb{T}^1_{2\pi}\times[0,+\infty)\times[-1,1])}\leq\frac{C}{t+1}\|\mu_{in}\|_{L^p(\mathbb{T}^1_{2\pi})}.
\]
To this purpose, notice that we can estimate the four summands in the right-hand side of \eqref{thm1:st1} as follows. For any $q\geq1$ we have
\[
\|\mu_0\|_{L^q(\mathbb{T}^1_{2\pi}\times[A,+\infty)\times[-1,1])}=\|\mu_{in}\|_{L^q(\mathbb{T}^1_{2\pi})}\|E\|_{L^q([A,+\infty)\times[-1,1])},
\]
and the terms in the right-hand side of \eqref{thm1:st1} correspond to the four cases $q=1$ and  $A=0$, $q=1$ and $ A=\frac{t}{4}$, $q=p$ and $ A=0$, $q=p$ and $ A=\frac{t}{4}$. Moreover $\|E\|_{L^{\infty}([A,+\infty)\times[-1,1])}\leq\frac{C}{A+1}$ from Lemma \ref{decE}, while for finite $q$ we have

\begin{align*}
\|E\|_{L^q([A,+\infty)\times[-1,1])}&=\left(\int_A^{\infty}ds\int_{-1}^1dh E(s,h)^q\right)^{\frac{1}{q}}
\\
&\leq C\left(\int_A^{\infty}ds\frac{1}{(s+1)^{q+1}} \right)^{\frac{1}{q}}\leq\frac{C}{A+1},
\end{align*}

and therefore:
\begin{itemize}
\item if $q=1$ and $A=0$ then
\[
\|\mu_0\|_{L^q(\mathbb{T}^1_{2\pi}\times[A,+\infty)\times[-1,1])}\leq\|\mu_{in}\|_{L^1(\mathbb{T}^1_{2\pi})}\leq\|\mu_{in}\|_{L^p(\mathbb{T}^1_{2\pi})};
\]
\item if $q=p$ and $A=0$ then

\begin{align*}
\|\mu_0\|_{L^q(\mathbb{T}^1_{2\pi}\times[A,+\infty)\times[-1,1])}&\leq\|\mu_{in}\|_{L^p(\mathbb{T}^1_{2\pi})}\underbrace{\|E\|_{L^p([0,+\infty)\times[-1,1])}}_{\leq(\int dsdh E(s,h)^p)^{\frac{1}{p}}\leq(\int dsdh E(s,h))^{\frac{1}{p}}=1}
\\
&\leq\|\mu_{in}\|_{L^p(\mathbb{T}^1_{2\pi})};
\end{align*}

\item if $q=1$ and $A=\frac{t}{4}$ then
\[
\|\mu_0\|_{L^q(\mathbb{T}^1_{2\pi}\times[A,+\infty)\times[-1,1])}\leq\underbrace{\|\mu_{in}\|_{L^1(\mathbb{T}^1_{2\pi})}}_{\leq\|\mu_{in}\|_{L^p(\mathbb{T}^1_{2\pi})}}\frac{4C}{t+4}\leq\frac{4C}{t+4}\|\mu_{in}\|_{L^p(\mathbb{T}^1_{2\pi})};
\]
\item if $q=p$ and $A=\frac{t}{4}$ then
\[
\|\mu_0\|_{L^q(\mathbb{T}^1_{2\pi}\times[A,+\infty)\times[-1,1])}\leq\frac{4C}{t+4}\|\mu_{in}\|_{L^p(\mathbb{T}^1_{2\pi})}.
\]
\end{itemize}
By substituting these estimates into \eqref{thm1:st1} we get \eqref{thm1:st2}.
\end{proof}

\newpage
\section{The long time evolution of a density depending on $(x,\theta,s,h)$.}\label{xthetash}
In this Section we want to prove Theorems \ref{thm:mutk}, \ref{thm:conv_x,v,s,h} and \ref{thm:conv_R2}, that we recall below for clarity. Let the Fourier coefficients of a mild solution of \eqref{rapp:ev} with initial datum $\mu_0$ be
\[
\mu_t^k(\theta,s,h):=\int dxe^{2\pi ik\cdot x}\mu_t(x,\theta,s,h),\qquad k\in\mathbb{R}^2\textrm{ or }\mathbb{Z}^2,
\]
where the integral over $x$ is performed over $x\in\mathbb{R}^2$ or $x\in\mathbb{Z}^2$ depending on whenever the solution is defined for $x\in\mathbb{R}^2$ or $x\in\mathbb{T}^2$ and the same holds for the ambiguity $k\in\mathbb{R}^2$ or $k\in\mathbb{Z}^2$. Theorem \ref{thm:mutk} states there exists a constant $C>0$, depending only $Q$ and not on $p$, such that for any $k\in\mathbb{R}^2,k\neq(0,0)$ 
{\tiny\[
\left\|\mu_t^k\right\|_{L^p}\leq\frac{C}{\min\{1,|k|^6\}}\left[\frac{\|\mu_0^k\|_{L^1}+\|\mu_0^k\|_{L^p}}{t+1}+\|\mu_0^k\|_{L^1(\mathbb{T}^1_{2\pi}\times[\frac{t}{4},+\infty)\times[-1,1])}+\|\mu_0^k\|_{L^p(\mathbb{T}^1_{2\pi}\times[\frac{t}{4},+\infty)\times[-1,1])}\right].
\]}
When not specified, the $L^1$ and $L^p$ norms are taken where the coefficients are defined, that is $\mathbb{T}^1_{2\pi}\times[0,+\infty)\times[-1,1]$.

On the flat torus $\mathbb{T}^2$, this result combined with Theorem \ref{thm:convergenza_v,s,h} implies Theorem \ref{thm:conv_x,v,s,h}, that is, under suitable conditions on $\mu_0$, whenever $p$ is finite
\[
\left\|\mu_t-\frac{\langle\mu_0\rangle}{2\pi}E\right\|_{L^p(\mathbb{T}^2\times\mathbb{T}^1_{2\pi}\times[0,+\infty)\times[-1,1])}\xrightarrow[t\to+\infty]{}0,
\]
and the same result holds for the weak-$*$ convergence in $L^{\infty}$.

On $\mathbb{R}^2$, Theorem \ref{thm:mutk} implies that $\mu_t\xrightarrow[t\to+\infty]{}0$ (Theorem \ref{thm:conv_R2}) in a very weak sense (of course not in $L^1$ because the total mass is preserved).
\subsection{Long time behavior of Fourier coefficients.} The first goal is to prove Theorem \ref{thm:mutk}. To this purpose we make several steps. First we characterize the time evolution of a particular Fourier coefficient: by multiplying equation \eqref{rapp:ev} times $e^{2\pi ik\cdot x}$ and integrating with respect to $x\in\mathbb{T}^2$ (or $x\in\mathbb{R}^2$) we get

{\footnotesize\begin{align}
\nonumber\mu_t^k(\theta,s,h)&=e^{2\pi itk\cdot v(\theta)}\mu_0^k(\theta,s+t,h)
\\
\label{rapp:ev_mutk}&+\int_0^tdt'\int_{-1}^1dh'Q(s+t-t',h|h')e^{2\pi i(t-t')k\cdot v(\theta)}\mu_{t'}^k(\theta+\pi-2\arcsin(h'),0,h'),
\end{align}}

and evaluating the obtained result at $s=0$, we obtain

\begin{align}
\nonumber\mu_t^k(\theta,0,h)&=e^{2\pi itk\cdot v(\theta)}\mu_0^k(\theta,t,h)
\\
\label{rapp:ev_mutk_s=0}&+\int_0^tdt'\int_{-1}^1dh'Q(t-t',h|h')e^{2\pi i(t-t')k\cdot v(\theta)}\mu_{t'}^k(\theta+\pi-2\arcsin(h'),0,h').
\end{align}

As in the case without dependance on the variable $x$, the long time behavior of $\mu_t^k(\theta,s,h)$ is fully characterized by the long time behavior of $\mu_t^k(\theta,s=0,h)$. That is the next step.
\newline
\subsubsection{Writing $\mu_t^k$ as a linear function of $\mu_0^k$.}
As in the case $k=(0,0)$ we studied in Section \ref{thetash}, we want to separate the dependance on $\mu_0^k$ in order to get more precise estimates. In particular, what we need to prove Theorem \ref{thm:mutk} is the following result.
\begin{proposition}\label{prop:phik}
Let $\mu_0\in L^1(\mathbb{T}^2\times\mathbb{T}^1_{2\pi}\times[0,+\infty)\times[-1,1])$ (or $\mu_0\in L^1(\mathbb{R}^2\times\mathbb{T}^1_{2\pi}\times[0,+\infty)\times[-1,1])$) and let $\{\mu_t^k\}$ be the Fourier coefficients of the mild solution of \eqref{rapp:ev} defined in \eqref{def:mutk}. For any $k\in\mathbb{Z}^2,k\neq(0,0)$ (respectively $k\in\mathbb{R}^2,k\neq(0,0)$) there exists a function $\varphi^k:\mathbb{T}^1_{2\pi}\times[0,+\infty)\times[-1,1]\times\mathbb{T}^1_{2\pi}\times[-1,1]\xrightarrow[]{L^{\infty}}\mathbb{C}$ such that $\mu_t^k(\theta,0,h)$ writes as a linear function of $\mu_0^k$ and an affine function of $\varphi^k$ as

{\scriptsize\begin{align*}
&\mu_t^k(\theta,0,h)
\\
&=e^{2\pi itk\cdot v(\theta)}\mu_0^k(\theta,t,h)
\\
&+\int_0^tdt'\int_{-1}^1dh'Q(t-t',h|h')e^{2\pi i(t-t')k\cdot v(\theta)}e^{2\pi it' k\cdot v(\theta+\pi-2\arcsin(h'))}\mu_0^k(\theta+\pi-2\arcsin(h'),t',h')
\\
&+\int_{\mathbb{T}^1_{2\pi}}d\theta'\int_0^tdt'\int_{-1}^1dh'\varphi^k(\theta,t-t',h|\theta',h')e^{2\pi it' k\cdot v(\theta')}\mu_0^k(\theta',t',h'),
\end{align*}}

and moreover there exists a constant $C$ depending only on $Q$ (and not on $k$) such that
{\footnotesize\[
|\varphi^k(\theta,t,h|\theta',h')|\leq\frac{C}{\min\{1,|k|^6\}(t+1)}\qquad\forall(\theta,t,h|\theta',h')\in\mathbb{T}^1_{2\pi}\times[0,+\infty)\times[-1,1]\times\mathbb{T}^1_{2\pi}\times[-1,1].
\]}
\end{proposition}
This Proposition allows us to prove Theorem \ref{thm:mutk} since none of the summands which $\mu_t^k(\theta,0,h)$ is composed of survives for large $t$. To prove this Proposition, we recall the Definition  \ref{defh''} of  $h''(\theta,h')$, that is
\[
h''(\theta,h'):=\sin\left(\frac{\theta+2\pi-2\arcsin(h')}{2}\right)\mathbbm{1}_{2\arcsin(h')+[-3\pi,-\pi)}(\theta),
\]
and this time, instead of using $f(\theta,t,h|\theta',h')$, we use $g^k$ as in Definition \ref{def:gk}, that is,

{\scriptsize\begin{align*}
 g^k(\theta,t,h|\theta',h'):=\sum_{\ell\in\mathbb{Z}}\frac{\partial h''}{\partial\theta}e^{2\pi itk\cdot v(\theta)}\int_0^tdt' Q(t-t',h|h'') Q(t',h''|h')e^{2\pi it' k\cdot[v(\theta'-\pi+2\arcsin(h'))-v(\theta)]},
 \end{align*}}
 
 with
 \[
 h''=h''(\theta-\theta'+2\ell\pi,h'),\textrm{ and }h''\textrm{ from Definition \ref{defh''}}.
 \]
Of course, because of the complex exponential, $g^k$ is no more a function of  $\theta-\theta'$, differently from $f$ of Definition \ref{defeffe} in the previous Section. Nevertheless, it is still periodic both with respect to $\theta$ and to $\theta'$. In Lemma \ref{lemma:proprgk} in Section \ref{app:funzionidiQ} we collect some properties of the functions $g^k$, $k\neq(0,0)$.

The first necessary step to prove Proposition \ref{prop:phik} is the following Lemma, that proves the existence of the function $\varphi^k$, but not its decaying properties in long time, that we will study immediately after.
\begin{lemma}\label{lemma:esphik}
For any $k\in\mathbb{R}^2$, $k\neq(0,0)$, there exists a unique function $\varphi^k\in L^{\infty}_{loc}(\mathbb{T}^1_{2\pi}\times[0,+\infty)\times[-1,1]\times\mathbb{T}^1_{2\pi}\times[-1,1])$ satisfying

{\tiny\begin{align}
\nonumber\varphi^k(\theta,t,h|\theta',h')&=g^k(\theta,t,h|\theta',h')
\\
\label{espr:phik}&+\int_0^tdt'\int_{-1}^1dh''Q(t-t',h|h'')e^{2\pi i(t-t')k\cdot v(\theta)}\varphi^k(\theta+\pi-2\arcsin(h''),t',h''|\theta',h'),
\end{align}}

and moreover for any $\mu_0\in L^1(\mathbb{T}^2\times\mathbb{T}^1_{2\pi}\times[0,+\infty)\times[-1,1])$ or $\mu_0\in L^1(\mathbb{R}^2\times\mathbb{T}^1_{2\pi}\times[0,+\infty)\times[-1,1])$, $\mu_t^k(\theta,0,h)$ writes as a function of $\varphi^k$ and $\mu_0$ as:

{\tiny\begin{align}
\nonumber&\mu_t^k(\theta,0,h)
\\
\nonumber&=e^{2\pi itk\cdot v(\theta)}\mu_0^k(\theta,t,h)
\\
\nonumber&+\int_0^tdt'\int_{-1}^1dh'Q(t-t',h|h')e^{2\pi i(t-t')k\cdot v(\theta)}e^{2\pi it' k\cdot v(\theta+\pi-2\arcsin(h'))}\mu_0^k(\theta+\pi-2\arcsin(h'),t',h')
\\
\label{rapp:nucleophik}&+\int_{\mathbb{T}^1_{2\pi}}d\theta'\int_0^tdt'\int_{-1}^1dh'\varphi^k(\theta,t-t',h|\theta',h')e^{2\pi it' k\cdot v(\theta')}\mu_0^k(\theta',t',h').
\end{align}}

\end{lemma}
\begin{proof} The proof of this Lemma is very similar to the proof of Lemma \ref{lemma:esphi} in the previous Section \ref{thetash}, but we write it entirely because of the presence of the complex exponential that changes some steps.

We are looking for a function $\varphi^k:\mathbb{T}^1_{2\pi}\times[0,+\infty)\times[-1,1]\times\mathbb{T}^1_{2\pi}\times[-1,1]\to\mathbb{C}$ such that any $\mu_t^k$ can be expressed through $\varphi^k$ as in equation \eqref{rapp:nucleophik}. Since $\mu_t^k$ satisfies equation \eqref{rapp:ev_mutk_s=0}, substituting the desired equation \eqref{rapp:nucleophik} in both sides of \eqref{rapp:ev_mutk_s=0}, we get the following condition on $\varphi^k$

{\scriptsize\begin{align}
\nonumber&e^{2\pi itk\cdot v(\theta)}\mu_0^k(\theta,t,h)
\\
\nonumber&+\int_0^tdt'\int_{-1}^1dh'Q(t-t',h|h')e^{2\pi i(t-t')k\cdot v(\theta)}e^{2\pi it' k\cdot v(\theta+\pi-2\arcsin(h'))}\mu_0^k(\theta+\pi-2\arcsin(h'),t',h')
\\
\label{ariadd1}&+\int_{\mathbb{T}^1_{2\pi}}d\theta'\int_0^tdt'\int_{-1}^1dh'\varphi^k(\theta,t-t',h|\theta',h')e^{2\pi it' k\cdot v(\theta')}\mu_0^k(\theta',t',h')
\\
\nonumber&=e^{2\pi itk\cdot v(\theta)}\mu_0^k(\theta,t,h)
\\
\nonumber&+\int_0^tdt'\int_{-1}^1dh'Q(t-t',h|h')e^{2\pi i(t-t')k\cdot v(\theta)}\biggl[e^{2\pi it' k\cdot v(\theta+\pi-2\arcsin(h'))}\mu_0^k(\theta+\pi-2\arcsin(h'),t',h')
\\
\label{ariadd2}&+\int_0^{t'}dt''\int_{-1}^1dh''Q(t'-t'',h'|h'')e^{2\pi i(t'-t'')k\cdot v(\eta')} e^{2\pi i t''k\cdot v(\eta'')}\mu_0^k(\eta'',t'',h'')
\\
\label{ariadd3}&+\int_{\mathbb{T}^1_{2\pi}}d\theta'\int_0^{t'}dt''\int_{-1}^1dh'' e^{2\pi it''k\cdot v(\theta')}\varphi^k(\eta',t'-t'',h'|\theta',h'')\mu_0^k(\theta',t'',h'')\biggr],
\end{align}}

where 

\begin{align}
\nonumber\eta'&=\eta'(\theta,h'):=\theta+\pi-2\arcsin(h'),
\\
\label{notacazzo}\eta''&=\eta''(\theta,h',h''):=\theta+2\pi-2\arcsin(h')-2\arcsin(h'').
\end{align}

In the previous equality, the first two summands of both sides get cancelled. Therefore the following step is to write better the equality
\[
\eqref{ariadd1}=\eqref{ariadd2}+\eqref{ariadd3}.
\]
We leave aside \eqref{ariadd1} that is fine as it is.

As for the term in \eqref{ariadd2} we use the following steps. First we change the integration order: before with respect to $(h',h'',t'')$ and then with respect to $t'$. Then we change the variables names: $t''\leftrightarrow t'$ and $h'\leftrightarrow h''$. Finally we change variables in such a way to have
\[
\eta''(\theta,h',h'')=\theta+2\pi-2\arcsin(h')-2\arcsin(h'')=\theta',
\]
that is 
\[
h''=h''(\theta-\theta',h')\textrm{ as in Definition \ref{defh''}.}
\]
Through this passages, shortening $h''=h''(\theta-\theta',h')$ and using again the notation in \eqref{notacazzo} with the new $h'$ and $h''$, we get
\[
\eqref{ariadd2}=\int_0^tdt'\int_{-1}^1dh'\int_{\theta+\pi-2\arcsin(h')}^{\theta+3\pi-2\arcsin(h')}d\theta'e^{2\pi it' k\cdot v(\theta')}\mu_0^k(\theta',t',h')g^k(\theta,t-t',h|\theta',h')
\]
by Definition \ref{def:gk} of $g^k$, and therefore
\[
\eqref{ariadd2}=\int_0^tdt'\int_{-1}^1dh'\int_{\mathbb{T}^1_{2\pi}}d\theta'e^{2\pi it' k\cdot v(\theta')}\mu_0^k(\theta',t',h')g^k(\theta,t-t',h|\theta',h').
\]
since the integrand is periodic.

As for the term in \eqref{ariadd3} we change the integration order: before with respect to $(\theta',t'',h'')$ and then with respect to $(t',h')$. Then we exchange the names of the variables $t''\leftrightarrow t'$ and $h'\leftrightarrow h''$. In this way we obtain

\begin{align*}
\eqref{ariadd3}=\int_{\mathbb{T}^1_{2\pi}}d\theta'\int_0^tdt'\int_{-1}^1dh'e^{2\pi it' k\cdot v(\theta')}\mu_0^k(\theta',t',h')\psi_k(\theta,t-t',h|\theta'h'),
\end{align*}

with
{\footnotesize\[
\psi_k(\theta,t-t',h|\theta'h'):=\int_0^{t-t'}dt''\int_{-1}^1dh''Q(t-t'-t'',h|h'')e^{2\pi i(t-t'-t'')k\cdot v(\theta)}\varphi^k(\eta''',t'',h''|\theta',h'),
\]}
and with
\[
\eta'''=\eta'''(\theta,h'):=\theta+\pi-2\arcsin(h'').
\]
Thus, if we compute 
\[
\eqref{ariadd1}=\eqref{ariadd2}+\eqref{ariadd3},
\]
and if we impose it to be valid for any $\mu_0^k$, we get

{\scriptsize\begin{align*}
\varphi^k(\theta,t,h|\theta',h')&=g^k(\theta,t,h|\theta',h')
\\
&+\int_0^tdt'\int_{-1}^1dh''Q(t-t',h|h'')e^{2\pi i(t-t')k\cdot v(\theta)}\varphi^k(\theta+\pi-2\arcsin(h''),t',h''|\theta',h'),
\end{align*}}

which is exactly \eqref{espr:phik}.

The existence and the uniqueness of $\varphi^k$ are ensured as follows: for fixed $\theta'\in\mathbb{R}$ and $h'$, $\varphi^k$ exists  in the space $L^{\infty}_{loc}(\mathbb{T}^1_{2\pi}\times[0,+\infty)\times[-1,1];\mathbb{C})$: this can be proven by using Lemma \ref{lemma:rhovarianteC} with $p=\infty$ and $\mu(\theta,t,h)=g^k(\theta,t,h|\theta',h')$. The estimate $\|\varphi^k\|_{L^{\infty}(\mathbb{T}^1_{2\pi}\times[0,T]\times[-1,1]\times\mathbb{T}^1_{2\pi}\times[-1,1])}$, that is, a local $L^{\infty}$ estimate not depending on $(h',\theta')$, is a consequence of the fact that $\varphi^k$ is obtained in Lemma \ref{lemma:rhovarianteC} through the Banach fixed-point Theorem, and therefore the quantity
\[
\|\varphi^k(\cdot,\cdot,\cdot|\theta',h')\|_{L^{\infty}(\mathbb{T}^1_{2\pi}\times[0,T]\times[-1,1])},
\]
is bounded by a linear function, with a possible dependance on $T$, of
\[
\|g^k(\cdot,\cdot,\cdot|\theta',h')\|_{L^{\infty}(\mathbb{T}^1_{2\pi}\times[0,T]\times[-1,1])}.
\]
By \eqref{decgk} of Lemma \ref{lemma:proprgk}, the previous quantity is in turn uniformly bounded with respect to $(\theta',h')$.

Let us focus on the periodicity. $\varphi^k$ is periodic with respect to $\theta$ because $g^k$ is too. Moreover, since $g^k$ is periodic with respect to $\theta'$, for fixed $\theta'\in\mathbb{R}$ and $h'\in[-1,1]$, $\varphi^k(\cdot,\cdot,\cdot|\theta',h')$ and $\varphi^k(\cdot,\cdot,\cdot|\theta'+2\pi,h')$ solve the same problem in \eqref{espr:phik}. Therefore they coincide, indeed Lemma \ref{lemma:rhovarianteC} ensures that the solution of such a problem is unique.
\end{proof}
To conclude the proof of Proposition \ref{prop:phik}, it remains to study the estimates of $\varphi^k$ for large $t$. The following Lemma collects all the necessary estimates to prove that $\varphi^k$ is bounded and decays at most as $\frac{1}{t}$.
\begin{lemma}\label{lemma:raccolta}
For any $k\in\mathbb{R}^2$, $k\neq(0,0)$, being $g^k$ introduced in Definition \ref{def:gk}, the function $\varphi^k$ defined in Lemma \ref{lemma:esphik} verifies the following identity

{\scriptsize\begin{align}
\nonumber\varphi^k(\theta,t,h|\theta',h')&=g^k(\theta,t,h|\theta',h')
\\
\nonumber&+\int_0^tdt'\int_{-1}^1dh''Q(t-t',h|h'')e^{2\pi i(t-t')k\cdot v(\theta)}g^k(\theta+\pi-2\arcsin(h''),t',h''|\theta',h')
\\
\label{phikreiter}&+\int_{\mathbb{T}^1_{2\pi}}d\theta''\int_0^tdt'\int_{-1}^1dh'''g^k(\theta,t-t',h|\theta'',h''')\varphi^k(\theta'',t',h'''|\theta',h').
\end{align}}

and moreover for $k\in\mathbb{R}^2$, $k\neq(0,0)$, $g^k$ is a contractive map, that is, it holds

\begin{align}
\label{gkcontrae}\sup_{\theta\in\mathbb{T}^1_{2\pi},h\in[-1,1]}\|g^k(\theta,\cdot,h|\cdot,\cdot)\|_{L^1}\leq1-C\min\{1,|k|^2\},
\end{align}

up to a constant $C\in(0,1)$ not depending on $k$.
\end{lemma}
\begin{proof} To prove property \eqref{phikreiter} we have to iterate over \eqref{espr:phik} twice, that is

{\scriptsize\begin{align}
\nonumber&\varphi^k(\theta,t,h|\theta',h')
\\
\nonumber&=g^k(\theta,t,h|\theta',h')+\int_0^tdt'\int_{-1}^1dh''Q(t-t',h|h'')e^{2\pi i(t-t')k\cdot v(\theta)}g^k(\theta+\pi-2\arcsin(h''),t',h''|\theta',h')
\\
\label{tizioacaso}&+\int_0^tdt'\int_{-1}^1dh''Q(t-t',h|h'')e^{2\pi i(t-t')k\cdot v(\theta)}\psi_k(\theta,t',h''|\theta',h'),
\end{align}}

with
{\small\[
\psi_k(\theta,t',h''|\theta',h'):=\int_0^{t'}dt''\int_{-1}^1dh'''Q(t'-t'',h''|h''') e^{2\pi i(t'-t'')k\cdot v(\eta'')}\varphi^k(\eta''',t'',h'''|\theta',h')
\]}
and with
\begin{align*}
\eta''&=\eta''(\theta,h''):=\theta+\pi-2\arcsin(h''),
\\
\eta'''&=\eta'''(\theta,h'''):=\theta+2\pi-2\arcsin(h'')-2\arcsin(h''').
\end{align*}
In the third summand in the right-hand side we exchange the integration order: before with respect to $(h'',h''',t'')$ and then with respect to $t'$. Then we change the names of the variables $t''$ and $t'$. Then we change variables from $h''$ to $\theta''$ in such a way to have
\[
\theta+2\pi-2\arcsin(h'')-2\arcsin(h''')=\theta'',
\]
that is 
\[
h''=h''(\theta-\theta'',h''')\textrm{ as in Definition \ref{defh''}.}
\]
By using these steps, the term in \eqref{tizioacaso} writes also as

{\small\begin{align*}
\eqref{tizioacaso}&=\int_0^tdt'\int_{-1}^1dh'''\int_{\theta+\pi-2\arcsin(h''')}^{\theta+3\pi-2\arcsin(h''')}d\theta''\varphi^k(\theta'',t',h'''|\theta',h')g^k(\theta,t-t',h|\theta'',h''')
\\
&=\int_{\mathbb{T}^1_{2\pi}}d\theta''\int_0^tdt'\int_{-1}^1dh'''\varphi^k(\theta'',t',h'''|\theta',h')g^k(\theta,t-t',h|\theta'',h'''),
\end{align*}}

where the last equality holds because $g^k$ and $\varphi^k$ are periodic with respect to $\theta''$. This proves property \eqref{phikreiter}.

Finally, \eqref{gkcontrae} is exactly \eqref{gkint<1} of Lemma \ref{lemma:proprgk} in Section \ref{app:funzionidiQ}.
\end{proof}
As in the case without the dependance on variable $x$, the properties \eqref{phikreiter} and \eqref{gkcontrae} of the previous Lemma \ref{lemma:esphik} can be combined to prove that $\varphi^k$ has the same properties as $g^k$. We can now prove Proposition \ref{prop:phik}.
\newline
\subsubsection{Proof of Proposition \ref{prop:phik}.}
\begin{proof} For $k\in\mathbb{R}^2,k\neq(0,0)$, the function $\varphi^k$ in the statement of Proposition \ref{prop:phik} is exactly the function $\varphi^k$ defined through Lemma \ref{lemma:esphik}. Therefore it also writes as in \eqref{phikreiter} of Lemma \ref{lemma:raccolta}.

Therefore we need to prove that there exists a constant $C$ depending only $Q$, and not on $k$, such that $\varphi^k$ defined in Lemma \ref{lemma:esphik} satisfies
{\footnotesize\[
|\varphi^k(\theta,t,h|\theta',h')|\leq\frac{C}{\min\{1,|k|^6\}(t+1)}\quad\forall(\theta,t,h|\theta',h')\in\mathbb{T}^1_{2\pi}\times[0,+\infty)\times[-1,1]\times\mathbb{T}^1_{2\pi}\times[-1,1].
\]}
The proof is exactly analogue to the proof of Lemma \ref{decphi}: the differences are in the fact that one has to substitute $f(\theta-\theta',t-t',h|h')$ with $g^k(\theta,t-t',h|\theta',h')$, therefore we only sketch it.

To prove that $\varphi^k$ is bounded, one can use \eqref{phikreiter} of Lemma \ref{lemma:raccolta}, that is the analogue of Lemma \ref{lemma:phimeglio} in the previous Section, to get

\begin{align}\label{phiklimitata}
\|\varphi^k\|_{L^{\infty}}\leq\frac{\|J^k\|_{L^{\infty}}}{1-d^k},
\end{align}

with

{\footnotesize\begin{align*}
J^k(\theta,t,h|\theta',h')&:=g^k(\theta,t,h|\theta',h')
\\
&+\int_0^tdt'\int_{-1}^1dh''Q(t-t',h|h'')e^{2\pi i(t-t')k\cdot v(\theta)}g^k(\theta+\pi-2\arcsin(h''),t',h''|\theta',h'),
\\
d^k&:=\sup_{(\theta,h)\in\mathbb{T}^1_{2\pi}\times[-1,1]}\|g^k(\theta,\cdot,h|\cdot,\cdot)\|_{L^1}.
\end{align*}}

Notice that for any $k$ by Definition \ref{def:gk} of $g^k$ and \ref{defeffe} of $f$ we have
\[
\|J^k\|_{L^{\infty}}\leq2\|g^k\|_{L^{\infty}}\leq2\|f\|_{L^{\infty}}\textrm{ thanks to \eqref{deceffe} of Lemma \ref{propreffe}},
\]
that is, $J^k$ is bounded uniformly with respect to $k$. Moreover, by \eqref{gkcontrae} of Lemma \ref{lemma:raccolta}, there exists a constant $C'\in(0,1)$ such that
\[
d^k\leq1-C'\min\{1,|k|^2\}.
\]
Thus by \eqref{phiklimitata} one obtains
\begin{align}\label{phiklimitata2}
\|\varphi^k\|_{L^{\infty}}\leq\frac{2\|f\|_{L^{\infty}}}{C'\min\{1,|k|^2\}}.
\end{align}

Moreover, to prove the statement one has also to prove that $\varphi^k$ decays at most as $\frac{1}{t}$. To this purpose, one can notice that arguing as in the proof of Lemma \ref{decphi} one gets
\begin{align}\label{tphiklimitata}
\|t\varphi^k\|_{L^{\infty}}\leq\frac{\|tJ^k\|_{L^{\infty}}+C\frac{\alpha^k}{\alpha^k-1}\|\varphi^k\|_{L^{\infty}}}{1-\alpha^kd^k},
\end{align}

with $J^k$ and $d^k$ defined as in the previous step, and
\[
\alpha^k:=\frac{1+d^k}{2d^k}\in\left(1,\frac{1}{d^k}\right).
\]
Therefore by Definitions \ref{def:gk} of $g^k$ and \ref{defeffe} of $f$, one has 

\begin{align*}
\|tJ^k\|_{L^{\infty}}&\leq2\|tg^k\|_{L^{\infty}}\leq2\|tf\|_{L^{\infty}}\textrm{ thanks to \eqref{deceffe} of Lemma \ref{propreffe}},
\\
\|\varphi^k\|_{L^{\infty}}&\leq\frac{\|f\|_{L^{\infty}}}{C'\min\{1,|k|^2\}}\textrm{ thanks to \eqref{phiklimitata2}},
\end{align*}

and
{\small\[
\frac{\alpha^k}{\alpha^k-1}=\frac{1+d^k}{1-d^k}\leq\frac{2}{1-d^k}\leq\frac{2}{C'\min\{1,|k|^2\}},\quad\frac{1}{1-\alpha^kd^k}=\frac{2}{1-d^k}\leq\frac{2}{C'\min\{1,|k|^2\}}.
\]}
Substituting these estimates into \eqref{tphiklimitata} one gets
\[
\|t\varphi^k\|_{L^{\infty}}\leq2\frac{2\|tf\|_{L^{\infty}}+C\frac{2}{C'\min\{1,|k|^2\}}\frac{2}{C'\min\{1,|k|^2\}}2\|f\|_{L^{\infty}}}{C'\min\{1,|k|^2\}}\leq \frac{C''}{\min\{1,|k|^6\}},
\]
that is the thesis.
\end{proof}
\subsubsection{\textbf{Proof of Theorem \ref{thm:mutk}.}} Now that we have proved Proposition \ref{prop:phik}, we can use it to prove Theorem \ref{thm:mutk}.
\begin{proof} Since the estimates are completely analogue to the estimates in Theorem \ref{thm:convergenza_v,s,h}, here we are only going to write the summands which $\mu_t^k(\theta,s,h)$ is made up of. Thanks to the evolution equation \eqref{rapp:ev_mutk} and to Proposition \ref{prop:phik}, we can write

\begin{align*}
&\mu_t^k(\theta,s,h)
\\
&=e^{2\pi itk\cdot v(\theta)}\mu_0^k(\theta,s+t,h)
\\
&+\int_0^tdt'\int_{-1}^1dh'Q(s+t-t',h|h')e^{2\pi i(t-t')k\cdot v(\theta)}e^{2\pi it' k\cdot v(\eta')}\mu_0^k(\eta',t',h')
\\
&+\int_0^tdt'\int_{-1}^1dh'Q(s+t-t',h|h')e^{2\pi i(t-t')k\cdot v(\theta)}\xi_1(\theta,t',h')
\\
&+\int_0^tdt'\int_{-1}^1dh'Q(s+t-t',h|h')e^{2\pi i(t-t')k\cdot v(\theta)}\xi_2(\theta,t',h'),
\end{align*}

with

{\footnotesize\begin{align*}
\xi_1(\theta,t',h')&:=\int_0^{t'}dt''\int_{-1}^1dh''Q(t'-t'',h'|h'') e^{2\pi i(t'-t'')k\cdot v(\eta')}e^{2\pi it''k\cdot v(\eta'')}\mu_0^k(\eta'',t'',h''),
\\
\xi_2(\theta,t',h')&:=\int_{\mathbb{T}^1_{2\pi}}d\theta'\int_0^{t'}dt''\int_{-1}^1dh'' \varphi^k(\eta',t'-t'',h'|\theta',h'')e^{2\pi it''k\cdot v(\theta')}\mu_0^k(\theta',t'',h''),
\end{align*}}

and with the notation

\begin{align*}
\eta'&=\eta'(\theta,h'):=\theta+\pi-2\arcsin(h'),
\\
\eta''&=\eta''(\theta,h',h''):=\theta+2\pi-2\arcsin(h')-2\arcsin(h'').
\end{align*}

Therefore we have four summands that can be bounded respectively as \eqref{addendo1}, \eqref{addendo2}, \eqref{addendo3} and \eqref{addendo6} in Theorem \ref{thm:convergenza_v,s,h}. The only difference is that when we estimate \eqref{addendo6}, we have to take into account that the upper bounds for $\varphi^k$ and $t\varphi^k$ ensured by Proposition \ref{prop:phik} depend on whether $k$ is close to $0$. Therefore, in the statement, we get $\frac{C}{\min\{1,|k|^6\}}$ instead of $C$.
\end{proof}
\subsection{Convergence of the joint probability density.}
Now that we can affirm that all the Fourier coefficients with $k\neq(0,0)$ are vanishing for large $t$, we can prove Theorems \ref{thm:conv_x,v,s,h} and \ref{thm:conv_R2}. We begin with the proof of the first one.
\subsubsection{\textbf{Proof of Theorem \ref{thm:conv_x,v,s,h}.}}
\begin{proof} We start by proving \eqref{thm3:st1}. It is easily checked that $\mu_0$ satisfying the hypothesis in the statement of the Theorem belongs to $L^1(\mathbb{T}^2\times\mathbb{T}^1_{2\pi}\times[0,+\infty)\times[-1,1])$, indeed

{\small\begin{align*}
\int_{\mathbb{T}^1_{2\pi}}d\theta\int_0^{\infty}ds\int_{-1}^1dh\|\mu_0(\cdot,\theta,s,h)\|_{L^p(\mathbb{T}^2)}&\geq\int_{\mathbb{T}^1_{2\pi}}d\theta\int_0^{\infty}ds\int_{-1}^1dh\|\mu_0(\cdot,\theta,s,h)\|_{L^1(\mathbb{T}^2)}
\\
&=\|\mu_0\|_{L^1(\mathbb{T}^2\times\mathbb{T}^1_{2\pi}\times[0,+\infty)\times[-1,1])}.
\end{align*}}

Then, if $\mu_0$ is chosen as in the hypothesis, let us fix $\varepsilon>0$ and choose $\tilde\mu_0\in L^{\infty}$ supported in $\mathbb{T}^2\times\mathbb{T}^1_{2\pi}\times[0,T]\times[-1,1]$, for $T>0$,  such that

{\scriptsize\begin{align}\label{approssimazione}
\|\mu_0-\tilde\mu_0\|_{L^p(\mathbb{T}^2\times\mathbb{T}^1_{2\pi}\times[0,+\infty)\times[-1,1])}+\int_{\mathbb{T}^1_{2\pi}}d\theta\int_0^{\infty}ds\int_{-1}^1dh\|\mu_0(\cdot,\theta,s,h)-\tilde\mu_0(\cdot,\theta,s,h)\|_{L^p(\mathbb{T}^2)}\leq\varepsilon.
\end{align}}

Before going on, we point out that such a choice is possible, indeed, if $\mu_0$ satisfies the hypothesis, one can preliminarily choose $T>0$ such that
\[
\|\mu_0\|_{L^p(\mathbb{T}^2\times\mathbb{T}^1_{2\pi}\times[T,+\infty)\times[-1,1])}+\int_{\mathbb{T}^1_{2\pi}}d\theta\int_T^{\infty}ds\int_{-1}^1dh\|\mu_0(\cdot,\theta,s,h)\|_{L^p(\mathbb{T}^2)}\leq\frac{\varepsilon}{2},
\]
and then, since for Jensen's inequality and dominated convergence one has

{\scriptsize\begin{align*}
&\|\mu_0\mathbbm{1}(|\mu_0|> M)\|_{L^p(\mathbb{T}^2\times\mathbb{T}^1_{2\pi}\times[0,T]\times[-1,1])}+\int_{\mathbb{T}^1_{2\pi}}d\theta\int_0^Tds\int_{-1}^1dh\|\mu_0(\cdot,\theta,s,h)\mathbbm{1}(|\mu_0|> M)\|_{L^p(\mathbb{T}^2)}
\\
&\leq\|\mu_0\mathbbm{1}(|\mu_0|> M)\|_{L^p(\mathbb{T}^2\times\mathbb{T}^1_{2\pi}\times[0,T]\times[-1,1])}
\\
&+(4\pi T)^{1-\frac{1}{p}}\|\mu_0\mathbbm{1}(|\mu_0|> M)\|_{L^p(\mathbb{T}^2\times\mathbb{T}^1_{2\pi}\times[0,T]\times[-1,1])}\xrightarrow[M\to+\infty]{}0,
\end{align*}}

one can fix $M>0$ such that
{\scriptsize\[
\|\mu_0\mathbbm{1}(|\mu_0|> M)\|_{L^p(\mathbb{T}^2\times\mathbb{T}^1_{2\pi}\times[0,T]\times[-1,1])}+\int_{\mathbb{T}^1_{2\pi}}d\theta\int_0^Tds\int_{-1}^1dh\|\mu_0(\cdot,\theta,s,h)\mathbbm{1}(|\mu_0|> M)\|_{L^p(\mathbb{T}^2)}\leq\frac{\varepsilon}{2},
\]}
and in this way, if $\tilde\mu_0:=\mu_0\mathbbm{1}(|\mu_0|\leq M,s\in[0,T])$, by triangle inequality one gets the approximation \eqref{approssimazione}.

Notice also that, by choosing $\tilde\mu_0$ as we said, the hypothesis that $\mu_0$ satisfies with $p$ are satisfied by $\tilde\mu_0$ with any $q\in[1,+\infty]$, indeed

{\tiny\begin{align}\label{ogniq}
\|\tilde\mu_0\|_{L^q(\mathbb{T}^2\times\mathbb{T}^1_{2\pi}\times[0,+\infty)\times[-1,1])}+\int_{\mathbb{T}^1_{2\pi}}d\theta\int_0^{\infty}ds\int_{-1}^1dh\|\tilde\mu_0(\cdot,\theta,s,h)\|_{L^q(\mathbb{T}^2)}\leq M(4\pi T)^{\frac{1}{q}}+M(4\pi T)<+\infty.
\end{align}}

Going back to the proof, if now $\tilde\mu_t$ is the time evolution of $\tilde\mu_0$, by triangle inequality we get

{\small\begin{align}
\label{primopezzo}\left\|\mu_t-\frac{\langle\mu_0\rangle}{2\pi}E\right\|_{L^p(\mathbb{T}^2\times\mathbb{T}^1_{2\pi}\times[0,+\infty)\times[-1,1])}&\leq\left\|\mu_t-\tilde\mu_t\right\|_{L^p(\mathbb{T}^2\times\mathbb{T}^1_{2\pi}\times[0,+\infty)\times[-1,1])}
\\
\label{secondopezzo}&+\left\|\tilde\mu_t-\frac{\langle\tilde\mu_0\rangle}{2\pi}E\right\|_{L^p(\mathbb{T}^2\times\mathbb{T}^1_{2\pi}\times[0,+\infty)\times[-1,1])}
\\
\label{terzopezzo}&+\left\|\frac{\langle\tilde\mu_0\rangle}{2\pi}E-\frac{\langle\mu_0\rangle}{2\pi}E\right\|_{L^p(\mathbb{T}^2\times\mathbb{T}^1_{2\pi}\times[0,+\infty)\times[-1,1])},
\end{align}}

therefore we are left with the estimate of each one of these summands.

\textbf{Step 1: }estimate of \eqref{primopezzo} and \eqref{terzopezzo}. For the first term, by \eqref{normaLpfinita} of Proposition \ref{prop:es/un}, thanks to the choice of $\tilde\mu_0$ in \eqref{approssimazione} we have

\begin{align*}
\eqref{primopezzo}&\leq C\left\|\mu_0-\tilde\mu_0\right\|_{L^p(\mathbb{T}^2\times\mathbb{T}^1_{2\pi}\times[0,+\infty)\times[-1,1])}
\\
&+C\int_{\mathbb{T}^1_{2\pi}}d\theta\int_0^tds\int_{-1}^1dh\|\mu_0(\cdot,\theta,s,h)-\tilde\mu_0(\cdot,\theta,s,h)\|_{L^p(\mathbb{T}^2)}
\\
&\leq C\left\|\mu_0-\tilde\mu_0\right\|_{L^p(\mathbb{T}^2\times\mathbb{T}^1_{2\pi}\times[0,+\infty)\times[-1,1])}
\\
&+C\int_{\mathbb{T}^1_{2\pi}}d\theta\int_0^{\infty}ds\int_{-1}^1dh\|\mu_0(\cdot,\theta,s,h)-\tilde\mu_0(\cdot,\theta,s,h)\|_{L^p(\mathbb{T}^2)}
\\
&\leq C\varepsilon.
\end{align*}

Moreover, for the third summand \eqref{terzopezzo}, since the flat torus $\mathbb{T}^2$ has finite measure, using again the property of approximation of $\tilde\mu_0$ in \eqref{approssimazione}, one has

\begin{align*}
\eqref{terzopezzo}&\leq|\langle\mu_0\rangle-\langle\tilde\mu_0\rangle|\left\|\frac{1}{2\pi}E\right\|_{L^p(\mathbb{T}^2\times\mathbb{T}^1_{2\pi}\times[0,+\infty)\times[-1,1])}
\\
&\leq|\langle\mu_0\rangle-\langle\tilde\mu_0\rangle|\leq\|\mu_0-\tilde\mu_0\|_{L^1(\mathbb{T}^2\times\mathbb{T}^1_{2\pi}\times[0,+\infty)\times[-1,1])}
\\
&=\int_{\mathbb{T}^1_{2\pi}}d\theta\int_0^{\infty}ds\int_{-1}^1dh\int_{\mathbb{T}^2}dx|\mu_0(x,\theta,s,h)-\tilde\mu_0(x,\theta,s,h)|
\\
&\leq\int_{\mathbb{T}^1_{2\pi}}d\theta\int_0^{\infty}ds\int_{-1}^1dh\|\mu_0(\cdot,\theta,s,h)-\tilde\mu_0(\cdot,\theta,s,h)\|_{L^p(\mathbb{T}^2)}
\\
&\leq\varepsilon.
\end{align*}

\textbf{Step 2:} estimate of \eqref{secondopezzo}.

\textbf{Step 2A: } we begin with the $L^1$ norm. Our purpose is to prove that 

\begin{align}\label{convergenzanormaL1}
\left\|\tilde\mu_t-\frac{\langle\tilde\mu_0\rangle}{2\pi}E\right\|_{L^1(\mathbb{T}^2\times\mathbb{T}^1_{2\pi}\times[0,+\infty)\times[-1,1])}\xrightarrow[t\to+\infty]{}0.
\end{align}

To this purpose, notice that since $\tilde\mu_0\in L^2$, one can choose $N_{\varepsilon}>0$ such that:
\[
\left\|\tilde\mu_0-\sum_{|k|\leq N_{\varepsilon}}\tilde\mu_0^ke_k\right\|_{L^2(\mathbb{T}^2\times\mathbb{T}^1_{2\pi}\times[0,+\infty)\times[-1,1])}\leq\frac{\varepsilon}{\sqrt{4\pi T}},\qquad e_k(x):=e^{-2\pi ik\cdot x}.
\]
A priori this approximation should be valid only in $L^2(\mathbb{T}^2)$ and depending on $(\theta,s,h)$, but since the convergence of the sum of the squared moduli of the Fourier coefficients to the $L^2$ norm of a function is monotone, the estimate holds also in $L^2(\mathbb{T}^2\times\mathbb{T}^1_{2\pi}\times[0,+\infty)\times[-1,1])$.

Moreover, if $\tilde\mu_0$ is supported in $\mathbb{T}^2\times\mathbb{T}^1_{2\pi}\times[0,T]\times[-1,1]$ (whose measure is $4\pi T$), then the same do its Fourier coefficients. Thus, since the $L^1$ norm of a mild solution is not increasing, one can estimate

\begin{align*}
\eqref{secondopezzo}&\leq\left\|\tilde\mu_t-\sum_{|k|\leq N_{\varepsilon}}\tilde\mu_t^ke_k\right\|_{L^1(\mathbb{T}^2\times\mathbb{T}^1_{2\pi}\times[0,+\infty)\times[-1,1])}
\\
&+\left\|\sum_{|k|\leq N_{\varepsilon}}\tilde\mu_t^ke_k-\frac{\langle\tilde\mu_0\rangle}{2\pi}E\right\|_{L^1(\mathbb{T}^2\times\mathbb{T}^1_{2\pi}\times[0,+\infty)\times[-1,1])}
\\
&\leq\left\|\tilde\mu_0-\sum_{|k|\leq N_{\varepsilon}}\tilde\mu_0^ke_k\right\|_{L^1(\mathbb{T}^2\times\mathbb{T}^1_{2\pi}\times[0,+\infty)\times[-1,1])}
\\
&+\sum_{0<|k|\leq N_{\varepsilon}}\left\|\tilde\mu_t^ke_k\right\|_{L^1(\mathbb{T}^2\times\mathbb{T}^1_{2\pi}\times[0,+\infty)\times[-1,1])}
\\
&+\left\|\tilde\mu_t^{(0,0)}-\frac{\langle\tilde\mu_0\rangle}{2\pi}E\right\|_{L^1(\mathbb{T}^2\times\mathbb{T}^1_{2\pi}\times[0,+\infty)\times[-1,1])}
\\
&\leq\sqrt{4\pi T}\left\|\tilde\mu_0-\sum_{|k|\leq N_{\varepsilon}}\tilde\mu_0^ke_k\right\|_{L^2(\mathbb{T}^2\times\mathbb{T}^1_{2\pi}\times[0,+\infty)\times[-1,1])}
\\
&+\sum_{0<|k|\leq N_{\varepsilon}}\left\|\tilde\mu_t^k\right\|_{L^1(\mathbb{T}^1_{2\pi}\times[0,+\infty)\times[-1,1])}
\\
&+\left\|\tilde\mu_t^{(0,0)}-\frac{\langle\tilde\mu_0\rangle}{2\pi}E\right\|_{L^1(\mathbb{T}^1_{2\pi}\times[0,+\infty)\times[-1,1])}
\\
&\leq\varepsilon+C\sum_{0\leq|k|\leq N_{\varepsilon}}\frac{\|\tilde\mu_t^k\|_{L^1(\mathbb{T}^1_{2\pi}\times[0,+\infty)\times[-1,1])}}{t+1}
\\
&+C\sum_{0\leq|k|\leq N_{\varepsilon}}\|\tilde\mu_t^k\|_{L^1(\mathbb{T}^1_{2\pi}\times[\frac{t}{4},+\infty)\times[-1,1])}
\\
&\leq2\varepsilon,
\end{align*}

if $t$ is large enough, that is, $t\geq t(N_{\varepsilon})$. In the last but one inequality we have used Theorems \ref{thm:convergenza_v,s,h} and \ref{thm:mutk} with $p=1$. Therefore the convergence in \eqref{convergenzanormaL1} for $p=1$ is proved.

\textbf{Step 2B: } we aim to use \eqref{convergenzanormaL1} to prove it for any finite $p$. To this purpose, we interpolate the $L^p$ norm by using $L^1$ and $L^{\infty}$, as follows.

{\small\begin{align*}
\eqref{secondopezzo}&=\left\|\tilde\mu_t-\frac{\langle\tilde\mu_0\rangle}{2\pi}E\right\|_{L^p(\mathbb{T}^2\times\mathbb{T}^1_{2\pi}\times[0,+\infty)\times[-1,1])}
\\
&\leq\left\|\tilde\mu_t-\frac{\langle\tilde\mu_0\rangle}{2\pi}E\right\|_{L^{\infty}(\mathbb{T}^2\times\mathbb{T}^1_{2\pi}\times[0,+\infty)\times[-1,1])}^{1-\frac{1}{p}}\left\|\tilde\mu_t-\frac{\langle\tilde\mu_0\rangle}{2\pi}E\right\|_{L^1(\mathbb{T}^2\times\mathbb{T}^1_{2\pi}\times[0,+\infty)\times[-1,1])}^{\frac{1}{p}},
\end{align*}}

and since the second factor is vanishing by the previous step \eqref{convergenzanormaL1}, it is sufficient to prove that the first one is bounded. But it is, indeed by using first triangle inequality and then \eqref{normaLpfinita} of Proposition \ref{prop:es/un}, one gets

\begin{align*}
&\left\|\tilde\mu_t-\frac{\langle\tilde\mu_0\rangle}{2\pi}E\right\|_{L^{\infty}(\mathbb{T}^2\times\mathbb{T}^1_{2\pi}\times[0,+\infty)\times[-1,1])}
\\
&\leq\left\|\tilde\mu_t\right\|_{L^{\infty}(\mathbb{T}^2\times\mathbb{T}^1_{2\pi}\times[0,+\infty)\times[-1,1])}+\left\|\frac{\langle\tilde\mu_0\rangle}{2\pi}E\right\|_{L^{\infty}(\mathbb{T}^2\times\mathbb{T}^1_{2\pi}\times[0,+\infty)\times[-1,1])}
\\
&\leq C\left\|\tilde\mu_0\right\|_{L^{\infty}(\mathbb{T}^2\times\mathbb{T}^1_{2\pi}\times[0,+\infty)\times[-1,1])}+C\int_{\mathbb{T}^1_{2\pi}}d\theta\int_0^tds\int_{-1}^1dh\|\tilde\mu_0(\cdot,\theta,s,h)\|_{L^{\infty}(\mathbb{T}^2)}
\\
&+\|\tilde\mu_0\|_{L^1(\mathbb{T}^2\times\mathbb{T}^1_{2\pi}\times[0,+\infty)\times[-1,1])}
\\
&\leq C\left\|\tilde\mu_0\right\|_{L^{\infty}(\mathbb{T}^2\times\mathbb{T}^1_{2\pi}\times[0,+\infty)\times[-1,1])}+C\int_{\mathbb{T}^1_{2\pi}}d\theta\int_0^{\infty}ds\int_{-1}^1dh\|\tilde\mu_0(\cdot,\theta,s,h)\|_{L^{\infty}(\mathbb{T}^2)}
\\
&+\|\tilde\mu_0\|_{L^1(\mathbb{T}^2\times\mathbb{T}^1_{2\pi}\times[0,+\infty)\times[-1,1])}
\\
&=:C(\tilde\mu_0)<+\infty\textrm{ by \eqref{ogniq}}.
\end{align*}

Thus, one has
{\tiny\[
\eqref{secondopezzo}\leq C(\tilde\mu_0)^{1-\frac{1}{p}}\left\|\tilde\mu_t-\frac{\langle\tilde\mu_0\rangle}{2\pi}E\right\|_{L^1(\mathbb{T}^2\times\mathbb{T}^1_{2\pi}\times[0,+\infty)\times[-1,1])}^{\frac{1}{p}}\xrightarrow[t\to+\infty]{}0\textrm{ thanks to \eqref{convergenzanormaL1}, that is, \textbf{Step 2A}}.
\]}
Summarizing, by using \textbf{Step 1}, \textbf{Step 2A} and \textbf{Step 2B}, one gets the statement \eqref{thm3:st1}.

The second statement \eqref{thm3:st2} can be proved as follows. First notice that if $\mu_0$ satisfies the hypothesis with $p=\infty$, then it satisfies them for every $p\in[1,+\infty]$. To prove this, we first notice that $\mu_0\in L^1(\mathbb{T}^2\times\mathbb{T}^1_{2\pi}\times[0,+\infty)\times[-1,1])$, indeed, as we said at the beginning of the proof

{\small\begin{align*}
\int_{\mathbb{T}^1_{2\pi}}d\theta\int_0^{\infty}ds\int_{-1}^1dh\|\mu_0(\cdot,\theta,s,h)\|_{L^{\infty}(\mathbb{T}^2)}&\geq\int_{\mathbb{T}^1_{2\pi}}d\theta\int_0^{\infty}ds\int_{-1}^1dh\int_{\mathbb{T}^2}dx|\mu_0(x,\theta,t,h)|
\\
&=\|\mu_0\|_{L^1(\mathbb{T}^2\times\mathbb{T}^1_{2\pi}\times[0,+\infty)\times[-1,1])},
\end{align*}}

and thus, by interpolation, one has
{\tiny\[
\|\mu_0\|_{L^p(\mathbb{T}^2\times\mathbb{T}^1_{2\pi}\times[0,+\infty)\times[-1,1])}\leq\|\mu_0\|_{L^1(\mathbb{T}^2\times\mathbb{T}^1_{2\pi}\times[0,+\infty)\times[-1,1])}^{\frac{1}{p}}\|\mu_0\|_{L^{\infty}(\mathbb{T}^2\times\mathbb{T}^1_{2\pi}\times[0,+\infty)\times[-1,1])}^{1-\frac{1}{p}}<+\infty,
\]}
therefore the $L^p$ norm of $\mu_0$ is finite. As for the second summand in the hypothesis, it holds
{\small\[
\int_{\mathbb{T}^1_{2\pi}}d\theta\int_0^{\infty}ds\int_{-1}^1dh\|\mu_0(\cdot,\theta,s,h)\|_{L^{\infty}(\mathbb{T}^2)}\geq\int_{\mathbb{T}^1_{2\pi}}d\theta\int_0^{\infty}ds\int_{-1}^1dh\|\mu_0(\cdot,\theta,s,h)\|_{L^p(\mathbb{T}^2)},
\]}
and therefore

\begin{align}\label{ognip}
\|\mu_0\|_{L^p(\mathbb{T}^2\times\mathbb{T}^1_{2\pi}\times[0,+\infty)\times[-1,1])}+\int_{\mathbb{T}^1_{2\pi}}d\theta\int_0^{\infty}ds\int_{-1}^1dh\|\mu_0(\cdot,\theta,s,h)\|_{L^p(\mathbb{T}^2)}<+\infty,
\end{align}

which is our claim.

To prove \eqref{thm3:st2}, we have to prove that for any $\eta\in L^1(\mathbb{T}^2\times\mathbb{T}^1_{2\pi}\times[0,+\infty)\times[-1,1])$, it holds

\begin{align*}
\int_{\mathbb{T}^2}dx\int_{\mathbb{T}^1_{2\pi}}d\theta\int_0^{\infty}ds\int_{-1}^1dh\eta(x,\theta,s,h)\left(\mu_t(x,\theta,s,h)-\frac{\langle\mu_0\rangle}{2\pi}E(s,h)\right)\xrightarrow[t\to+\infty]{}0.
\end{align*}

For the sake of brevity, we will shorten
\[
( A,B):=\int_{\mathbb{T}^2}dx\int_{\mathbb{T}^1_{2\pi}}d\theta\int_0^{\infty}ds\int_{-1}^1dhA(x,\theta,s,h)B(x,\theta,s,h),
\]
and, since here it is clear that we are integrating on the space $\mathbb{T}^2\times\mathbb{T}^1_{2\pi}\times[0,+\infty)\times[-1,1]$, we will also write
\[
\|A\|_{L^p}:=\|A\|_{L^p(\mathbb{T}^2\times\mathbb{T}^1_{2\pi}\times[0,+\infty)\times[-1,1])}.
\]
To prove the statement, fix $\tilde\eta\in L^1\cap L^2(\mathbb{T}^2\times\mathbb{T}^1_{2\pi}\times[0,+\infty)\times[-1,1])$ such that

\begin{align}\label{gapprox}
\|\eta-\tilde\eta\|_{L^1}\leq\frac{\varepsilon}{\|\mu_0\|_{L^{\infty}}+\int_{\mathbb{T}^1_{2\pi}}d\theta\int_0^{\infty}ds\int_{-1}^1dh\|\mu_0(\cdot,\theta,s,h)\|_{L^{\infty}(\mathbb{T}^2)}+\|\mu_0\|_{L^1}}.
\end{align}

We want to use \eqref{ognip} with $p=2$. 

With the notations above, thanks to triangle inequality we have

{\tiny\begin{align*}
\left|\left( \eta,\mu_t-\frac{\langle\mu_0\rangle}{2\pi}E\right)\right|&\leq\left|\left(\eta-\tilde\eta,\mu_t-\frac{\langle\mu_0\rangle}{2\pi}E\right)\right|+\left|\left(\tilde\eta,\mu_t-\frac{\langle\mu_0\rangle}{2\pi}E\right)\right|
\\
&\leq\|\eta-\tilde\eta\|_{L^1}\left\|\mu_t-\frac{\langle\mu_0\rangle}{2\pi}E\right\|_{L^{\infty}}+\|\tilde\eta\|_{L^2}\left\|\mu_t-\frac{\langle\mu_0\rangle}{2\pi}E\right\|_{L^2}
\\
&\leq C\|\eta-\tilde\eta\|_{L^1}\left(\|\mu_0\|_{L^{\infty}}+\int_{\mathbb{T}^1_{2\pi}}d\theta\int_0^tds\int_{-1}^1dh\|\mu_0(\cdot,\theta,s,h)\|_{L^{\infty}(\mathbb{T}^2)}+\left\|\frac{\langle\mu_0\rangle}{2\pi}E\right\|_{L^{\infty}}\right)
\\
&+\|\tilde\eta\|_{L^2}\left\|\mu_t-\frac{\langle\mu_0\rangle}{2\pi}E\right\|_{L^2}
\\
&\leq C\|\eta-\tilde\eta\|_{L^1}\left(\|\mu_0\|_{L^{\infty}}+\int_{\mathbb{T}^1_{2\pi}}d\theta\int_0^{\infty}ds\int_{-1}^1dh\|\mu_0(\cdot,\theta,s,h)\|_{L^{\infty}(\mathbb{T}^2)}+\|\mu_0\|_{L^1}\right)
\\
&+\|\tilde\eta\|_{L^2}\left\|\mu_t-\frac{\langle\mu_0\rangle}{2\pi}E\right\|_{L^2}
\\
&\leq C\varepsilon+\|\tilde\eta\|_{L^2}\left\|\mu_t-\frac{\langle\mu_0\rangle}{2\pi}E\right\|_{L^2},
\end{align*}}

where in the last inequality we have used \eqref{normaLpfinita} of Proposition \ref{prop:es/un} and the property of $\tilde\eta$ in \eqref{gapprox}. Finally, for the last term, one can apply property \eqref{thm3:st1}, that we have proved before, and notice that \eqref{ognip} with $p=2$ is exactly the hypothesis we need to infer that
\[
\|\tilde\eta\|_{L^2}\left\|\mu_t-\frac{\langle\mu_0\rangle}{2\pi}E\right\|_{L^2}\xrightarrow[t\to+\infty]{}0.
\]
To conclude, we further need to prove the third statement of the Theorem, that is, \eqref{thm3:st3}. To this purpose, notice that by \eqref{thm1:st1} of Theorem \ref{thm:convergenza_v,s,h} and \eqref{thm2:st1} of Theorem \ref{thm:mutk}, one has

{\small\begin{align*}
&\left\|\mu_t-\frac{\langle\mu_0\rangle}{2\pi}E\right\|_{L^2(\mathbb{T}^2\times\mathbb{T}^1_{2\pi}\times[0,+\infty)\times[-1,1])}^2
\\
&=\left\|\mu_t^{(0,0)}-\frac{\langle\mu_0\rangle}{2\pi}E\right\|_{L^2(\mathbb{T}^1_{2\pi}\times[0,+\infty)\times[-1,1])}^2+\sum_{k\in\mathbb{Z}^2,k\neq(0,0)}\left\|\mu_t^k\right\|_{L^2(\mathbb{T}^1_{2\pi}\times[0,+\infty)\times[-1,1])}^2
\\
&\leq C^2\sum_{k\in\mathbb{Z}^2}\frac{\left\|\mu_0^k\right\|_{L^2(\mathbb{T}^1_{2\pi}\times[0,+\infty)\times[-1,1])}^2+\left\|\mu_0^k\right\|_{L^1(\mathbb{T}^1_{2\pi}\times[0,+\infty)\times[-1,1])}^2}{(t+1)^2}
\\
&+C^2\sum_{k\in\mathbb{Z}^2}\left[\left\|\mu_0^k\right\|_{L^2(\mathbb{T}^1_{2\pi}\times[t/4,+\infty)\times[-1,1])}^2+\left\|\mu_0^k\right\|_{L^1(\mathbb{T}^1_{2\pi}\times[t/4,+\infty)\times[-1,1])}^2\right]
\\
&=\frac{C^2}{(t+1)^2}\|\mu_0\|_{L^2(\mathbb{T}^2\times\mathbb{T}^1_{2\pi}\times[0,+\infty)\times[-1,1])}^2+\frac{C^2}{(t+1)^2}\left\|\|\mu_0\|_{L^1(\mathbb{T}^1_{2\pi}\times[0,+\infty)\times[-1,1])}\right\|_{L^2(\mathbb{T}^2)}^2
\\
&+C^2\|\mu_0\|_{L^2(\mathbb{T}^2\times\mathbb{T}^1_{2\pi}\times[\frac{t}{4},+\infty)\times[-1,1])}^2+C^2\left\|\|\mu_0\|_{L^1(\mathbb{T}^1_{2\pi}\times[\frac{t}{4},+\infty)\times[-1,1])}\right\|_{L^2(\mathbb{T}^2)}^2,
\end{align*}}

and taking the square root of both summands one gets

\begin{align*}
&\left\|\mu_t-\frac{\langle\mu_0\rangle}{2\pi}E\right\|_{L^2(\mathbb{T}^2\times\mathbb{T}^1_{2\pi}\times[0,+\infty)\times[-1,1])}
\\
&\leq\frac{C}{t+1}\|\mu_0\|_{L^2(\mathbb{T}^2\times\mathbb{T}^1_{2\pi}\times[0,+\infty)\times[-1,1])}+\frac{C}{t+1}\left\|\|\mu_0\|_{L^1(\mathbb{T}^1_{2\pi}\times[0,+\infty)\times[-1,1])}\right\|_{L^2(\mathbb{T}^2)}
\\
&+C\left[\|\mu_0\|_{L^2(\mathbb{T}^2\times\mathbb{T}^1_{2\pi}\times[\frac{t}{4},+\infty)\times[-1,1])}+\left\|\|\mu_0\|_{L^1(\mathbb{T}^1_{2\pi}\times[\frac{t}{4},+\infty)\times[-1,1])}\right\|_{L^2(\mathbb{T}^2)}\right].
\end{align*}

The inequality in \eqref{thm3:st3} follows from applying again Minkowski's inequality, which allows to exchange the order of the $L^1$ and $L^2$ norms. 

Finally, statement \eqref{thm3:st4} follows the same way as the previous one, the only difference is that one has to use \eqref{thm1:st2}, instead of \eqref{thm1:st1}, from Theorem \ref{thm:convergenza_v,s,h} and \eqref{thm2:st2}, instead of \eqref{thm2:st1}, from Theorem \ref{thm:mutk}.\end{proof}

\subsubsection{\textbf{Proof of Theorem \ref{thm:conv_R2}.}}
\begin{proof} Let $\eta\in{\mathcal S}(\mathbb{R}^2)$, where ${\mathcal S}(\mathbb{R}^2)$ is the Schwartz space of $\mathbb{R}^2$. Since the Fourier transform preserves the scalar product, we have
\[
\int_{\mathbb{R}^2}dx\eta(x)\overline{\mu_t(x,\theta,s,h)}=\frac{1}{(2\pi)^2}\int_{\mathbb{R}^2}dk\hat\eta(k)\mu_t^{\frac{k}{2\pi}}(\theta,s,h),\quad\hat \eta(k):=\int_{\mathbb{R}^2}dye^{-ik\cdot y}\eta(y).
\]
Moreover
{\scriptsize\[
\left\|\int_{\mathbb{R}^2}dk\hat \eta(k)\mu_t^{\frac{k}{2\pi}}(\cdot,\cdot,\cdot)\right\|_{L^1(\mathbb{T}^1_{2\pi}\times[0,+\infty)\times[-1,1])}\leq\int_{\mathbb{R}^2}dk|\hat \eta(k)|\underbrace{\int_{\mathbb{T}^1_{2\pi}}d\theta\int_0^{\infty}ds\int_{-1}^1dh|\mu_t^{\frac{k}{2\pi}}(\theta,s,h)|}_{\to0\textrm{ as }t\to+\infty\textrm{ for any }k\neq(0,0)\textrm{ by Theorem \ref{thm:mutk}}},
\]}
with

{\scriptsize\begin{align*}
|\hat \eta(k)|\int_{\mathbb{T}^1_{2\pi}}d\theta\int_0^{\infty}ds\int_{-1}^1dh|\mu_t^{\frac{k}{2\pi}}(\theta,s,h)|&=|\hat \eta(k)|\int_{\mathbb{T}^1_{2\pi}}d\theta\int_0^{\infty}ds\int_{-1}^1dh\left|\int_{\mathbb{R}^2}dxe^{2\pi i\frac{k}{2\pi}\cdot x}\mu_t(x,\theta,s,h)\right|
\\
&\leq|\hat \eta(k)|\|\mu_t\|_{L^1}\leq|\hat \eta(k)|\|\mu_0\|_{L^1},
\end{align*}}

where in the last inequality we used \eqref{distanzadecr} of Proposition \ref{prop:es/un}. Thus, for dominated convergence of the term above, we get Theorem \ref{thm:conv_R2} since $\hat \eta\in L^1(\mathbb{R}^2)$ if $\eta\in{\mathcal S}(\mathbb{R}^2)$.
\end{proof}

\subsubsection{Acknowledgments} I would like to express my gratitude to Emanuele Caglioti and Sergio Simonella for the useful discussions on the topic and the nice suggestions they gave me during the preparation of this work.
\newline
\subsubsection{Conflict of interest} The author has no competing interests to declare that are relevant to the content of this article.

\end{document}